\documentclass[12pt]{article}
\usepackage[margin=1.25in]{geometry}
\usepackage{documentstyle}
\RequirePackage{hypernat}
\RequirePackage[colorlinks]{hyperref}
\hypersetup{
	colorlinks = true,
	citecolor = blue,
	linkcolor = [rgb]{.5,.05,.05}
}

\usepackage{mathabx}

\usepackage{tikz}
\usetikzlibrary{decorations.pathreplacing,angles,quotes}

\title{Robust inference for the treatment effect variance in experiments using machine learning \thanks{\scriptsize Email: \href{mailto:asb712@nyu.edu}{alejandro.sanchez.becerra@emory.edu}. I would like to thank Abhishek Ananth, Xu Cheng, Juan Estrada, Wayne Gao, Guido Imbens, Toru Kitagawa, Francesca Molinari, Mikkel Plagborg-M{\o}ller, Ulrich M\"{u}ller, Megan Reed, J\"{o}rg Stoye, Mark Watson, and Kaspar Wuthrich, for helpful comments and suggestions. I would also like to thank seminar participants at the New York Econometrics Camp 2022, NYU, the Canadian Economic Association 2022, the International Association for Applied Econometrics Conference 2022, Syracuse, LACEA/LAMES 2022, Facebook/Meta Core Data Science Group, the CFE-CM Statistics Conference 2022, the European Winter Meeting of the Econometric Society 2022, and the Duke Microeconometrics Conference for the Class of 2020 and 2021.}}

\author{
Alejandro S\'{a}nchez-Becerra \vspace{-1em} \\ Emory QTM}

\date{First Version: February 12, 2022 \\ This version: \today \vspace{-1.5em}} 

\begin{document}

\clearpage
\maketitle
\thispagestyle{empty}

\begin{abstract}

  \singlespacing    
       

Experimenters often collect baseline data to study heterogeneity. I propose the first valid confidence intervals for the VCATE, the treatment effect variance explained by observables. Conventional approaches yield incorrect coverage when the VCATE is zero. As a result, practitioners could be prone to detect heterogeneity even when none exists. The reason why coverage worsens at the boundary is that all efficient estimators have a locally-degenerate influence function and may not be asymptotically normal. I solve the problem for a broad class of multistep estimators with a predictive first stage. My confidence intervals account for higher-order terms in the limiting distribution and are fast to compute. I also find new connections between the VCATE and the problem of deciding whom to treat. The gains of targeting treatment are (sharply) bounded by half the square root of the VCATE. Finally, I document excellent performance in simulation and reanalyze an experiment from Malawi.
	
	\medskip

\end{abstract}


\onehalfspacing
\selectfont

\newpage


\section{Introduction}

In recent years, there has been a rapid expansion of experiments to evaluate public policy programs and corporate initiatives. There is also more evidence that the effectiveness of a program can vary across individuals. For instance,  \cite{dizon2019parents} studies a population of low-income parents in Malawi with large misperceptions about their children's school performance. She finds that a simple intervention can bridge these gaps and that the Conditional Average Treatment Effect (CATE) varies by children's initial scores. In practice, even though researchers collect baseline surveys with many characteristics, the CATE is typically estimated via regressions with one or two interactions, thus underutilizing the full set of variables. The promise of leveraging the vast and readily available baseline data has sparked more applications of supervised machine learning  \citep{crepon2021cream,davis2020rethinking,deryugina2019mortality}. Methods such as LASSO, neural networks, random forests, or boosting are data-driven and allow for more variables and flexibility.

In this paper, I focus on the unconditional variance of the CATE, the VCATE, which measures the dispersion of treatment effects predicted by a set of baseline characteristics. The VCATE has a clear interpretation, even if the CATE is nonlinear or depends on many characteristics.  \cite{chernozhukov2020generic} and \cite{ding2019decomposing} separately propose estimators with a misspecification robust interpretation, whereas  \cite{levy2021fundamental} propose an efficient estimator. Despite these recent advances, there are currently no valid confidence intervals for the VCATE. In fact, \cite{levy2021fundamental} study the performance of confidence intervals based on the efficient influence function, i.e., the conventional way. They find that coverage degrades near the boundary, reaching a low of 32\% in simulations. They speculate that poor coverage is due to the degeneracy of the efficient influence function when the VCATE is zero. As a result, conventional guarantees for $\sqrt{n}-$asymptotic normality do not apply to this part of the parameter space. Moreover, a VCATE close to zero is economically meaningful because it could reflect null effects, low effect heterogeneity, or irrelevant covariates. For such situations, which are common in practice, conventional approaches to inference could be misleading. 

This paper provides fresh insights regarding the VCATE and proposes a solution for inference in experiments. I propose (a) novel ways to interpret the VCATE for decision-making by deriving sharp bounds on the population gains of personalized treatment assignment, (b) novel estimators of the VCATE that are both misspecification-robust and efficient, combining the best features of previous approaches, and (c)  novel confidence intervals that are shape-adaptive and fast to compute.

I break down why conventional confidence intervals have incorrect size. I show that the boundary inference problem can manifest even when the CATE is linear and univariate. I solve the problem for the linear case by proposing adaptive confidence intervals that meet the high-level conditions outlined in \cite{andrews2020generic}. I then show that these can be readily extended to the class of nonlinear models in \cite{chernozhukov2020generic} that combine regression adjustments and a machine learning first stage. In addition to providing new confidence intervals, my paper has novel implications for estimation by showing that only a subset of these multi-step estimators is efficient. Adaptive inference and regression adjustments work well for various predictive models under weak assumptions.\footnote{While I specialize my results to the VCATE, my inference approach relies on more general principles: I use knowledge of the limiting distribution function, conditional on the cross-fitted estimates. Correct coverage follows from verifying high-level assumptions that can be satisfied by a wide array of machine learning method used in the first step. In principle, my approach could extend to other non-standard inference problems.} In the fully nonparametric case, I use a conservative procedure with valid coverage over multiple sample splits.\footnote{This adjustment is in the spirit of \citep{chernozhukov2020generic}, who propose robust t-tests assuming a conditionally normal distribution. Their results are not directly applicable here due to the boundary inference problem. However, I apply the principles behind their ``median-parameter'' confidence intervals to my adaptive intervals.} I derive the local power curve for the associated tests of homogeneity and their relationship to the tests in \cite{crump2008nonparametric} and \cite{ding2019decomposing}. I also propose confidence intervals for settings with cluster dependence.

I document excellent root mean square error (RMSE) performance and coverage in simulations using LASSO, even in high dimensions. I benchmark my multi-step approach against a two-step debiased machine learning estimator. As predicted by theory, all approaches are asymptotically normal, efficient, and have good coverage in highly heterogeneous designs. However, when the VCATE is zero or close to zero, coverage of two-step alternatives can be as low as 45\%. By contrast, my adaptive intervals produce coverage at the intended 95\% level and better RMSE at all regions of the parameter space. I study the robustness of the multi-step approach in both the theory and simulations. I consider situations where the predictive component is misspecified or slow to converge. I also discuss issues related to uniform vs. pointwise coverage.


I apply my approach to data from \cite{dizon2019parents}, an information experiment with low-income parents in Malawi who had at least two school-age children. The intervention redesigned the way in which parents received information about their children's school performance. The endline survey measured parental beliefs about student grades and asked parents to allocate tickets to a scholarship between their children. \cite{dizon2019parents} presented graphs with a non-parametric CATE by baseline test scores, which had an approximately linear shape, and separately tested for significance using a regression with an interaction. I use LASSO to compute the VCATE for the two outcomes (parental beliefs and lottery allocations) by different characteristics of students, parents, and households.  To make the results interpretable, I focus on the standard deviation of the CATE, i.e., $\sqrt{VCATE}$, and normalize it by the standard deviation of the outcome in the control group.

My approach allows us to quantify the magnitude of effect heterogeneity. I find that the treatment effect heterogeneity explained by test scores is equivalent to 40\% of the standard deviation (SD) of the beliefs of the control group, and 16\% of the SD of the control group lottery allocation. I also find that the effect heterogeneity collectively explained by other student variables (grade, age, gender, attendance, and educational expenditures) is comparable to 11\% of the SD of beliefs in the control group. The VCATE of beliefs by student variables is significant at the 5\% level, but the VCATE of lottery outcomes by student variables is not. The combined VCATE associated with student scores and 12 other key characteristics has a similar value to the VCATE with only scores. Despite being conservative, the intervals for the VCATE are short in length in this empirical example. Using my new welfare bounds $(-|ATE| + \sqrt{VCATE + ATE^2})/2$, I predict that targeted interventions using the baseline covariates have a maximum added benefit of 7.9\% SD and 7.4\% SD (standard deviations of the outcome for the control group) on beliefs and lottery allocations, respectively.

Researchers should focus on the VCATE because it is a model-free quantity with good properties: it is well-defined even if the CATE is continuous or discrete, and it weakly increases when researchers add more covariates to their analysis. Researchers can test for homogeneity by evaluating whether confidence intervals for the VCATE include zero. In addition to testing, by quantifying the VCATE, researchers can compare the magnitude of heterogeneity relative to a benchmark, such as the variance at baseline, the VCATE for different covariates, or experiments in other sites.

\subsection{Contribution}

My first main contribution is to show that the VCATE provides a bound for the welfare gains of policy targeting. A policymaker might decide to use the information from the CATE to design personalized treatment recommendations  \citep{manski2004statistical,athey2021policy,kitagawa2018should,mbakop2021model}. One can measure utilitarian welfare by computing the expected outcome under different policies. Such policies can be further constrained to a class that respects budget limits, incentive compatibility, or fairness considerations \citep{viviano2020fair,sun2021empirical}. I show that the difference in mean outcomes between a targeted policy and a non-targeted policy using only the average treatment effect (ATE), is bounded by $\sqrt{VCATE}/2$. For instance, under homogeneity ($VCATE = 0$), there are no gains from targeting. I show that this bound holds in the population regardless of the choice of policy class and the underlying distribution. Furthermore, the bound is \textit{sharp} in the sense that it holds exactly for at least one policy and distribution.

The proposed bound on utilitarian welfare communicates information to practitioners about whether a targeting exercise is even worth pursuing, without needing to solve the targeting problem itself. The VCATE can be a supplemental quantity reported in regression analyses, or a benchmark for analysts choosing the optimal policy. If the VCATE is very low, practitioners may consider expanding the set of covariates in the analysis. To derive the bound, I use a constructive approach to solve the most adversarial distribution. I also prove a more general bound $(-|ATE| + \sqrt{VCATE + ATE^2})/2$, and show that the distribution that leads to a maximum welfare gain is one where the CATE has binary support and mean zero. The gains from targeting easily diminish if the value of the ATE is relatively higher than the VCATE.

My second contribution is related to efficient estimation and robust inference. New theory is required here because of a unique feature of the VCATE: the efficient influence function is degenerate when the CATE is homogeneous \citep{levy2021fundamental}. Classical results by \cite{newey1990semiparametric} show that any regular, efficient estimator can be decomposed as $\frac{1}{n}\sum_{i=1}^n \varphi_i + R_n$, where $n$ is the sample size, $\{\varphi_i\}_{i=1}^{n}$ are a set of i.i.d.  mean-zero influence functions, and $R_n$ is a residual with higher-order terms that are $o_p(n^{-1/2})$.\footnote{Many standard estimators can achieve this property, e.g., the ``debiased machine learning'' estimator \citep{chernozhukov2018double} or the targeted maximum likelihood estimator in \cite{levy2021fundamental}.} Conventional approaches assume that $\mathbb{V}(\varphi) > 0$ and in this case, the estimation error converges at $\sqrt{n}$ to $\mathcal{N}(0,\mathbb{V}(\varphi_i))$ by the CLT. However, when the VCATE is zero, $\mathbb{V}(\varphi_i) = 0$ as well. Hence, the limiting distribution is dominated by the higher terms in $R_n$, which may not be asymptotically normal. Therefore, while $\sqrt{n}-$estimation is still possible, t-tests that plug in an estimate of $\mathbb{V}(\varphi)$ may have incorrect coverage. By contrast, other common quantities such as the average treatment effect (ATE) or the local average treatment effect (LATE) do not have this problem because they satisfy $\mathbb{V}(\varphi_i) > 0$ uniformly  \citep{chernozhukov2018double}.

I start by analyzing a simple two step estimator, assuming that the CATE is linear in covariates and can be estimated from a regression. I show that the limiting distribution of the VCATE estimator can be written as a linear combination of a Chi-square that converges at $n$-rate and a normal distribution that converges at $\sqrt{n}-$rate. The weights are determined by the value of the VCATE, which means that the shape of the distribution changes depending on the region of the parameter space. At the boundary, it behaves like a rescaled chi-square, is $O_p(n^{-1})$, and confidence intervals with normal critical values will have incorrect coverage. For values of the VCATE bounded away from zero, the distribution is asymptotically normal as in the classical results.

In the linear case, I construct adaptive confidence intervals that account for the higher terms of the distribution. I apply the framework of \cite{andrews2020generic} to show that this produces uniform, exact coverage when the linear model is correctly specified.\footnote{This type of strategy has proven effective to deal with other non-standard problems where the shape of the limiting distribution depends on an unknown parameter, such as the AR coefficient in a time series, the effect parameter under weak instruments, or the quasi-likelihood ratio test for nonlinear regression \citep{andrews2020generic}.} The intervals are fast to compute because the expressions are all analytic. When there is a single covariate, I also show that a homogeneity test that evaluates whether zero is contained in the confidence intervals is algebraically identical to (i) a test of whether the interaction in the regression model is equal to zero, and (ii) the single-covariate homogeneity test of \cite{crump2008nonparametric}. The test is also asymptotically equivalent to \cite{ding2019decomposing}. However, these other tests only apply to series estimators and are not nested with mine in the multivariate and non-parametric cases.

I extend my results to the class of nonlinear models proposed by \cite{chernozhukov2020generic}. \cite{chernozhukov2020generic} showed that in experiments with known assignment probabilities, their models produce a meaningful pseudo-VCATE even if the functional form is misspecified. The pseudo-VCATE is non-negative, weakly lower than the VCATE, and converges to the true value under mild conditions on the estimated CATE.\footnote{This monotonicity property means that in experiments the pseudo-VCATE will not falsely detect heterogeneity, even if the machine learning stage is misspecified.} \cite{chernozhukov2020generic} argue that the pseudo-VCATE might be of independent interest as a measure of model fit.\footnote{ \cite{ding2019decomposing} also define a similar pseudo-VCATE based on randomization inference.} They describe a three-step estimator with a machine learning/prediction first-stage, a regression second-stage, and a sample-variance third stage. Related multi-step estimators have also been considered in other work \citep{guo2021machine}.

To the best of my knowledge, there are no existing asymptotic results for the multi-step VCATE estimator proposed in \cite{chernozhukov2020generic}. I fill in that gap by proving two sets of results. First, I show that all the estimators in their class converge to the true VCATE at least at $\sqrt{n}$-rate, are $o_p(n^{-1/2})$ at the boundary (as in the simple linear model), and have the convenient property that they are always non-negative.  This builds on the asymptotic expansion for the linear case I introduced above. Second, I prove that only a subset of the \cite{chernozhukov2020generic} estimators are efficient, i.e. converge at $\sqrt{n}-$ to an average of i.i.d. efficient influence functions. The key ingredient is to prove a novel finite-sample equivalence result. I find that the first order conditions of the regression step and the bias-correction component of the VCATE influence function are in fact identical, given a particular decomposition of the nuisance functions. The asymptotic results follow from fairly standard assumptions on convergence rates \citep{chernozhukov2018double,belloni2017program}. To get the limiting distribution, the only meaningful extra assumption is that the estimated CATE has bounded kurtosis (thin tails).

I show that extending the adaptive confidence intervals (CIs) to multi-step estimators is straightforward. The procedure randomly splits the data into subsets or \textit{folds} and estimates the nuisance functions and the VCATE on different folds. To compute the confidence intervals for a particular fold, the researcher can treat the second-stage regression as if the variables were given, and then construct the CIs as in the simple case. I construct median confidence intervals (CIs) to aggregate information across multiple folds. I show that the single fold procedure produces uniform, exact coverage for the pseudo-VCATE and point-wise, exact coverage for the VCATE for all points in the parameter space, at a nominal level $1-\alpha$. The multifold CIs have pointwise conservative coverage. 

Furthermore, the probability that the true VCATE is below the confidence interval bounds is uniformly bounded by $\alpha$ in large samples. This result applies to the single and multifold CIs and does not require that the first-stage estimates to converge. Instead it relies on the fact that in experiments the pseudo-VCATE is weakly lower than the true VCATE. Tests for homogeneity (whether zero is contained in the CI) belong to this broader class of tests. Having uniform size control for this class of one-sided tests means that my tests of homogeneity are robust.

This paper is also related to a growing literature on debiased-machine learning \citep{chernozhukov2016locally,chernozhukov2022automatic,belloni2014inference,belloni2017program,chernozhukov2018double}, semiparametric efficiency \citep{newey1990semiparametric}, uniform inference for non-standard problems \citep{andrews2020generic}, and tests of treatment effect homogeneity \citep{ding2019decomposing,crump2008nonparametric,heckman1997making,bitler2017can}. My approach combines results from these literatures by addressing a boundary inference problem with a machine learning stage, and applying techniques of uniform inference. A related literature also focuses on confidence intervals around point-predictions of the CATE \citep{athey2019generalized,semenova2021debiased}, rather than overall measures of dispersion.  

Section \ref{overviewframework} provides key definitions, introduces the welfare bound, and presents a version of the adaptive confidence intervals for the univariate regression case. Section \ref{adaptiveinference} frames the inference problem in a more general setting, and extends the adaptive confidence intervals for VCATE estimation with a machine learning first stage. Section \ref{largesampletheory} presents the large sample theory. Section \ref{simulations} introduces the simulations. Section \ref{sec:empirical_example} applies my approach to an empirical example from Malawi. Section \ref{conclusion} concludes.


\section{Overview of framework}
\label{overviewframework}
Consider a program evaluation setting in which an individual is assigned to either a treatment $(D=1)$ or a control group $(D= 0)$. The outcome of interest $Y$ depends on the treatment status. I denote the potential outcome under treatment and control status as $Y_1$ and $Y_0$, respectively, and the treatment effect as $Y_1 - Y_0$. The conditional average treatment effect (CATE) given covariates $X$ is defined as
$$ \tau(x) := \mathbb{E}[Y_1 - Y_0 \mid X = x],$$
and the average treatment effect (ATE) is defined as $\tau_{av} := \mathbb{E}[Y_1 - Y_0]$. This paper proposes an estimator of the \textit{variance of the CATE} (VCATE) defined as
$$ V_{\tau} := \mathbb{V}(\tau(X)).$$
The variance $V_\tau$ measures the dispersion of treatment effects that can be attributed to observable characteristics $X$. The value of $V_{\tau}$ depends on the choice of covariates. To understand how different covariates might impact the VCATE, let $V_{\tau}' = \mathbb{V}(\mathbb{E}[Y_1 - Y_0 \mid X'])$ be the VCATE for a different set of covariates $X'$. 
\begin{lem} If $X$ is $X'$-measurable, then $V_{\tau} \le V_{\tau}' \le \mathbb{V}(Y_1-Y_0)$.
	\label{lem:monotonicity}	
\end{lem}
Lemma \ref{lem:monotonicity} shows that the VCATE has the following monotonicity property: if the researcher adds more covariates to the analysis, or breaks down an existing covariate into more categories, then the VCATE will be weakly larger.

The propensity score, $p(x)$, is defined as follows
\begin{equation}
\label{eq:pscore}
p(x):= \mathbb{P}(D = 1 \mid X = x).
\end{equation}
I restrict attention to experimental settings where $p(x)$ is known. The CATE can be identified under further assumptions.
\begin{assump}\text{ }(i) Stable unit treatment value assumption (SUTVA), $Y = Y_1D + (1-D)Y_0$ (ii) Strong overlap, there is a constant $\delta \in (0,1/2)$ such that $\mathbb{P}(\delta < p(X) < 1-\delta) = 1$, (iii) Selection on observables, $Y_1,Y_0 \ \indep \ D  \mid X$.	
	\label{assump:unconfound_sutva}
\end{assump}
Assumption \ref{assump:unconfound_sutva}.(i) formalizes the idea that the researcher can only observe either $Y_1$ or $Y_0$, but not both, for any particular individual. Assumption (ii) holds in randomized controlled trials with treatment probabilities bounded away from $\{0,1\}$. Assumption (iii) states that an individual's treatment probability depends on $X$ but not their potential outcomes. Let $\mu_d(x)$ be the conditional mean of $Y$ given $X$ and a fixed value of $d \in \{0,1\}$,
\begin{equation}
\label{eq:conditionalmean}
\mu_d(x) := \mathbb{E}[Y \mid D = d, X = x].
\end{equation}
Under Assumption \ref{assump:unconfound_sutva}, $\mathbb{E}[Y_d \mid X =x ] = \mu_d(x)$, and hence $\tau(x) = \mu_1(x) - \mu_0(x)$. This means that the VCATE is identified, with $V_{\tau} = \mathbb{V}(\mu_1(X)-\mu_0(X))$.

\subsection{The VCATE and policy targeting}
Practitioners can use estimates of $\tau(X)$  to decide whom to treat in future interventions \citep{athey2021policy,kitagawa2018should,manski2004statistical}. Program managers can target the treatment recipients based on their initial covariates. However, whether targeting can substantially improve average outcomes depends on the dispersion of $\tau(x)$. I show that a simple function of the VCATE bounds the marginal gains of targeting.

Let $\gamma$ denote the joint distribution of $(X,Y_1,Y_0)$, $\mathcal{X}$ a set containing the support of $X$, $\tau_\gamma(x) := \mathbb{E}_\gamma[Y_1 - Y_0 \mid X =x]$ the CATE given $\gamma$. A function which maps $x$ to a probability of treatment $\pi(x)$ is known as a statistical allocation rule \citep{manski2004statistical}. Furthermore, I denote the set of all possible allocation rules by $\Pi$, which contains all functions $\{\pi:\mathcal{X} \to [0,1]\}$.  The set $\Pi$ includes many well-known assignment rules. For instance, it includes the  ``non-targeted'' policy which assigns everyone to treatment if $\mathbb{E}_\gamma[Y_1] > \mathbb{E}_\gamma[Y_0]$ and to the control group otherwise. Moreover, the average outcome under rule $\pi$  is $\mathbb{E}_\gamma[\pi(X)Y_1 + (1-\pi(X))Y_0]$, and the marginal benefit compared to the non-targeted policy is defined as $\mathcal{U}_\gamma(\pi)  := \mathbb{E}_\gamma[\pi(X)Y_1 + (1-\pi(X))Y_0] - \max\{\mathbb{E}_\gamma[Y_1],\mathbb{E}_\gamma[Y_0]\}$. 

\begin{thm}
	Let $\Gamma$ denote the set of distributions such that $\mathbb{V}_\gamma[\tau_\gamma(X)] = V_{\tau}$. For all $\gamma \in \Gamma$ and $\pi \in \Pi$,
	$$ \mathcal{U}_\gamma(\pi) \le \underbrace{\left(\begin{array}{c} \text{Welfare} \\ \text{Optimal} \\ \text{Targeting} \end{array}\right)}_{\sup_{\pi \in \Pi} \mathbb{E}_\gamma[\pi(X)Y_1 + (1-\pi(X))Y_0]} - \underbrace{\left(\begin{array}{c} \text{Welfare} \\ \text{No} \\ \text{Targeting} \end{array}\right)}_{\max\{\mathbb{E}_\gamma[Y_1],\mathbb{E}_\gamma[Y_0]\}} \le \frac{1}{2}\sqrt{V_{\tau}}.$$
	The bound is sharp in the sense that $\mathcal{U}_\gamma(\pi) = \frac{1}{2}\sqrt{V_{\tau}}$ for at least one $\gamma \in \Gamma$ and $\pi \in \Pi$.
	\label{thm:simple_regretbound}
\end{thm}
\begin{thm}
	Consider distributions where $\mathbb{E}_\gamma[\tau_\gamma(X)]= \tau_{av}$ and $\mathbb{V}_\gamma[\tau_\gamma(X)] = V_{\tau}$, then $U_\gamma(\pi)\le  \frac{1}{2}\left(-|\tau_{av}| + \sqrt{V_{\tau}+\tau_{av}^2}\right)$. This bound is sharp over this subset of distributions.
	\label{thm:general_regretbound}
\end{thm}

Theorem \ref{thm:simple_regretbound} shows the VCATE provides a welfare bound over the superset of policy classes. Furthermore, any type of restrictions on $\Pi$ such as budget constraints or incentive compatibility will achieve utilitarian welfare gains that are weakly lower than $\frac{1}{2}\sqrt{V_\tau}$. The bound in Theorem \ref{thm:simple_regretbound} and the generalization in Theorem \ref{thm:general_regretbound} provide simple bounds on the prospective gains of targeting, without needing to solve for $\pi(x)$. The bounds are most informative when $V_{\tau}$ is low. For instance, when $V_{\tau} = 0$ there is no heterogeneity explained by the observables $X$ and therefore there are no gains from targeting.  However, the fact that the bound is sharp does not imply that it is always achievable for every $\gamma \in \Gamma$, and when $V_{\tau}$ is high it is still be necessary to optimize $\pi(x)$ to determine whether personalized offers are worthwhile.

Finding the bound in Theorem \ref{thm:simple_regretbound} relies on two important insights. On one hand, the optimal policy in $\Pi$ treats an individual if and only if $\tau(x) \ge 0$ \citep{kitagawa2018should}. Substituting the optimal policy, $\sup_{\pi \in \Pi} \mathcal{U}_\gamma(\pi)$ is equal to $ \mathbb{E}_\gamma[\max\{\tau_\gamma(X),0\}] - \max\{\mathbb{E}_\gamma[\tau_\gamma(X)],0\}$. On the other hand, to avoid optimizing over all $\gamma \in \Gamma$, I break the problem down into equivalence classes based on the moments of the negative, zero, and positive components of the CATE.  I use a constructive approach to derive the most ``adversarial'' distribution. The upper bound is achieved when the CATE has a binary support, which is partly why the bounds in Theorems \ref{thm:simple_regretbound} and \ref{thm:general_regretbound} have simple closed forms.

\begin{cor}
	\label{cor:invariance_transformation}
	Let $\kappa_1,\kappa_2 \in \mathbb{R}$ and define a new outcome $\tilde{Y} = \kappa_1 + \kappa_2 Y $. The maximum welfare gain for the transformed outcome is $\frac{|\kappa_2|}{2}(-|\tau_{av}| + \sqrt{V_{\tau} + (\tau_{av})^2})$. 
\end{cor}
Corollary \ref{cor:invariance_transformation} shows that the welfare bound is invariant to location shifts in the outcome, and grows linearly with scale shifts. This result implies that transformations that change the sign, e.g. $\kappa_2 = -1$, do not change the value of the welfare bound. Consequently, the bound applies regardless of whether the welfare objective is to increase a desirable outcome or to decrease an undesirable outcome.

\subsection{Inference using regressions}
\label{sec:reg_inference}

Consider a simple situation where $X$ is real valued, the treatment $D$ is experimentally assigned with constant probability, and $U$ is a mean zero error term. The researcher runs the following linear regression,
\begin{equation}
\label{eq:regression}
Y = c_1 + c_2 X + \beta_1 D + \beta_2 D X + U, \quad \mathbb{E}[(1,X,D,DX)'U] = 0.
\end{equation}
Define the auxiliary quantities $\tau^*(x) := \beta_1 + \beta_2 x$ and  $V_x := \mathbb{V}(X)$. The pseudo-VCATE is defined as
\begin{equation} V_{\tau}^* := \mathbb{V}(\tau^*(X)) = \beta_2^2V_x.
\end{equation}	
The pseudo-VCATE has a close connection to the VCATE. If the linear model describes the conditional mean $\mu_d(x)$, then $\tau^*(x) = \tau(x)$ and $V_{\tau} = V_{\tau}^*$. For instance, in models with binary $X$, the functional form is correctly specified and $V_{\tau} = V_{\tau}^*$. For now, assume that the pseudo-VCATE and the VCATE coincide. In later sections, I analyze models that allow for misspecification.

Consider a sequence of distributions $\{\gamma_n\}_{n=1}^\infty \in \Gamma^{\infty}$. I index the regression coefficients and model variances by the sample size $n$ as $\beta_{2n}$ and $(V_{xn}, V_{\tau n})$, respectively. Define an estimator of the VCATE as $\widehat{V}_{\tau n} = \widehat{\beta}_{2n}^2\widehat{V}_{xn}$, where $\widehat{\beta}_{2n}$ is the least squares estimator of \eqref{eq:regression} and $\widehat{V}_{xn} :=\frac{1}{n}\sum_i X_i^2 -\left[\frac{1}{n}\sum_i X_i \right]^2$. With some algebraic manipulations the estimation error can be decomposed as
\begin{equation}
\label{eq:decomposition_simplevtaustar}
\widehat{V}_{\tau n}-V_{\tau n}^* = V_{xn}(\widehat{\beta }_{2n}-\beta_{2n})^2\left(\frac{\widehat{V}_{xn}}{V_{xn}}\right)+ 2\beta_{2n}V_{xn}(\widehat{\beta}_{2n}-\beta_{2n})\left(\frac{\widehat{V}_{xn}}{V_{xn}}\right) + \beta_{2n}^2V_{xn}\left(\frac{\widehat{V}_{xn}}{V_{xn}}-1\right).
\end{equation}	
To derive the asymptotic distribution we can apply the central limit theorem to individual components. For generality, I state joint convergence to a normal distribution as an assumption. This holds as a special case if the observations are i.i.d. and key moments of the distribution are bounded, but may also hold under other forms of dependence. I defer stating primitive conditions until Section \ref{dml_regression}.
	\begin{assump}
	\label{assump:clt_regcoef}
	There is a sequence of distributions $\{\gamma_n\}_{n=1}^\infty \in \Gamma^{\infty}$ with associated quantities $\{V_{xn},V_{\tau n}^*,\beta_{2n},\Omega_n \}_{n=1}^{\infty}$, which are related by the identity $V_{\tau n}^* = \beta_{2n}^2V_{xn}$, and satisfy the following properties: (i) $V_{xn} > 0$, (ii) $V_{\tau n}^*$ is contained in a bounded subset of $[0,\infty)$, and (iii) $\Omega_n$ is a positive definite matrix with eigenvalues bounded way from zero and a finite upper bound.  There is a sequence of estimators  $\{\widehat{V}_{xn},\widehat{V}_{\tau n},\widehat{\beta}_{2n},\widehat{\Omega}_n \}_{n=1}^{\infty}$ which satisfy 	$\widehat{V}_{\tau n} = \widehat{\beta}_{2n}^2\widehat{V}_{xn}$. As $n \to \infty$,  $\widehat{\Omega}_n \to^p \Omega_n$, and
	\begin{equation}
	\label{eq:clt_coefficients} \Omega_n^{-1/2}\sqrt{n}\begin{pmatrix} \sqrt{V_{xn}}(\widehat{\beta}_{2n} - \beta_{2n}) \\ \frac{\widehat{V}_{xn}}{V_{xn}} - 1 \end{pmatrix} \to^d Z_n \sim \mathcal{N}(0,I_{2 \times 2}).
	\end{equation}		
\end{assump}
The normalization by $V_{xn}$ is intended to align with the decomposition in \eqref{eq:decomposition_simplevtaustar}. The $2 \times 2$ matrix $\Omega_n$ is an estimator of the covariance matrix. I present Assumption \ref{assump:clt_regcoef} as a triangular array because it makes it easier to formalize discussions of uniform coverage over the parameter space. Assumption \ref{assump:clt_regcoef} allows for cases where $V_{\tau n}^*$ is arbitrarily close to or includes zero. Let $\Omega^{1/2}$ denote the Cholesky decomposition of a matrix $\Omega$. The estimator of the VCATE converges to the empirical process $G$, defined as
\begin{equation}
\label{eq:empiricalprocess}
G(n,V_{\tau}^*,\Omega,z,\zeta) := \frac{(e_1'\Omega^{1/2}z)^2}{n} + 2 \zeta \sqrt{\frac{V_{\tau}}{n}^*}(e_1'\Omega^{1/2}z)+ \frac{V_{\tau}^*}{\sqrt{n}}(e_2'\Omega^{1/2}z),
\end{equation}
where $z \in \mathbb{R}^2$, $\zeta \in \{-1,1\}$, $e_1 = [1,0]'$ and $e_2 = [0,1]'$.
\begin{lem}
	\label{lem:decomposition_remainder_lvci_reg}		
	Suppose that Assumption \ref{assump:clt_regcoef} holds, then $\widehat{V}_{\tau n}-V_{\tau n}^* = O_p\left( \max\left\{\frac{1}{n},\sqrt{\frac{V_{\tau n}^*}{n}}\right\} \right)$, and there exists a sequence of $\zeta_n \in \{-1,1\}$, such that
	\begin{equation}
	\label{eq:decomposition_remainder_lvci_reg}
	\widehat{V}_{\tau n}-V_{\tau n}^* = G(n,V_{\tau n}^*,\Omega_n,Z_n,\zeta_n) +  o_p\left(\frac{1}{n}\right) +  o_p\left(\sqrt{\frac{V_{\tau n}^*}{n}}\right)+ o_p\left(\frac{V_{\tau n}^*}{\sqrt{n }}\right).
	\end{equation}
	
\end{lem}
Lemma \ref{lem:decomposition_remainder_lvci_reg} shows that the limiting distribution of $\widehat{V}_{\tau n}$ is a linear combination of a Chi-square and a normal, whose weights depend on the value of $V_{\tau n}^*$. The relative magnitude of $V_{\tau n}^*$ determines the fit of the normal approximation. In the heterogeneous case, $V_{\tau n}^* \ge \delta > 0$, $\sqrt{n}G$ converges to a normal as $n \to \infty$ because the first term in \eqref{eq:empiricalprocess} is asymptotically negligible. However, when $V_{\tau n}^* = 0$, only the first term remains and $nG$ converges to a non-central Chi-Square distribution, which is asymmetric. Using normal critical values here (even if everything else was known) would produce distorted coverage. Furthermore, when $V_{\tau n}^* = 0$ the rate of convergence is $n$, which is faster than $\sqrt{n}$, and hence the estimator is ``super consistent'' near the boundary. The error is dominated by the first stage sampling uncertainty in estimating the nuisance parameter $\beta_{2n}^2$, which converges at $n$ rate.

In practice, all three components in \eqref{eq:empiricalprocess} contribute to the limiting distribution, and this information can be used for inference. I propose an analytic approach based on the quantiles of the empirical process that can deliver exact coverage. Let  $F_{n,V_{\tau}^*,\Omega,\zeta}(v)$ be the conditional CDF of the empirical process, defined as
\begin{equation} F_{n,V_{\tau}^*,\Omega,\zeta}(v) = \mathbb{P}( G(n,V_{\tau}^*,\Omega,Z,\zeta) \le v ), \qquad Z \sim \mathcal{N}(0,I_{2 \times 2}), \quad v \in \mathbb{R}.
\end{equation}
Based on this CDF we can construct a test statistic,
$$ F_{n,V_{\tau}^*,\widehat{\Omega}_{n},\zeta}(\widehat{V}_{\tau n}-V_{\tau}^*), $$
indexed by unknown values of $(V_{\tau}^*,\zeta)$ and substituting the estimated covariance matrix $\widehat{\Omega}_{n}$.
By construction, the test statistic is contained in $[0,1]$. Similarly, I construct critical values as functions of the parameters for a nominal level $\alpha$, as follows	
\begin{equation}
\begin{aligned}
q_{\alpha/2}(n,V_{\tau}^*,\Omega,\zeta)  &:= \min
\left\{\alpha/2,F_{n,V_{\tau}^*,\Omega,\zeta}(0)\right\} \\
q_{1-\alpha/2}(n,V_{\tau}^*,\Omega,\zeta) &:= 1-\alpha+\min\left\{\alpha/2,F_{n,V_{\tau}^*,\Omega,\zeta}(0)\right\}
\end{aligned}
\label{eq:adjusted_criticalvalues}
\end{equation}
The difference in the critical values is $(1-\alpha)$ to achieve the desired coverage. The lower critical value is the minimum of the $\alpha/2$ percentile and $0$. This adjustment is meant to increase the power of tests of homogeneity (see Remark \ref{rem:adjust_critical_values}). I propose an adaptive confidence interval by substituting the $\widehat{\Omega}_{n}$, $n$, and $\widehat{V}_{\tau n}$ into the following formula
\begin{align} \label{eq:feasible_confidenceset}
\begin{split}
\widehat{CI}_{\alpha n} = \bigg\{ &V_{\tau }^* \in \mathbb{R}_{+}, \zeta \in \{-1,1\}: \\  &F_{n,V_{\tau}^*,\widehat{\Omega}_{n},\zeta}(\widehat{V}_{\tau n}-V_{\tau}^*) \in \left[ q_{\alpha/2}(n,V_{\tau}^*,\widehat{\Omega}_n,\zeta), q_{1-\alpha/2}(n,V_{\tau}^*,\widehat{\Omega}_n,\zeta)\right] \bigg\}.
\end{split}
\end{align}
The set $\widehat{CI}_{\alpha n}$ can be constructed via a grid search between $0$ and an arbitrarily high value, to test whether a particular $V_{\tau}^*$ satisfies the inequality constraints. The procedure achieves correct asymptotic size because the test statistic converges to a uniform random variable in $[0,1]$ for each value of $V_{\tau}^*$. In general, the distribution in \eqref{eq:decomposition_remainder_lvci_reg} depends on the value of $\zeta$ and I obtain a conservative interval in \eqref{eq:feasible_confidenceset} by considering the union of intervals with different values of $\zeta$. Moreover, if the off-diagonal element of $\Omega_n$ is zero, then the distribution of the empirical process in \eqref{eq:empiricalprocess} does not depend on the value of $\zeta$. This property is plausible and I introduce primitive conditions that satisfy it in Section \ref{largesampletheory}. Under those conditions the confidence interval has exact asymptotic coverage .

The procedure is fast because at each point in the grid the researcher evaluates the condition in \eqref{eq:feasible_confidenceset}, using the same estimate of $(\widehat{V}_{\tau n},\widehat{\Omega}_n)$. The critical values can be computed numerically from the quantiles of a generalized Chi-square with distribution $F$, which are available in most statistical software packages.

\begin{rem}[Equivalence of homogeneity test, $\beta_2 = 0$]
	Researchers can test for homogeneity by evaluating whether $0 \in \widehat{CI}_{\alpha n}$. By definition, $e_1'\widehat{\Omega}_n^{1/2}Z = \sqrt{\Omega_{n,11}}Z_1$, where $\Omega_{n,11}$ is the upper-left entry. Under the null, the test statistic is  $F_{n,0,\widehat{\Omega}_n,\zeta}(v) = \mathbb{P}(\widehat{\Omega}_{n,11}Z_1^2/n \le v )$, which is the CDF of a rescaled Chi-square distribution with one degree of freedom. Furthermore, the critical values are $\{0,1-\alpha\}$, given that $F_{n,0,\widehat{\Omega}_n,\zeta}(0) = 0$. Neither quantity depends on the choice of $\zeta$. Because of the normalization in \eqref{eq:clt_coefficients}, we can choose $\widehat{\Omega}_{n,11} = \widehat{V}_{xn} \widehat{\mathbb{V}}(\widehat{\beta}_{2n})$, where $\widehat{\mathbb{V}}(\widehat{\beta}_{2n})$ is an estimate of the asymptotic variance of $\widehat{\beta}_{2n}$ such as the robust sandwich estimator. Therefore, evaluating $F_{n,0,\widehat{\Omega}_n,\zeta}(\widehat{V}_{\tau n}) \in [0,1-\alpha]$ is algebraically equivalent to a test of whether $n(\widehat{\beta}_{2n}^2\widehat{V}_{xn}/(\widehat{V}_{x n}\widehat{\mathbb{V}}(\widehat{\beta}_{2n})) = n\widehat{\beta}_{2n}^2/\widehat{\mathbb{V}}(\widehat{\beta}_{2n})$ exceeds the $1-\alpha$ quantile of a Chi-square with one degree of freedom. This is identical to a test of $\beta_2 = 0$ in the regression in \eqref{eq:regression}.

\end{rem}

\begin{rem}[Adjusting critical values]
	\label{rem:adjust_critical_values}
	The critical value $q_{\alpha/2}(V_{\tau}^*,\Omega,n,\zeta)$  in \eqref{eq:adjusted_criticalvalues} is constructed to guarantee that $\widehat{V}_{\tau n} \in \widehat{CI}_{\alpha n}$. In this case, $\widehat{V}_{\tau n}$ belongs to the CI if and only if $F_{n,\widehat{V}_{\tau n},\widehat{\Omega}_{n},\zeta}(0)$ is contained in the critical region for some $\zeta \in \{-1,1\}$. The unadjusted CI with critical values $\left\{ \alpha/2, 1-\alpha/2\right\}$ is not guaranteed to contain  the test statistic.\footnote{For example, suppose that $\widehat{V}_{\tau n} = 0$. Then the empirical process has a Chi-square distribution for $V_{\tau} = 0$. Since the unadjusted critical value is bounded away from zero,  $\widehat{V}_{\tau n}$ would not be contained in the unadjusted CI.} Another rationale for doing the adjustment in \eqref{eq:adjusted_criticalvalues}, is to increase the power of the test of homogeneity, $0 \in \widehat{CI}_{\alpha n}$, relative to a test based on the unadjusted CI. The unadjusted test has correct size but the rejection region is discontinuous: it rejects when the test statistics is very close to zero or when it exceeds a threshold. Instead, the adjusted test shifts the critical region left and has the form of a Chi-squared test. It only rejects the null if the test statistic is larger than $1-\alpha$, which is a threshold that is smaller than $1-\alpha/2$ for the unadjusted CI.
\end{rem}

\begin{rem}[Comparison to other tests of homogeneity]
	\cite{crump2008nonparametric} suggest estimating $\widehat{\mu}_d(x)$ by a series estimator with $K$ terms, for subsamples $D = d \in \{0,1\}$. They propose a bias-corrected Wald statistic, which takes the form $\widehat{T}_n^{series} := \{[\widehat{\xi}_1 - \widehat{\xi}_0 ]'[\widehat{\mathbb{V}}(\widehat{\xi}_1-\widehat{\xi}_0)]^{-1}[\widehat{\xi}_1 - \widehat{\xi}_0]-(K-1)\}/\sqrt{2(K-1)}$, where $(\widehat{\xi}_1,\widehat{\xi}_0)$ are non-intercept coefficients associated with $\widehat{\mu}_1(x)$ and $\widehat{\mu}_0(x)$, respectively and $K$ is the number of covariates. For regressions with univariate $X$ as in \eqref{eq:regression}, $K=2$ and $\widehat{T}_n^{series} = (n\widehat{\beta}_{2n}^2/\widehat{\mathbb{V}}(\widehat{\beta}_{2n}) - 1) / \sqrt{2}$. Essentially this is just a transformation of the test statistic proposed above, which will produce the same acceptance/rejection result for significance level $\alpha$ (using the critical values in their equation 3.11). \cite{ding2019decomposing} study a framework with a fixed population where the only source of randomness is the experimental assignment of offers. They propose a similar Wald estimator, but replace estimates of $(\widehat{\xi}_1,\widehat{\xi}_0)$ and the asymptotic variance with randomization inference counterparts. In samples with large $n$, this leads to very similar test statistics, but may produce slightly different results in small samples.
	
	The approach that I introduce in the following section differs substantially in the way that I handle multivariate cases. For $K > 2$, the approaches are non-nested because I use sample splitting and consider a wider range of methods to estimate $\widehat{\mu}_d(x)$ than series estimators. 
\end{rem}


\section{Inference for nonparametric CATE}
\label{adaptiveinference}

In this section I provide an overview of the inference problems associated with efficient estimators of the VCATE and how to solve them for the nonlinear/high-dimensional case. Let $\{Y_i,D_i,X_i\}$ be i.i.d.. As shown in \cite{newey1990semiparametric}, efficient estimators can be decomposed as
\begin{equation}
\sqrt{n}(\widehat{V}_{\tau n} - V_{\tau n}) = \frac{1}{\sqrt{n}}\sum_{i=1}^n \varphi_i +  \underbrace{\  \text{Residual}_n\  }_{o_p(1)},
\label{eq:efficient_lineardecomposition}
\end{equation}
where $\varphi_i$ is an i.i.d. realization from the efficient influence function with mean zero, and the residual becomes asymptotically negligible as $n \to \infty$. The semiparametric lower bound is $\mathbb{V}(\varphi_i)$. Let $\eta(\cdot)$ be a set of nuisance functions defined as
\begin{equation}
\eta(x) := (\tau(x),\mu_0(x), p(x),\tau_{av}).
\label{eq:defn_nuisancefunctions}
\end{equation}
\cite{levy2021fundamental} showed that the efficient influence function for the VCATE is equal to $\varphi_i = \varphi(Y_i,D_i,X_i,\eta) - V_{\tau n}$, where $\varphi$ is defined as
\begin{equation} \varphi(y,d,x,\eta) := (\tau(x)-\tau_{av})^2 + 2 (\tau(x)-\tau_{av})\left[ \frac{d(y-\mu_0(x)-\tau(x))}{p(x)} - \frac{(1-d)(y-\mu_0(x))}{1-p(x)} \right].
\label{eq:efficient_influence}
\end{equation}
By \eqref{eq:efficient_lineardecomposition}, all  efficient estimators --regardless of their form-- are $\sqrt{n}$--asymptotically equivalent to $\frac{1}{n}\sum_{i=1}^n\varphi_i$.  Let $\varphi_{i} = \varphi(Y_i,D_i,X_i,\eta(X_i))$ be a realization of the efficient influence function. If $\mathbb{V}(\varphi_i) > 0$, then  
$$\mathbb{V}(\varphi_i)^{-1}\times \sqrt{n}(\widehat{V}_{\tau n}-V_{\tau n}) =  \mathbb{V}(\varphi_i)^{-1}\left[\frac{1}{\sqrt{n}}\sum_{i=1}^n \varphi_i \right]+ o_p(1) \to^p \mathcal{N}(0,1). $$
In this case, any confidence interval based on normal critical values and a consistent estimator of $\mathbb{V}(\varphi_i)$, produces valid coverage. For common  functionals such as the ATE, $\mathbb{V}(\varphi_i) > 0$. However, this cannot be guaranteed for the VCATE. 

\begin{lem}
	\label{lem:variance_eff_influence}
	Let $\sigma_d^2(x) := \mathbb{V}(Y \mid D =d,X=x)$. 
	\begin{equation} \mathbb{V}(\varphi_i) = \mathbb{V}((\tau(X)-\tau_{av})^2) + 4\mathbb{E}\left[(\tau(X)-\tau_{av})^2\left(\frac{\sigma_1^2(X)}{p(X)} + \frac{\sigma_0^2(X)}{p(X)}\right) \right].
	\label{eq:variance_eff_influence}
	\end{equation}
\end{lem}
Both of the inner terms in \eqref{eq:variance_eff_influence} are multiplied by $(\tau(x) - \tau)$. When the VCATE is zero, $\tau(x) = \tau$ almost surely, and the influence function is degenerate.
The condition that $\mathbb{V}(\varphi_i) > 0$ does not hold uniformly over all $V_{\tau}$ in the parameter space. In this case, the distribution of $\sqrt{n}(\widehat{V}_{\tau n}-V_{\tau n})$ is dominated by the higher order terms of the residual \eqref{eq:efficient_lineardecomposition}, and the CLT cannot be applied to guarantee normality near the boundary. The linear estimator discussed in the previous section is just one example. Moreover, if the tails of $\tau(x)$ are thin, then the value of $\mathbb{V}(\varphi_i)$ is also small near the boundary.

\begin{cor}
	\label{cor:bound_varefficient}
	If $\mathbb{E}[(\tau(X)-\tau)^4] \le \kappa^2 V_{\tau}^2$ for $\kappa \in \mathbb{R}_+$, then
	$$ \mathbb{V}(\varphi_i) \le  \kappa^2 V_{\tau}^2+ 4\kappa V_{\tau}\sqrt{\mathbb{E}\left[ \left( \frac{\sigma_1^2(X)}{p(X)}+\frac{\sigma_0^2(X)}{1-p(X)} \right) \right]}. $$		
\end{cor}
Corollary \ref{cor:bound_varefficient} shows that the variance of the efficient influence function is bounded by a quantity that scales up or down proportional to the value of $V_{\tau}$. Consequently, when $\mathbb{V}(\varphi_i)$ is relatively small, the higher order terms in the residual may still dominate.

\subsection{Pseudo-VCATE, regressions, and efficiency}

A robust way to introduce nonlinearity is to consider a regression with real-valued basis functions $M(x)$ and $S(x)$. For now, I will leave these unspecified but in the next section I will show how they can be estimated non-parametrically in a first stage.

\begin{equation}
Y = \underbrace{c_0 + M(X)c_1 + [D-p(X)]\beta_1 + [D-p(X)]S(X)\beta_2 }_{W(X,D)'\theta} + U
\label{eq:weightedreg}
\end{equation}
with weights $\lambda(X) := [p(X)(1-p(X))]^{-1}$, regressors $W(X,D)$, and parameters $\theta = [c_0,c_1,\beta_1,\beta_2]$. This specification accommodates experiments with heterogeneous assignment probabilities.\footnote{When $p(x) = 1/2$ this produces exactly the same coefficients as a regression of $Y$ on $(1,M(X),D,D\times S(X))$, but differs when the probabilities are heterogeneous. If $M(X) = S(X) = X$ as well, this reduces to \eqref{eq:regression}.} Consider the following minimizers:
\begin{equation}
\label{eq:argmin_ols}
\Theta^* := \arg\min_{\theta \in \mathbb{R}^4} \  \mathbb{E}[\lambda(X)(Y-W(X,D)'\theta)^2].
\end{equation}
\cite{chernozhukov2020generic} showed that if $\mathbb{V}(S(X)) > 0 $, $\mathbb{E}[S(X)] = 0$, and the vector $(\beta_1,\beta_2)$ are part of the solution to \eqref{eq:argmin_ols}, then $(\beta_1,\beta_2)$ are also the intercept and slope of the best linear projection of $\tau(X)$ on $S(X)$.\footnote{When $\mathbb{V}(S(X)) = 0$, $\beta_2$ does not have a unique solution in \eqref{eq:argmin_ols}, but $V_{\tau}^* = \beta_2^2\mathbb{V}(S(X)) = 0 \le V_{\tau}$ is still the best linear projection, regardless of the value of $\beta_2$.} Hence the pseudo-VCATE has an upper bound, $V_{\tau}^* = \beta_2^2\mathbb{V}(S(X)) \le V_{\tau}$. Because of this bound, if $V_{\tau} = 0$, then the pseudo-VCATE $(V_{\tau}^*)$ will not falsely detect heterogeneity even if $S(X)$ is misspecified. If anything, poor choices of $S(X)$ will possibly understate the amount of heterogeneity. When $\tau(X)$ is spanned by $S(X)$, $V_{\tau}^* = V_{\tau}$ and the two notions coincide.

To obtain a feasible estimator we define $\widehat{S}(x) := S(x) - \frac{1}{n}\sum_{i=1}^n S(X_i)$, and compute $\widehat{W}(x,d)$ by substituting $\widehat{S}(x)$ in \eqref{eq:weightedreg}. Now consider a value of $\widehat{\theta}_n$ that minimizes $\frac{1}{n}\sum_{i=1}^n \lambda(X_i)(Y_i-\widehat{W}(X_i,D_i)'\theta)^2$, by solving the first order condition

\begin{equation}
\mathcal{Q}(\widehat{\theta}_n) := \frac{1}{n}\sum_{i=1}^n \lambda(X_i)(Y-\widehat{W}(X_i,D_i)'\widehat{\theta}_n) \widehat{W}(X_i,D_i)'.
\label{eq:foc_adjustedols}
\end{equation}
The regression parameters can be used to construct the CATE and other nuisance functions. For a given $\theta \in \mathbb{R}^4$, 
\begin{equation} 
\tilde{\eta}_\theta(x) := \begin{pmatrix} \tilde{\tau}_\theta(x) \\ \tilde{\mu}_{0,\theta}(x) \\ \tilde{p}_\theta(x)  \\  \tilde{\tau}_{av,\theta} \end{pmatrix} = \begin{pmatrix}  (\widehat{W}(x,1)-\widehat{W}(x,0))'\theta \\ \widehat{W}(x,0)'\theta \\ p(x) \\ \frac{1}{n}\sum_{i=1}^n \tilde{\tau}_\theta(X_i) \end{pmatrix}.
\label{eq:pseudonuisancefunctions}
\end{equation}
Lemma \ref{lem:equivalence_efficient_simple} shows that the sample variance of the estimated CATE can be interpreted as an estimator that plugs-in \eqref{eq:pseudonuisancefunctions} to the efficient influence function in \eqref{eq:efficient_influence}.

\begin{lem} 
	\label{lem:equivalence_efficient_simple}
	Define $\widehat{V}_{xn} = \frac{1}{n}\sum_{i=1}^n\widehat{S}(X_i)^2$. Let $\varphi$, $\tilde{\eta}_\theta$ and $\mathcal{Q}(\theta)$ be defined as in \eqref{eq:efficient_influence}, 
	\eqref{eq:foc_adjustedols}, and \eqref{eq:pseudonuisancefunctions}, respectively. If $\widehat{\theta}_n = (\widehat{c}_{1n},\widehat{c}_{2n},\widehat{\beta}_{1n},\widehat{\beta}_{2n})$ solves $\mathcal{Q}(\widehat{\theta}_n) = 0$, then $\widehat{V}_{\tau n} := \widehat{\beta}_{2n}^2\widehat{V}_{xn} = \frac{1}{n}\sum_{i=1}^n \varphi(Y_i,D_i,X_i,\tilde{\eta}_{\widehat{\theta}_n})$.	
\end{lem}
Mechanically, the influence function can be decomposed into primary and bias-correction components. As an intermediate step for Lemma \ref{lem:equivalence_efficient_simple}, I show that the bias-correction terms and the fourth component of \eqref{eq:foc_adjustedols}  are proportional to each other. The optimal $\widehat{\theta}_n$ implicitly sets the average bias correction to zero. Intuitively, the linear model minimizes the covariate imbalances between the treatment and control group in-sample. Lemma \ref{lem:equivalence_efficient_simple} suggests that $\widehat{V}_{\tau n}$ could be asymptotically efficient if $\tilde{\eta}_{\widehat{\theta}_n}$ is sufficiently close to $\eta$.  In Section \ref{consistencyefficiency}, I show that my proposed semiparametric estimator can indeed achieve this.


As a preliminary step, it is necessary to determine which $S(x)$ and $M(x)$ ensure that $\tilde{\eta}_\theta = \eta$. Not all choices achieve this property.\footnote{This point highlights that while all regressions of the form in 
	\eqref{eq:weightedreg} proposed by \cite{chernozhukov2020generic} estimate an interpretable $V_{\tau}^*$ --regardless of the choice of $M(x)$--, not every regression in this class is efficient. The functions $S(x)$, and $M(x)$ in particular, both affect efficiency.} However, if they are chosen in such a way that $W(x,d)'\theta = \mu_d(x)$ for some $\theta \in \mathbb{R}^4$, then that's sufficient to guarantee that  $\eta_\theta = \eta$. Lemma \ref{lem:projection_weak} shows that any $\theta$ with this property is also a solution to the regression problem, and provides guidance on the choice of $S(x)$ and $M(x)$.
\begin{lem}
	\label{lem:projection_weak}
	Let $ \Theta^*$ be the optimizer set defined in \eqref{eq:argmin_ols}. If (i) $\mathbb{E}[S(X)] = 0$ and (ii) $W(x,d)'\theta = \mu_d(x)$ for some $\theta \in \mathbb{R}^4$, then $\theta \in \Theta^*$. Conditions (i) can be satisfied by setting   $S(x) = \tau(x) - \mathbb{E}[\tau]$. Condition (ii) can  be satisfied by setting $M(x) =  \mu_0(x) + p(x)\tau(x)$. In this special case, $\theta = (0,1,\mathbb{E}[\tau(X)],1)' \in \Theta^*$.
	
\end{lem}
Lemma \ref{lem:projection_weak} provides efficient choices of $S(x)$ and $M(x)$ that can be expressed in terms of conditional moments, and that for this choice, the optimal $\theta$ has a known, simple form. In practice, $S(x)$ and $M(x)$ can be estimated non-parametrically.

\subsection{Multi-step approach}
\label{dml_regression}
My proposed procedure randomly partitions the observations $\mathcal{I}_n := \{1,\ldots,n\}$ into $K$ folds of equal size $n_k := n/K$. Denote the observations in each fold by $\mathcal{I}_{nk}$, so that $\bigcup_{k=1}^{K} \mathcal{I}_{nk} = \mathcal{I}_n$, and let $\mathcal{I}_{-nk} := 
\mathcal{I}_n \backslash \mathcal{I}_{nk}$ be the set of observations that are not in fold $k$. In a slight abuse of notation, I use $\mathcal{I}_{-nk}$ when defining conditional expectations, to denote the full set of random variables associated with observations not included in fold $k$. For simplicity, I also label the fold of observation $i \in \mathcal{I}_{nk}$, by $k_i$. 

Let $\widehat{\eta}_{-k}(x) := (\widehat{\tau}_{-k}(x),\widehat{\mu}_{0,-k}(x), p(x),\widehat{\tau}_{-k,av})$ denote a prediction of the nuisance function $\eta(x)$ over the set $\mathcal{I}_{-nk}$, using the researcher's preferred prediction algorithm. This could include traditional methods such as linear regression, or more modern ``machine learning'' approaches such as LASSO, neural networks, or random forests. The only function that is known in advance is the propensity score, since I restrict attention to randomized experiments. Guided by Lemma \ref{lem:projection_weak}, define
\begin{align}
\begin{split}
S_{-k}(x) &:= \widehat{\tau}_{-k}(x) - \mathbb{E}[\widehat{\tau}_{-k}(X_i) \mid \mathcal{I}_{-nk}], \\
M_{-k}(x) &:= \widehat{\mu}_{0,-k}(x) + p(x)\widehat{\tau}_{-k}(x), \\	
\lambda(x) &:= [p(x)(1-p(x))]^{-1}, \\
W_{i}' &:= \begin{bmatrix} 1 & M_{-k_i}(X_i) & (D_i-p(X_i)) & (D_i-p(X_i))S_{-k_i}(X_i)  \end{bmatrix}
\end{split}
\label{eq:dml_defn_auxiliaryfunctions}
\end{align}
Consider a regression with weights $\lambda(X_i)$, parameters $\theta := (c_1,c_2,\beta_1,\beta_2)$, and 
\begin{equation}
\label{eq:outcomemodel_dml}
Y_{i} = 
W_i'\theta_{nk} + U_i, \qquad \mathbb{E}[W_iU_i \mid \mathcal{I}_{-nk}] = 0
\end{equation}
In practice, $\mathbb{E}[\widehat{\tau}_{-k}(X_i) \mid \mathcal{I}_{-nk}]$ needs to be estimated, and I use a sample analog:
\begin{align}\begin{split}
\widehat{S}_{-k}(x) &:= \widehat{\tau}_{-k}(x) - \widehat{\tau}_{nk,av},\qquad \widehat{\tau}_{nk,av} := \frac{1}{n_k}\sum_{i \in \mathcal{I}_{nk}}\widehat{\tau}_{-k}(X_i), \\	
\widehat{W}_{i}' &:= \begin{bmatrix} 1 & M_{-k_i}(X_i) & (D_i-p(X_i)) & (D_i-p(X_i))\widehat{S}_{-k_i}(X_i) 
\end{bmatrix}, \\
\widehat{\theta}_{nk} &:= \left[ \frac{1}{n}\sum_{i \in \mathcal{I}_{nk}} \widehat{W}_i\widehat{W}_{i}' \right]^{-1} \left[ \frac{1}{n}\sum_{i \in \mathcal{I}_{nk}} \widehat{W}_iY_i \right].
\end{split}
\label{eq:thetahat_dml}
\end{align}
Let $\widehat{\theta}_{nk} = (\widehat{c}_{1nk},\widehat{c}_{2nk},\widehat{\beta}_{1nk},\widehat{\beta}_{2nk})$ be the estimator over the subsample $\mathcal{I}_{nk}$.
The fold-specific variance of $\widehat{S}_{-k_i}(x) $ is defined as
\begin{equation}
\label{eq:vxnk_dml}
\widehat{V}_{xnk} := \frac{1}{n_k}\sum_{i \in \mathcal{I}_{nk}} 	\widehat{S}_{-k_i}(X_i)^2.
\end{equation}
The estimator of the VCATE for fold $k$ is
\begin{equation}
\label{eq:vtaunk_dml}
\widehat{V}_{\tau n k} := \widehat{\beta}_{2 nk}^2\widehat{V}_{x nk}.
\end{equation}
In this case $\widehat{V}_{xnk}$ can be viewed as a preliminary estimate of the VCATE  using the data in $\mathcal{I}_{-nk}$, whereas $\widehat{V}_{\tau nk}$ is a regression-adjusted estimator that fits the sample $\mathcal{I}_{nk}$. This adjustment will produce better results, with a pseudo-VCATE interpretation even if the first step $\widehat{S}_{-k_i}(x)$ function is noisy, misspecified, or slow to converge to $\tau(x)$. The estimator in \eqref{eq:vxnk_dml} belongs to the class of multi-step estimators defined in \cite{chernozhukov2020generic}. I add a restriction on the choice of $M_{-k}(x)$, guided by Lemma \ref{lem:projection_weak}, to ensure asymptotic efficiency.

To quantify the uncertainty in $(\widehat{\theta}_{nk},\widehat{V}_{xnk})$ I compute a robust (sandwich) estimator. I start by defining two auxiliary residuals, $\widehat{T}_i := \widehat{V}_{x n k}^{-1}\widehat{S}_{-k_i}(X_i)^2 - 1$ and $\widehat{U}_i := Y_i - \widehat{W}_i'\widehat{\theta}_{nk}$. Let $\widehat{\Pi}_{nk}$ be a $4 \times 4$ diagonal matrix with diagonal entries $(1,1,1,\widehat{V}_{xnk}^{-1/2})$. Researchers can compute estimators of the individual components of the sandwich form $\widehat{J}_{nk}$, $\widehat{H}_{nk}$, and a selection matrix $\Upsilon$ defined as follows 

\begin{equation}
\label{eq:Jmatrix_sandwich_dml} \widehat{J}_{nk} := \begin{bmatrix} \frac{1}{n_k}\sum_{i \in \mathcal{I}_k}  \lambda(X_{i})\widehat{\Pi}_{nk} \widehat{W}_{i}\widehat{W}_{i}'\widehat{\Pi}_{nk}' & 0 \\ 0 & 1 \end{bmatrix}, \qquad \Upsilon:=\begin{bmatrix} 0 & 0 & 0 & 1 & 0 \\ 0 & 0 & 0 & 0 & 1 \end{bmatrix}
\end{equation}
\begin{equation}
\label{eq:Hmatrix_sandwich_dml}
\widehat{H}_{nk} := \frac{1}{n_k}\sum_{i \in \mathcal{I}_k} \begin{bmatrix}  \lambda(X_i)^2\left[\widehat{U}_i^2\widehat{\Pi}_{nk}\widehat{W}_{i}\widehat{W}_{i}'\widehat{\Pi}_{nk}' \right]  & \lambda(X_i) \left[\widehat{U}_i\widehat{\Pi}_{nk}\widehat{W}_{i}\widehat{T}_{i} \right]   \\ \lambda(X_i) \left[\widehat{U}_i\widehat{W}_{i}'\widehat{\Pi}_{nk}'\widehat{T}_{i} \right]  & \widehat{T}_{i}^2 \end{bmatrix}.
\end{equation}
The sandwich covariance estimator is
\begin{equation}
\label{eq:omegank_hat_dml} \widehat{\Omega}_{nk} =   \Upsilon \widehat{J}_{nk}^{-1}\widehat{H}_{nk}\widehat{J}_{nk}^{-1}\Upsilon'.
\end{equation}
The population covariance matrix is
\begin{equation}
\label{eq:omegank}
\Omega_{nk} = \mathbb{V}\left(\begin{bmatrix} \lambda(X_i)(D_i-p(x_i))V_{xnk}^{-1/2}S_{-k}(X_i) U_{i} \\ 
V_{xnk}^{-1}S_{-k}(X_i)^2
\end{bmatrix} \mid \mathcal{I}_{-nk} \right). 
\end{equation}
When the VCATE is zero, then $\widehat{V}_{xnk}$ (as a consistent estimator of $V_{\tau n}$) should converge to zero along the asymptotic sequence. To prevent asymptotic degeneracy, we need to rescale the estimands along the lines of Assumption \ref{assump:clt_regcoef}. The random variable $V_{xnk}^{-1/2}S_{-k}(X_i)$ is normalized to (conditionally) have variance one by design, even if $S_{-k}(X_i)$ converges to zero. This requires two much weaker conditions: (i) that $V_{xnk} > 0$, i.e. there is some noise in estimating the CATE;\footnote{I also propose an extension that allows for $V_{xnk} = 0$ in Remark \ref{rem:ci_degenerate}.} (ii)  $\Omega_{nk}$ has eigenvalues bounded away from zero. To ensure this, the tails  of $S_{-k}(X_i)$ need to be thin.\footnote{One sufficient additional restriction is that $\mathbb{E}[U_{i} \mid X_{i},D_{i},\mathcal{I}_{-nk}] = 0$ (the model is correctly specified), $\mathbb{V}(U_{i} \mid X_{i},D_{i} \mid \mathcal{I}_{-nk})$ is bounded away from zero, and $S_{-k}(X_i)$ has bounded kurtosis. In that case the off-diagonal elements of $\Omega_{nk}$ are zero and the diagonals are uniformly bounded. Positive-definiteness may also hold in a neighborhood where the nuisance functions are close to the true value and $\mathbb{E}[U_{i} \mid X_{i},D_{i},\mathcal{I}_{-nk}] \approx 0$.}

\begin{assump}[Moment Bounds]
	\label{assump:momentbounds_dml}
	Suppose that there exists a constant $\delta \in (0,1)$ such that for each fold $k$, almost surely, (i) $\mathbb{E}[M_{-k}(X_i)^2 \lambda(X_i) \mid \mathcal{I}_{-nk}] - \mathbb{E}[M_{-k}(X_i)\lambda(X_i)\mid \mathcal{I}_{-nk}]\mathbb{E}[\lambda(X_i) \mid \mathcal{I}_{-nk}] \ge \delta$, (ii) $\mathbb{E}[M_{-k}(X_i)^4 \mid \mathcal{I}_{-nk}] \le 1/\delta$, (iii) $\mathbb{E}[U_{i}^4 \mid \mathcal{I}_{-nk}] \le (1/\delta)$, and (iv) $\mathbb{E}[S_{-k}(X_i)^4 \mid \mathcal{I}_{-nk}] \le (1/\delta) \mathbb{E}[S_{-k}^2 \mid \mathcal{I}_{-nk}]^2$, (v) $\mathbb{E}[\widehat{\tau}_{-k}(X)^2 \mid \mathcal{I}_{-nk}] < 1/\delta$.
\end{assump}
Assumption \ref{assump:momentbounds_dml}.(i) is a rank condition that ensures that the auxiliary regressor $M_{-k}(X_i)$ is not degenerate. Assumption  \ref{assump:momentbounds_dml}.(ii) ensures that the second-moment of the candidate regressor $M_{-k}(X_i)$ is bounded. Assumption  \ref{assump:momentbounds_dml}.(iii) is a standard condition indicating that the fourth moment of the residuals are bounded. Assumption \ref{assump:momentbounds_dml}.(iv) is a bounded kurtosis condition indicating that the out-of-sample, machine learning predictions of $\tau(x)$ have thin tails. Finally, Assumption \ref{assump:momentbounds_dml}.(v) is a bound on the variance of the first-stage VCATE.

\begin{assump}[Non-degeneracy]
	\label{assump:boundsquantities}	
	The following properties hold almost surely over sequences of random data realizations  $\{\mathcal{I}_{n1},\ldots,\mathcal{I}_{nK}\}_{n=1}^{\infty}$. Conditional on $\mathcal{I}_{-nk}$: (i) $V_{xn k} > 0$, (ii) $V_{xnk}$ has a finite upper bound, (iii) $V_{\tau nk}^* := \beta_{2nk}^2V_{xnk}$ is contained in a bounded subset of $[0,\infty)$, and (iv) $\Omega_{nk}$ defined in \eqref{eq:omegank} is a positive definite matrix with bounded eigenvalues.
\end{assump}

\begin{assump}[Random Sampling]
	\label{assump:randomsampling}	
	The observations $\{Y_{0i},Y_{1i},D_{i},X_{i}\}_{i}^n$ are i.i.d. across $i$ for fixed $n$, and drawn from a sequence of data generating processes $\{\gamma_n\}_{n=1}^{\infty}$.
\end{assump}
Theorem \ref{thm:clt_lvci} shows how these primitive conditions imply an analog of Assumption \ref{assump:clt_regcoef} for the cross-fitted case.
\begin{thm}
	\label{thm:clt_lvci} Consider a sequence of random data realizations $\{\mathcal{I}_{n1},\ldots,\mathcal{I}_{nK}\}_{n=1}^{\infty}$ with associated quantities $\{V_{xnk},V_{\tau n k}^*,\beta_{2nk},\Omega_{nk} \}_{n=1}^{\infty}$ for each $k$, as well as a sequence of estimators  $\{\widehat{\beta}_{2nk},\widehat{V}_{xnk},\widehat{V}_{\tau nk}^*,\widehat{\Omega}_{nk} \}_{n=1}^{\infty}$ computed from \eqref{eq:thetahat_dml}, \eqref{eq:vxnk_dml}, \eqref{eq:vtaunk_dml}, and \eqref{eq:omegank_hat_dml}, respectively. Suppose that these quantities satisfy Assumptions \ref{assump:unconfound_sutva}.(ii),  \ref{assump:momentbounds_dml}, \ref{assump:boundsquantities}, and \ref{assump:randomsampling}. Then as $n_k \to \infty$, for all $k \in \{1,\ldots, K\}$, Conditional on a sequence of $\mathcal{I}_{-nk}$,
	
	\begin{enumerate}[(i)]
		\setlength{\parskip}{0pt}
		\item \quad $
		 \Omega_{nk}^{-1/2}\sqrt{n_k}\begin{pmatrix} \sqrt{V_{xnk}}(\widehat{\beta}_{2nk} - \beta_{2nk}) \\ \frac{\widehat{V}_{xnk}}{V_{xnk}} - 1 \end{pmatrix}\mid \mathcal{I}_{-nk} \to Z_{nk} \sim \mathcal{N}(0,I_{2 \times 2}).$
		 \item \quad $\widehat{\Omega}_{nk} \to^p \Omega_{nk}$.
	\end{enumerate}
\end{thm}

Theorem \ref{thm:clt_lvci} presents a central limit theorem for the components of $\widehat{V}_{\tau nk}^* = \widehat{\beta}_{2nk}^2\widehat{V}_{xnk}$, properly rescaled and conditional on $\mathcal{I}_{-nk}$. This result holds regardless of whether the nuisance parameters are properly specified and primarily relies on the  independence of the folds. By Lemma \ref{lem:decomposition_remainder_lvci_reg}, conditional on $\mathcal{I}_{-nk}$,
\begin{equation}
\label{eq:decomposition_remainder_lvci_reg_dml}
\widehat{V}_{\tau n k}-V_{\tau n k}^* = G(n_k,V_{\tau n k}^*,\Omega_{nk},Z_{nk}) +  o_p\left(\frac{1}{n_k}\right) +  o_p\left(\sqrt{\frac{V_{\tau n k}^*}{n_k}}\right)+ o_p\left(\frac{V_{\tau n k}^*}{\sqrt{n_k}}\right).
\end{equation}
Then it is possible to construct adaptive confidence intervals, substituting the sample size $n_k$ and estimated statistics $(\widehat{V}_{\tau nk},\widehat{\Omega}_{nk})$.
\begin{align} \label{eq:feasible_confidenceset_dml}
	\begin{split}
		\widehat{CI}_{\alpha n k} = \bigg\{ &V_{\tau }^* \in \mathbb{R}_{+}, \zeta \in \{-1,1\}: \\  &F_{n_k,V_{\tau}^*,\widehat{\Omega}_{nk},\zeta}(\widehat{V}_{\tau nk}-V_{\tau}^*) \in \left[ q_{\alpha/2}(n_k,V_{\tau}^*,\widehat{\Omega}_{nk},\zeta), q_{1-\alpha/2}(n_k,V_{\tau}^*,\widehat{\Omega}_{nk},\zeta)\right] \bigg\}.
	\end{split}
\end{align}
The confidence intervals take the same form as in the regression case in \eqref{eq:feasible_confidenceset}, except that now the inputs are obtained from the cross-fitted regression step. The confidence interval is fast to compute because $(\widehat{V}_{\tau nk},\widehat{\Omega}_{nk})$ only needs to be computed once. It is worth noting that because the confidence interval only uses information in fold $\mathcal{I}_{nk}$, the effective sample size is $n_k$. While this does not affect the nominal asymptotic size of the confidence interval, it may affect the power of tests against specific alternatives.

\subsubsection{Ensemble estimator}

We can construct an ``ensemble'' to aggregate across folds, defined as follows
\begin{equation}
\widehat{V}_{\tau n} := \frac{1}{K}\sum_{k=1}^K \widehat{\beta}_{2 n k}^2 \widehat{V}_{xnk}.
\label{eq:multistep_estimator}
\end{equation}
In Section \ref{consistencyefficiency}, I show that this ensemble estimator is efficient.

\subsubsection{Splitting uncertainty and median intervals}

So far in this section we have used the data from a single split or fold of the data. However, the choice of fold $k$ or the particular split may lead to different values of $\widehat{V}_{\tau nk}$ and hence distinct confidence intervals. \cite{chernozhukov2020generic} propose an aggregation procedure based on ``median parameter'' confidence intervals, inspired by false-discovery rate adjustments. Their proposed conditional t-tests are not directly applicable here because $\widehat{V}_{\tau nk}$ conditionally converges to a generalized Chi-square. However, I show that the basic idea can still be adapted.

Let $K$ be the total number of folds, obtained across one or more splits of the data. For instance, a 2-fold sample with 10 splits would have $K=20$, Let $\inf \widehat{CI}_{\alpha nk}$ and $\sup \widehat{CI}_{\alpha nk}$ denote the lower and upper bounds of $\widehat{CI}_{\alpha nk}$, respectively, and $\text{Med}_K\{\cdots\}$ denote the median over a  set indexed by $k = \{1,\ldots,K\}$. If $K$ is even then two quantities might be tied for the median, and in that case I compute their midpoint. The multifold confidence interval is defined  as


\begin{equation}
\widehat{CI}_{\alpha n}^{\text{multifold}} = \left[ \text{Med}_{K}\left\{ \inf \widehat{CI}_{\frac{\alpha}{2} nk} \right\}, \text{Med}_{K}\left\{ \sup \widehat{CI}_{\frac{\alpha}{2} nk} \right\} \right].
\label{eq:multisplit_cis}
\end{equation}

Intuitively, the $K$ fold-specific intervals ``vote'' to include a particular value, and $V_{\tau}^* \in \widehat{CI}_{\alpha n}^{\text{multifold}}$ only if there is a majority vote. The ``median'' interval $\widehat{CI}_{\alpha n}^{\text{multifold}}$ contains values within the median lower bound and the median upper bound across folds.  To control the overall false discovery rate, I adjust the nominal size to $\alpha/2$. This adjustment produces a conservative interval because it assumes a worst-case dependence structure between the folds and the splits, regardless of the size of $K$. In some instances, the asymptotic coverage probability may be strictly higher than $(1-\alpha)$, particularly when there is a lot of heterogeneity.\footnote{For instance, given a single split, Theorem \ref{thm:rootn_consistency_and_efficiency} implies that the $\{\sqrt{n_k}(\widehat{V}_{\tau nk}-V_\tau)\}_{k=1}^K$ are asymptotically uncorrelated. However, near the boundary, the estimators converge at a rate faster than $\sqrt{n}$ and their relative dependence structure at that rate is unclear.} At the boundary, with low effect heterogeneity or none at all, it is much harder to asses the dependence structure between the fold-specific estimators. One of the benefits of using a worst-case approach is that it provides coverage guarantees under weak assumptions. Moreover, the empirical example illustrates that even though these intervals are conservative, they may have a short length in practice.


\section{Large Sample Theory}
\label{largesampletheory}

\subsection{$\sqrt{n}$-Consistency, Efficiency, and Boundary Rates}
\label{consistencyefficiency}
Let $\gamma \in \Gamma$ denote a probability distribution over i.i.d observations $(Y_{1i},Y_{0i},D_i,X_i)$. I use the notation $\mathbb{E}_{\gamma}[\cdot]$ and $\mathbb{P}_{\gamma}(\cdot)$ to denote the expectation and probability under $\gamma$, respectively. Let $S_{-k}(x)$ be the function defined in \eqref{eq:dml_defn_auxiliaryfunctions}. The true value of the CATE and VCATE is given by $\tau_\gamma$ and $V_{\tau}(\gamma)$, respectively. The pseudo-VCATE is given by
\begin{equation} 
\label{eq:defn_pseudovcate}
V_{\tau}^*(\gamma,\mathcal{I}_{-nk}) := V_{\tau}(\gamma) - \inf_{(\beta_1,\beta_2)\in \mathbb{R}^2}\mathbb{E}_{\gamma}[(\tau_\gamma(X) - \beta_1 - \beta_2 S_{-k}(X))^2  \mid \mathcal{I}_{-nk}].
\end{equation}
Define the estimation error of the CATE in the $L_2$ norm as
\begin{equation}
\omega\left(\gamma \right):= \sqrt{\mathbb{E}_\gamma[\Vert \widehat{\tau}_{-k}(X) - \tau_\gamma(X) \Vert^2]}. \end{equation}
We can bound the difference between the pseudo-VCATE and its true value:
\begin{thm}[Bias of the pseudo-VCATE] 
	\label{lem:convergence_lvci_to_vci} Under the distribution $\gamma \in \Gamma$,
	\begin{equation} \mathbb{E}_\gamma[| V_{\tau}(\gamma)-V_{\tau}^*(\gamma,\mathcal{I}_{-nk}) |] \le \min\left\{16 \times \omega(\gamma)^2 ,V_{\tau}(\gamma)\right\}.
	\label{eq:nonasymptotic_bound_pseudovtau}
	\end{equation}
\end{thm}
Theorem \ref{lem:convergence_lvci_to_vci} derives a non-asymptotic bound for the VCATE as the minimum of two key quantities: (i) the conditional $L_2$ error between the candidate function and the true CATE, and (ii) the true value of the VCATE. This proof only relies on the definition in \eqref{eq:defn_pseudovcate}. For instance, when $V_{\tau}(\gamma) = 0$, then $V_{\tau}(\gamma)-V_{\tau}^*(\gamma,\mathcal{I}_{-nk}) = 0$, regardless of whether $\widehat{\tau}_{-k}(\cdot)$ is properly specified. The difference between the two quantities is also small if $\omega(\gamma)$ is sufficiently close to zero. In the multi-step approach, $\omega(\gamma)$ captures the first-stage uncertainty from estimating the CATE, which decreases with sample size. I consider the following convergence condition.
\begin{assump}[Convergence CATE]
	\label{assump:convergencenuisance}
	$\sqrt{n_k}\omega(\gamma_n)^2 = o(1)$ as $n \to \infty$.	
\end{assump}

Assumption \ref{assump:convergencenuisance} imposes an $L_2$ consistency condition on the CATE.  A large class of machine learning models can meet this requirement. For example, \cite{bickel2009simultaneous} and \cite{belloni2014inference} evaluate rates of convergence under sparse models, \cite{chen1999improved} for neural networks, and \cite{wager2015adaptive} for regression trees and random forest.

\begin{thm}[Faster than $\sqrt{n}$ convergence near boundary] 
	\label{thm:superconsistency}	
	Consider a sequence of data generating processes $\{\gamma_n\}_{n=1}^{\infty}$ where $V_{\tau}(\gamma_n) \to 0$ as $n \to \infty$ and Assumptions \ref{assump:unconfound_sutva},  \ref{assump:momentbounds_dml}, \ref{assump:boundsquantities}, \ref{assump:randomsampling} and \ref{assump:convergencenuisance}  hold. Define $\Delta_{nk} := \left( \widehat{V}_{\tau n k} - V_{\tau}(\gamma_n) \right)$ for the estimator defined in \eqref{eq:vtaunk_dml}.  Then (i) $\sqrt{n_k}\Delta_{nk} = o_p(1)$, and (ii) if in addition $n_k^{1/2+\rho}V_{\tau}(\gamma_n)=o(1) $ for $\rho \in [0,1/2)$, then $n_k^{1/2+\rho}\Delta_{nk} = o_p(1)$.
	
\end{thm}
Theorem \ref{thm:superconsistency} shows that multi-step estimators of the VCATE converge to zero faster than $\sqrt{n_k}$ near the boundary. I formalize ``near'' by considering sequences of distributions where the VCATE approaches zero. Theorem \ref{thm:superconsistency} relies on the non-asymptotic bound in Theorem \ref{lem:convergence_lvci_to_vci}, the normal approximation in Theorem \ref{thm:clt_lvci}, and the empirical process in Lemma \ref{lem:decomposition_remainder_lvci_reg}. There is no requirement on the rate of convergence of $\widehat{\mu}_{-0k}(\cdot)$ (and consequently on the generated regressor $M_{-k_i}(\cdot)$), only an assumption that $p(x)$ is known and that the CATE is estimated at a sufficiently fast rate. Furthermore, if the true CATE is nearly flat in the sense that for $\rho \in [0,1/2)$, then $n_k^{1/2+\rho}V_{\tau}(\gamma_n) = o(1)$ (or even exactly equal to zero), then the estimator has a faster rate guarantee.

To prove efficiency we have the stronger requirement that all the nuisance functions converge to their true value in the $L_4$ norm and at $n_k^{1/4}$ rate in the $L_2$ norm.

\begin{assump}[Regularity conditions]
	\label{assump:regularityconditions}
	Define the residuals $U_i = Y_i - \mathbb{E}_{\gamma_n}[Y_i \mid D_i,X_i]$. (i) $\mathbb{E}_{\gamma_n}[\Vert Y_i \Vert^4]$, $\mathbb{E}_{\gamma_n}[\Vert U_i\Vert^4]$, $\mathbb{E}_{\gamma_n}[\Vert \eta(X_i) \Vert^4]$, (ii) $\mathbb{E}_{\gamma_n}[\Vert \widehat{\eta}_{-k}(X_i) \Vert^4]$ are uniformly bounded, (iii) $\mathbb{E}_{\gamma_n}\left[ \Vert \widehat{\eta}_{-k}(X_) - \eta(X_i) \Vert^4\right] \to 0$, (iv) $\sqrt{n_k}\mathbb{E}_{\gamma_n}[\Vert\widehat{\eta}_{-k}(X_i)-\eta(X_i)\Vert^2] = o(1)$ for all $k \in \{1,\ldots,K\}$.	
\end{assump}

The next step is to show that the estimation error of the fold-specific VCATE converges at $\sqrt{n_k}$ to an average of efficient influence functions.

\begin{thm}[$\sqrt{n}$ Consistency and Efficiency]
	\label{thm:rootn_consistency_and_efficiency}
	Consider a sequence of data generating processes $\{\gamma_n\}_{n=1}^{\infty}$ where $V_{\tau}(\gamma_n) \to 0$ as $n \to \infty$ and Assumptions \ref{assump:unconfound_sutva},  \ref{assump:momentbounds_dml}, \ref{assump:boundsquantities}, \ref{assump:randomsampling}, \ref{assump:convergencenuisance}, and \ref{assump:regularityconditions}  hold. Then
	$$ \sqrt{n_k}(\widehat{V}_{\tau nk}-V_{\tau}(\gamma_n)) = \frac{1}{\sqrt{n_k}}\sum_{i \in \mathcal{I}_{nk}} \varphi_i + o_p(1). $$
\end{thm}
Theorem \ref{thm:rootn_consistency_and_efficiency} shows that the fold-specific estimator converges at $\sqrt{n_k}$-rate to an average of i.i.d influence function. This requires standard regularity conditions. The proof of Theorem \ref{thm:rootn_consistency_and_efficiency} is non-standard due to the multi-step nature of the procedure. I start by applying Lemma \ref{lem:equivalence_efficient_simple}, which shows hows to write $\widehat{V}_{\tau nk}$ as an average of estimated influence functions. I break down the proof into sequences where $V_{\tau}(\gamma_n)$ converges to zero and those where it's bounded away from zero. For the first part, I leverage (a) the boundary convergence result in Theorem \ref{thm:superconsistency}, (b) the bound for $\mathbb{V}(\varphi_i)$ in Lemma \ref{cor:bound_varefficient}. For the second part, I provide a novel decomposition of regression adjusted nuisance functions. The key is to prove that the regression parameters $\widehat{\theta}_{nk}$ converge at $n_k^{1/4}$ rate to the values in Lemma \ref{lem:projection_weak}  for sequences where $V_{\tau}(\gamma_n) \to V_{\tau} > 0$. Once in this form, the rest of the proof relies on a traditional Taylor expansion argument.

The ensemble estimator $\widehat{V}_{\tau n}$ combines information from the whole sample. By definition $n = n_k \times K$ and $K$ is finite, which means that algebraically $\sqrt{n}(\widehat{V}_{\tau n}-V_{\tau}(\gamma_n)) = \frac{\sqrt{n_k}}{\sqrt{K}}\sum_{k=1}^{K}(\widehat{V}_{\tau nk}-V_{\tau}(\gamma_n))$, and by Theorem \ref{thm:rootn_consistency_and_efficiency},
$$ \sqrt{n}(\widehat{V}_{\tau n}-V_{\tau}(\gamma_n)) = \left[ \frac{1}{\sqrt{n_k K}}\sum_{k=1}^K\sum_{i \in \mathcal{I}_{nk}} \varphi_i + 
\frac{1}{\sqrt{K}}\sum_{k=1}^K o_p(1)\right] = \frac{1}{\sqrt{n}}\sum_{i=1}^n \varphi_i + o_p(1). $$
This means that aggregating the estimators restores full efficiency, satisfying the property described in \eqref{eq:efficient_lineardecomposition}.

\subsection{Asymptotic Coverage}

I start by showing that the single fold confidence interval has uniform coverage for the pseudo-VCATE, and exact coverage under an additional assumption.

\begin{assump}[Exact coverage condition]
	\label{assump:exactcoverage}
	Let $\Omega_{nk,12}$ be the off-diagonal element of $\Omega_{nk}$. For each $t > 0$, $\underset{n \to \infty}{\lim\sup}\sup_{\gamma \in \Gamma} \mathbb{P}_{\gamma}(\sqrt{V_{\tau}^*(\gamma,\mathcal{I}_{-nk})}|\Omega_{nk,12}| > t) = 0$. 	
	
\end{assump}


Assumption \ref{assump:exactcoverage} states that the product of the pseudo-VCATE and the off-diagonal element of the limiting covariance matrix in \eqref{eq:omegank} needs to converge to zero uniformly.

\begin{thm}[Uniform Coverage of Pseudo-VCATE]
	\label{thm:exact_coverage_foldlvci}
	Let $\Gamma$ denote a set of distributions, constrained in such a way that Assumptions \ref{assump:unconfound_sutva},  \ref{assump:momentbounds_dml}, \ref{assump:boundsquantities}, and \ref{assump:randomsampling} hold. Let $\widehat{CI}_{\alpha n k} $ and $V_{\tau}^*(\gamma,\mathcal{I}_{-nk})$ be defined as in \eqref{eq:feasible_confidenceset_dml} and \eqref{eq:defn_pseudovcate}, respectively. Then
	\begin{equation}
	\label{eq:lower_uniformcoverage_pseudovcate}
	1-\alpha \le \underset{n \to \infty}{\lim\inf}\inf_{\gamma \in \Gamma} \mathbb{P}_{\gamma}\left( V_{\tau}^*(\gamma,\mathcal{I}_{-nk}) \in \widehat{CI}_{\alpha n k} \right) 
	\end{equation}
	If Assumption \ref{assump:exactcoverage} also holds, then
	\begin{equation}
	\label{eq:upper_uniformcoverage_pseudovcate}
	\underset{n \to \infty}{\lim\sup}\sup_{\gamma \in \Gamma} \mathbb{P}_{\gamma}\left( V_{\tau}^*(\gamma,\mathcal{I}_{-nk}) \in \widehat{CI}_{\alpha n k} \right) \le 1-\alpha.
	\end{equation}
\end{thm}
Theorem \ref{thm:exact_coverage_foldlvci} shows that the confidence intervals always have uniform coverage of the pseudo-VCATE of at least $(1-\alpha)$.\footnote{The theorem only uses Assumptions \ref{assump:unconfound_sutva},  \ref{assump:momentbounds_dml}, \ref{assump:boundsquantities}, and \ref{assump:randomsampling} to verify normality  in Assumption \ref{assump:clt_regcoef}. A broad class of confidence intervals of the form in \eqref{eq:feasible_confidenceset} constructed from regression adjusted estimators will satisfy these uniformity properties.} The key is to prove that the confidence intervals yield coverage under arbitrary sequences of distributions, which includes cases where $V_{\tau n}^*$ is either equal to zero or approaches zero as $n \to \infty$. The proof builds on the approximation of Lemma \ref{lem:decomposition_remainder_lvci_reg} and shows that for every sequence, the test statistic for a particular $\zeta_{n_k} \in \{-1,1\}$ converges to a uniform distribution. This sequential characterization suffices to apply generic results in \cite{andrews2020generic}, which guarantee uniform coverage even in non standard cases like this one. Coverage over the pseudo-VCATE holds regardless of whether the nuisance functions are slow to converge or even misspecified.

The intervals are in general conservative because we're not plugging in the unknown $\zeta$, and instead define a robust confidence interval as the union of CIs with given $\zeta \in \{-1,1\}$. However, the key insight is that $\zeta$ only affects the coverage when the pseudo-VCATE is bounded away from zero. If Assumption \ref{assump:exactcoverage} holds, the value of $\zeta$ doesn't enter the asymptotic distribution of the estimator. I show that this condition holds automatically if the nuisance functions converge to their true value at a sufficiently fast rate.

\begin{lem}[Verify Exact Coverage]
	\label{lem:verifyexactcoverage}
	Let $\Gamma$ denote a set of distributions that satisfy Assumptions \ref{assump:unconfound_sutva},  \ref{assump:momentbounds_dml}, \ref{assump:boundsquantities},  \ref{assump:randomsampling}, \ref{assump:convergencenuisance}, and \ref{assump:regularityconditions}. Then Assumption \ref{assump:exactcoverage} also holds.
\end{lem}
As a special case, when the model is correctly specified, i.e. $W_i'\theta = \mu_d(x)$ for some $\theta \in \mathbb{R}^4$, then $\Omega_{nk,12} = 0$ by construction. Lemma \ref{lem:verifyexactcoverage} states that we only need a model that is correctly specified asymptotically, given the rates in Assumptions \ref{assump:convergencenuisance} and \ref{assump:regularityconditions}. Then for non-boundary cases, $\Omega_{nk}$ converges to the population analog under correct specification. These conditions also imply point-wise coverage of the true VCATE.

\begin{thm}[Pointwise, Exact Coverage of VCATE]
	\label{thm:exact_coverage_foldlvci_pointwise}
	Let $\Gamma$ denote a set of distributions that satisfy Assumptions \ref{assump:unconfound_sutva},  \ref{assump:momentbounds_dml}, \ref{assump:boundsquantities},  \ref{assump:randomsampling}, \ref{assump:convergencenuisance}, and \ref{assump:regularityconditions}. Then
	$$ \inf_{\gamma \in \Gamma}  \underset{n \to \infty}{\lim\inf} \ \mathbb{P}_{\gamma}\left( V_{\tau}(\gamma) \in \widehat{CI}_{\alpha nk} \right) = \sup_{\gamma \in \Gamma} \underset{n \to \infty}{\lim\sup} \ \mathbb{P}_{\gamma}\left( V_{\tau}(\gamma) \in \widehat{CI}_{\alpha nk} \right) = 1-\alpha. $$

\end{thm}

Theorem \ref{thm:exact_coverage_foldlvci_pointwise} shows that if the nuisance functions converge at a sufficiently fast rate, then the proposed intervals achieve point-wise exact coverage. The confidence intervals provide correct size coverage for all regions of the parameter space, including $V_{\tau}(\gamma) = 0$.

Proving uniform coverage of the VCATE (rather than the pseudo-VCATE) is more challenging in the non-parametric case without much stronger conditions on the convergence rates of the nuisance functions. The lack of uniformity stems from a difficulty in controlling the ratio $\sqrt{n_k}\omega(\gamma_n) / \sqrt{V_{\tau}(\gamma_n)}$, which measures the relative error in estimating the CATE vs. the overall level of the VCATE. By the bound in \eqref{lem:convergence_lvci_to_vci}, this ratio is easy to control when $n_kV_{\tau}(\gamma_n) = o(1)$ (near homogeneity) or $V_{\tau}(\gamma_n) \to V_\tau > 0$ (strong heterogeneity). However, it is possible to construct sequences, e.g., $n_kV_{\tau}(\gamma_n) \to v > 0$,  where $(\widehat{V}_{\tau nk} - V_{\tau}^*(\gamma_n,\mathcal{I}_{-nk}))$ converges to zero at a faster or comparable rate to the error of the pseudo-VCATE. There may be distortions in coverage in smaller samples. I illustrate this issue in the simulations.

\begin{rem}[Uniform inference for one-sided tests] \label{rem:onesidedtests} Uniform inference is only challenging for two-sided tests. If instead, the researcher is only interested in left-sided tests, then uniform inference is still possible. To do so, we can make explicit use of the inequality $V_{\tau}^*(\gamma,\mathcal{I}_{-nk}) \le V_{\tau}(\gamma)$.  If $V_{\tau}(\gamma) < \inf_{V_{\tau}^*} \widehat{CI}_{\alpha nk}$ (the lower bound of the CI), then $V_{\tau}^*(\gamma,\mathcal{I}_{-nk}) \notin \widehat{CI}_{\alpha nk}$. Therefore, for all $\gamma \in \Gamma$,
\begin{equation}
	\label{eq:boundprob_lowerci}
	\mathbb{P}_{\gamma}\left( V_{\tau}(\gamma) \ge \inf \widehat{CI}_{\alpha n k} \right) \le \mathbb{P}_{\gamma}\left( V_{\tau}^*(\gamma,\mathcal{I}_{-nk}) \notin \widehat{CI}_{\alpha n k} \right).
\end{equation}
I prove a weaker uniformity result for one-sided tests building on Theorem \ref{thm:exact_coverage_foldlvci}.
\begin{cor}
	\label{cor:uniformity low}
	If Assumptions \ref{assump:unconfound_sutva},  \ref{assump:momentbounds_dml}, \ref{assump:boundsquantities}, and \ref{assump:randomsampling} hold, then
	$$ \underset{n \to \infty}{\lim\sup}\sup_{\gamma \in \Gamma} \mathbb{P}_{\gamma}\left( V_{\tau}(\gamma) < \inf \widehat{CI}_{\alpha n k} \right) \le \underset{n \to \infty}{\lim\sup}\sup_{\gamma \in \Gamma} \mathbb{P}_{\gamma}\left( V_{\tau}^*(\gamma,\mathcal{I}_{-nk}) \notin \widehat{CI}_{\alpha n k} \right) \le \alpha. $$	
\end{cor}
\end{rem}
Corollary \ref{cor:uniformity low} is empirically relevant for interpreting confidence intervals that do not include zero. It states that the asymptotic probability of having $V_{\tau}(\gamma) \in \left[0,\inf \widehat{CI}_{\alpha nk}\right)$ is uniformly less than $\alpha$. Tests of homogeneity belong to this class and therefore have the correct size when $V_\tau(\gamma) = 0$. Moreover, the result in Corollary \ref{cor:uniformity low} is much stronger because it guarantees that a broader class of one-sided tests also has the correct size. It is important to emphasize that I do not impose any assumptions on rates of convergence of $(\widehat{\eta}_{-k}(x)-\eta(x))$, but only the inequality on the pseudo-VCATE. Consequently, while estimating $\mu(x)$ and $\tau(x)$ may be important for increasing the power of tests of homogeneity, it is not necessary for controlling their size.

\subsection{Multifold Coverage}

The multi-fold confidence interval covers the VCATE asymptotically.

\begin{thm}
	\label{thm:uniform_coverage_variational_lvci}
	Let $\Gamma$ be a set of distributions that satisfy Assumptions \ref{assump:unconfound_sutva},  \ref{assump:momentbounds_dml}, \ref{assump:boundsquantities},  \ref{assump:randomsampling}. Then
	$$ \underset{n \to \infty}{\lim\sup}\sup_{\gamma \in \Gamma} \mathbb{P}_{\gamma}\left( V_{\tau}(\gamma) < \inf \widehat{CI}_{\alpha n k}^{multifold} \right) \le \alpha. $$	
	If Assumptions \ref{assump:convergencenuisance}, and \ref{assump:regularityconditions} also hold, then
	$$	\sup_{\gamma \in \Gamma} \underset{n \to \infty}{\lim\sup} \ \mathbb{P}_{\gamma}\left(V_{\tau}(\gamma) \notin \widehat{CI}_{\alpha n}^{\text{multifold}}\right)  \le \alpha.$$

\end{thm}
The first part of Theorem \ref{thm:uniform_coverage_variational_lvci} shows that the multifold CI uniformly controls the size of one-sided tests. The second part shows that if the nuisance functions converge to their true value asymptotically, then the multifold confidence interval provides point-wise size-control for two-sided tests. Coverage of the true parameter will be weakly larger that $(1-\alpha)$ asymptotically.

\subsection{Power}

The test of homogeneity has power against local alternatives.
\begin{lem}
	\label{lem:powerlocal_alternatives}
	Consider a sequence of distributions $\{\gamma_n\}_{n=1}^\infty$ and $\{\mathcal{I}_{-nk}\}_{n=1}^{\infty}$, where $\Omega_{nk} \to \Omega_{\infty}$ and $n_kV_{\tau}^*(\gamma_n,\mathcal{I}_{-nk}) = v + o(1)$, for $v \in [0,\infty)$. Assume that \ref{assump:unconfound_sutva},  \ref{assump:momentbounds_dml}, \ref{assump:boundsquantities}, and \ref{assump:randomsampling} hold. Let $\Omega_{\infty,11}$ be the upper-left entry of $\Omega_{\infty}$, $\Phi(\cdot)$ be the standard normal CDF, and $z_{1-\alpha}$ be the $(1-\alpha)-$quantile. Then
	$$ \lim_{n \to \infty} \mathbb{P}_{\gamma_n}(0 \notin \widehat{CI}_{\alpha nk}\mid \mathcal{I}_{-nk}) = 1-\Phi\left(z_{1-\alpha}-\frac{\sqrt{v}}{\sqrt{\Omega_{\infty,11}}}\right) + \Phi\left( -\frac{\sqrt{v}}{\sqrt{\Omega_{\infty,11}}} \right). $$
\end{lem}
Lemma \ref{lem:powerlocal_alternatives} computes the power curve for a sequence of local alternatives. When $v = 0$ the power is equal to $\alpha$, whereas when $v \to \infty$ the power tends to one. This shows that tests of homogeneity have local power   the null. When the pseudo-VCATE is bounded away from zero, the test rejects with probability approaching one.

\section{Extensions}

\begin{rem}[Clustered Standard Errors]
		\label{rem:clustered_ses}
	In some cases, assuming that units $i$ are independent may be strong. For example, in \cite{dizon2019parents} units are randomized at the household level, and it is reasonable to expects that units within a household have correlated outcomes and covariates. To deal with this dependence structure, suppose that the sample can be partitioned into $C$ clusters, $c \in \{1,\ldots,C\}$, which are independent and identically distributed. The researcher can compute $\widehat{\beta}_{2nk}$, $\widehat{V}_{\tau nk}$ and $\widehat{CI}_{\alpha nk}$ via cross-fitting by randomly partitioning entire clusters rather than the individual observations. 
	\begin{lem}
		\label{lem:clustered_se} Let $\{r_{nk}\}_{n=1}^{\infty}$ be a sequence of positive scalars. Suppose that $V_{xnk} > 0$, $\Omega_{nk}$ is positive definite with positive eigenvalues, and that conditional on a sequence $\{\mathcal{I}_{-nk}\}_{n=1}^{\infty}$, $r_{nk}^{-1}\widehat{\Omega}_{nk} \to^p \Omega_{nk}$, $n_k/r_{nk} \to \infty$, and
		\begin{equation}
		\label{eq:clt_coefficients_cluster} \Omega_{nk}^{-1/2}\sqrt{\frac{n_k}{r_{nk}}}\begin{pmatrix} \sqrt{V_{xn}}(\widehat{\beta}_{2nk} - \beta_{2nk}) \\ \frac{\widehat{V}_{xnk}}{V_{xnk}} - 1 \end{pmatrix} \mid \mathcal{I}_{-nk} \to^d Z_n \sim \mathcal{N}(0,I_{2 \times 2}).
		\end{equation}	
		Then $\widehat{CI}_{\alpha n k}$, substituting the arguments $(n_k,\widehat{V}_{\tau n k},\widehat{\Omega}_{nk})$, satisfies Theorem \ref{thm:exact_coverage_foldlvci}.
	\end{lem}
	
	Lemma \ref{lem:clustered_se} proposes high-level conditions that ensure that confidence intervals have correct coverage. The quantity $\sqrt{n_k/r_{nk}}$ is the effective rate of convergence, which features prominently in problems with cluster dependence \citep{mackinnon2022cluster}. For example, if the observations are fully correlated within clusters and the clusters have equal size, then $r_{nk}$ is the cluster size, $n_k/r_{nk} = C$, and the estimators in \eqref{lem:clustered_se} converge at $\sqrt{C}$ rate (the total number of clusters). The analyst does not need to specify the quantity $r_{nk}$ to apply the procedure, but merely specify an estimator of the covariance matrix that meets the rate requirement. Under minor modifications to the existing proofs, we can also prove analogs of Theorems \ref{thm:exact_coverage_foldlvci_pointwise} and \ref{thm:uniform_coverage_variational_lvci}.
	
	 We can construct estimators that satisfy Lemma \ref{lem:clustered_se}. Let $\mathcal{I}_{nkc}$ be the set of units in fold $k$ and cluster $c$, and let $\mathcal{C}_{nk}$ be the indexes of the clusters selected for fold $k$.
	$$ \widehat{H}_{nk}^{cluster} := \frac{1}{n_k}\sum_{c \in \mathcal{C}_{nk}}\left( \sum_{i \in \mathcal{I}_{nkc}} \begin{bmatrix}  \lambda(X_i)\widehat{U}_i\widehat{W}_{i} \\ \widehat{T}_{i} \end{bmatrix} \right)\left( \sum_{i \in \mathcal{I}_{nkc}} \begin{bmatrix} \lambda(X_i)\widehat{U}_i\widehat{W}_{i} \\ \widehat{T}_{i} \end{bmatrix} \right)'.$$
	The clustered standard errors are $\widehat{\Omega}_{nk}^{cluster} =   \widehat{\Upsilon}_{nk}\widehat{J}_{nk}^{-1}\widehat{H}_{nk}^{cluster}\widehat{J}_{nk}^{-1}\widehat{\Upsilon}_{nk}'$, where $\widehat{J}_{nk},\widehat{\Upsilon}_{nk}$ are computed as outlined in \eqref{eq:Jmatrix_sandwich_dml}.

\end{rem}

\begin{rem}[Confidence intervals when $V_{xnk} = 0$]
	\label{rem:ci_degenerate}
	When the conditional mean is constant, i.e. $\mu_d(x) = \mathbb{E}[Y_d]$, prediction models with corner solutions like LASSO may estimate a constant conditional mean, i.e.  $\widehat{\mu}_{d,-k}(x) = \widehat{\mu}_{d,av}$, $\widehat{\tau}_{-k}(x) = \widehat{\mu}_{1,-k}(x)  - \widehat{\mu}_{0,-k}(x) $ is constant, and consequently $V_{xnk} = 0$.\footnote{It is still possible to have $V_{xnk} > 0$ almost surely even if $V_{\tau} = 0$, as long as $\mu_0(x)$ is not constant.} This violates Assumption \ref{assump:boundsquantities}.(i), and it is challenging to construct a confidence interval with exact coverage. One alternative is to construct an ensemble of sparse and non-sparse estimators of the CATE in the first-stage. Another alternative is to use degenerate confidence intervals:
	\begin{equation}
	\label{eq:degenerate_ci}
	 \widehat{CI}_{\alpha n k}^{0} = \begin{cases} \widehat{CI}_{\alpha n k} & \text{if }V_{xnk} \ne 0, \\ [0,0] &\text{if }V_{xnk} = 0.  \end{cases}
	\end{equation}
	The confidence intervals collapse to zero when the $\widehat{\tau}_{-k}(x)$ prediction is degenerate. For example, in LASSO researchers can check whether the coefficients are zero, in tree-based methods when there are no splits, or whether $\widehat{V}_{xnk} = 0$. We can also define an analogous multifold confidence interval.
	\begin{equation}
	\widehat{CI}_{\alpha n}^{\text{0,multifold}} = \left[ \text{Med}_{K}\left\{ \inf \widehat{CI}_{\frac{\alpha}{2} nk} \right\}, \text{Med}_{K}\left\{ \sup \widehat{CI}_{\frac{\alpha}{2} nk} \right\} \right].
	\label{eq:multisplit_cis_degenerate}
	\end{equation}	
	I study the asymptotic properties of these confidence intervals.
	\begin{lem}
		\label{lem:degenerate_cate}
		Let $\Gamma$ denote a set of distributions that satisfy Assumptions \ref{assump:unconfound_sutva},  \ref{assump:momentbounds_dml}, and \ref{assump:randomsampling}. Suppose that Assumption \ref{assump:boundsquantities} holds, except for the requirement that $V_{xnk} = 0$. Then (i)
		\begin{equation}
		\label{eq:uniformconservativecoverage_degenerate}
		\underset{n \to \infty}{\lim\inf}\inf_{\gamma \in \Gamma} \ \mathbb{P}_{\gamma}\left( V_{\tau}^*(\gamma,\mathcal{I}_{-nk}) \in \widehat{CI}_{\alpha n k}^{0} \right) \ge 1-\alpha 
		\end{equation}
		(ii) If Assumptions \ref{assump:convergencenuisance}, and \ref{assump:regularityconditions} also hold, then
		\begin{equation} \inf_{\gamma \in \Gamma} \underset{n \to \infty}{\lim\inf} \ \mathbb{P}_{\gamma}\left( V_{\tau}(\gamma) \in \widehat{CI}_{\alpha n k}^{0} \right) \ge 1-\alpha
		\label{eq:degenerate_pointwise}
		\end{equation}
		\begin{equation}
		\label{eq:multifold_zero_ci}
		\inf_{\gamma \in \Gamma} \underset{n \to \infty}{\lim\inf} \ \mathbb{P}_{\gamma}\left( V_{\tau}(\gamma) \in \widehat{CI}_{\alpha n k}^{0,multifold} \right) \ge 1-\alpha.
		\end{equation}	
	\end{lem}
	To prove this result I focus on the coverage for subsequences where $V_{xnk} = 0$ and $V_{xnk} > 0$, and apply the results for conservative coverage results in \cite{andrews2020generic}. In subsequences where $V_{xnk} = 0$, then $V_{\tau}^*(\gamma,\mathcal{I}_{-nk}) = 0$ which means that coverage of the pseudo-VCATE is equal to one. In subsequences where $V_{xnk} > 0$ and assuming that $\Omega_{nk}$ has eigenvalues bounded away from zero,  then we can apply similar arguments as before to prove $1-\alpha$ coverage. To prove point-wise coverage, I separate the cases where $V_{\tau}(\gamma) = 0$ and $V_{\tau}(\gamma) > 0$. In the latter case, I show that $V_{xnk}$ is point-wise bounded away from zero, though not uniformly. The proof of Lemma \ref{lem:degenerate_cate} does not rely on the i.i.d. assumption, and can also accommodate cluster dependence. In the empirical example, I compute confidence intervals with clustered standard errors and degenerate CATE predictions.
	
	When there's more heterogeneity and the nuisance functions are estimated accurately, then $V_{xnk} > 0$ with high probability. However, when $V_{\tau n} \approx 0$ and $\mathbb{V}(\mu_0(X)) \approx 0$, then procedures like LASSO may imply $V_{xnk} = 0$ \citep{fu2000asymptotics}, which means that marginally heterogeneous CATEs could be estimated as homogeneous. This could be impact the power of tests of homogeneity. The size for two-sided tests is not uniformly bounded. Furthermore, the multifold confidence interval allows for some quantification of uncertainty across folds/splits: the CI is degenerate only if more than half the fold/split-specific CIs are degenerate. 
	
	Furthermore, the degenerate CI has correct size control for one-sided tests.	
	
	\begin{cor}
	\label{cor:one_sidedtests_degenerate}
	Under the assumptions of Lemma \ref{lem:degenerate_cate}.(i),
	\begin{equation}\label{eq:onesided_degenerate_singlefold} \underset{n \to \infty}{\lim\sup}\sup_{\gamma \in \Gamma} \mathbb{P}_{\gamma}\left( V_{\tau}(\gamma) < \inf \widehat{CI}_{\alpha n k}^{0} \right) \le \alpha.
	\end{equation}
	\begin{equation}\label{eq:onesided_degenerate_multifold} \underset{n \to \infty}{\lim\sup}\sup_{\gamma \in \Gamma} \mathbb{P}_{\gamma}\left( V_{\tau}(\gamma) < \inf \widehat{CI}_{\alpha n k}^{0,multifold} \right) \le \alpha.
	\end{equation}	
	\end{cor}
    The tests of homogeneity have the correct size when $V_\tau(\gamma) = 0$. Corollary \ref{cor:one_sidedtests_degenerate} guarantees that the probability of falsely rejecting a class of one-sided test is uniformly bounded in large samples.
	
	
\end{rem}

\begin{rem}[Monotonic Transformations]
	\label{rem:monotonic}
	It may be useful to report the standard deviation of the CATE, which is  $\sqrt{VCATE}$. I propose the following confidence interval:
	\begin{equation} \widehat{CI}_{\alpha n}^{0,multifold,sqrt}  = \left\{\sqrt{V_{\tau}^*}: V_{\tau}^* \in \widehat{CI}_{\alpha n}^{0,multifold}\right\}.
	\label{eq:sqrt_transform_ci}
	\end{equation}
	Since the square root is a strictly increasing transformation and the VCATE is non-negative, then $\sqrt{V_{\tau}(\gamma)} \in \widehat{CI}_{\alpha n}^{0,multifold,sqrt} $ if and only $V_{\tau}(\gamma) \in \widehat{CI}_{\alpha n}^{0,multifold}$. Since the events are equivalent, the transformed confidence interval preserves the coverage probabilities and will have valid coverage by Lemma \ref{lem:degenerate_cate}.

\end{rem}

\section{Simulations}
\label{simulations}

I use a simulation design to study the properties of the VCATE estimators. The baseline covariates are distributed as $[X_0,X_1] \in \mathcal{N}(0,\Sigma_x)$, where $\rho = 0.5$ and
$$ \Sigma_x = 
\begin{bmatrix} I_{J \times J} & \rho I_{J \times J} \\
\rho  I_{J \times J} & I_{J \times J} \end{bmatrix}.$$ 
The random variables $X_0$ and $X_1$ are standard normal vectors of dimension $J$. The covariance between pairs of components $X_{0j}$ and $X_{1j'}$ is equal to $\rho = 0.5$ for when $j = j'$, but zero otherwise. The outcome is generated from a model where $Y = DY_0 + (1-D)$, $D$ is generated by a Bernoulli draw with probability $0.5$, and
\begin{align}
	\begin{split}
	Y_0 &= c + \beta_{0}'X_0 + U_0 \sqrt{\tilde{\sigma}_0^2 + \kappa_{0}'X_0X_0'\kappa_{0}} \\
	Y_1 &= (c + \tau) + \beta_{0}'X_0 +\beta_{\tau}'X_1 + U_1 \sqrt{\tilde{\sigma}_1^2 + \kappa_{1}'X_1X_1'\kappa_{1}},
	\end{split}
	\label{eq:outcome_simulation}
\end{align}
where $c,\tau \in \mathbb{R}$, $\beta_0,\beta_\tau,\kappa_0,\kappa_1 \in \mathbb{R}^p$. The errors $(U_0,U_1)$ are independent of the covariates $(U_0,U_1) \indep (X_0,X_1)$, and distributed as standard normals $[U_0,U_1]' \in \mathcal{N}(0_{2 \times 1},I_2)$. The key model quantities have closed-form expressions. The conditional means at baseline and the CATE are given by $\mu_1(x) = \alpha + \beta_0'x_0$ and $\tau(x) = \tau + \beta_\tau'Z_1$, respectively. The conditional variances are $\sigma_d^2(x) = \kappa_d'x_dx_d'\kappa_d$ for $d \in \{0,1\}$. This formulation incorporates heteroskedasticity. Covariates that influence the outcomes at baseline may also affect the treatment effects.

The regressors are constructed in such a way that $\mathbb{E}[X_dX_d'] = I_p$ for $d \in \{0,1\}$. This implies simple expressions for the variances of the model, $\mathbb{V}(U_d) = \tilde{\sigma}_d^2 + \kappa_d'\kappa_d$,
$$ V_{\tau} = \beta_\tau'\beta_\tau, \qquad \mathbb{V}(Y_0) = \beta_0'\beta_0 + \tilde{\sigma}_d^2 + \kappa_0'\kappa_0, $$
 $$ \mathbb{V}(Y_1) = \beta_0'\beta_0 + \beta_\tau'\beta_\tau + 2(1-\rho)\beta_0'\beta_\tau + \tilde{\sigma}_1^2 + \kappa_1'\kappa_1, $$

I choose an approximately sparse specification for \eqref{eq:outcome_simulation} where the coefficients decay exponentially at a rate of decay of $\lambda = 0.7$. Let $\ell_j = \sqrt{\left(\frac{1-\lambda}{1-\lambda^J}\right)\lambda^{1-j}}$ be a geometric sequence, which  satisfies $\sum_{j=1}^J \left(\frac{1-\lambda}{1-\lambda^J}\right)\lambda^{1-j} = 1$. Given user-specified parameters $(V_{\mu},V_{\tau},\sigma_0^2,\sigma_1^2)$, the coefficients for the entries $j \in \{1,\ldots,J \}$ are determined by $\beta_{0,j} = \ell_j\sqrt{V_{\mu}}$, $\beta_{\tau,j} = \ell_j \sqrt{V_\tau}$, $\kappa_{d,j} = \ell_j \sqrt{\sigma_d^2-\tilde{\sigma}_d^2}$, for $d \in \{0,1\}$. Since $\sum_{j=1}^J \left(\frac{1-\lambda}{1-\lambda^J}\right)\lambda^{1-j} = 1$, then $\beta_0'\beta_0 = V_{\mu}$, $\beta_\tau\beta_\tau = V_{\tau}$, and $\beta_0'\beta_\tau = \sqrt{V_\mu V_\tau}$. We can obtain analogous expressions for the variances of the unobserved components, so that $\tilde{\sigma}_d^2 + \kappa_0'\kappa_0 = \sigma_d^2$ for $d \in \{0,1\}$.

I choose an average effect size of $\tau = 0.15$, that is coherent with the recent meta-analyses of economic experiments in \cite{vivalt2015heterogeneous}. To make sure that the magnitudes are interpretable, I normalize the coefficients so that the variance for the control group is $\mathbb{V}(Y_0) = 1$, by setting set $c = 1$, $\sigma_d = 0.7$,  $\tilde{\sigma}_d = 0.21$, and $V_{\mu} = 0.3$. The design is easy to scale for different values of $V_{\tau}$ and $J$. My design is similar to that in \cite{belloni2014inference} but I choose $\Sigma_x$ and the sparsity structure in such a way that $V_{\tau}$ has a closed form expression. I use LASSO to estimate $\mu_1(x)$ and $\mu_0(x)$, tuned via cross-validation. The coefficients of this model are consistent given this sparse linear structure, even in high dimensions. I randomly simulate 2000 datasets to compute each of the estimators, and split them into $K=2$ folds.

\begin{figure}[t]
	\centering
	\includegraphics[scale=0.45]{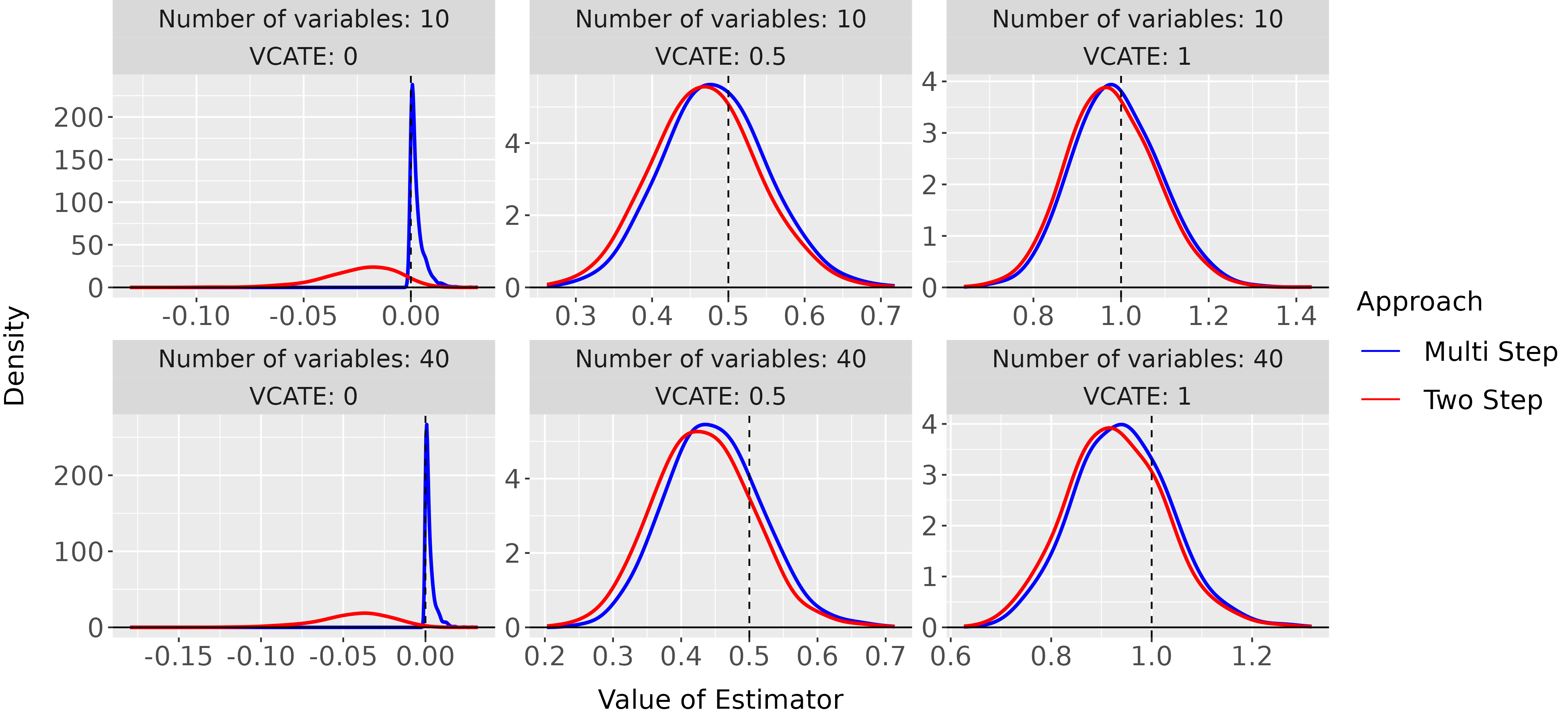}	
	\caption[Density]{\scriptsize \textbf{(Density of estimators)} The figure shows the distributions of $\widehat{V}_{\tau n}$ (multi-step) and $\widehat{V}_{\tau n}^{\text{two-step}}$, defined in \eqref{eq:multistep_estimator} and \eqref{eq:twostep_dml}, respectively, for $n = 2500$, $K = 2$ folds, and a single split. The horizontal panels show designs with homogeneity ($V_{\tau} = 0$, left), moderate heterogeneity ($V_{\tau} = 0.5$, middle) and high-heterogeneity ($V_{\tau} = 1$, right). The vertical panels have a low dimension case with $2J = 10$ (top), and a high-dimension case with $2J = 40$ (bottom).  }
	\label{fig:densityplots_vtauhat}		
\end{figure}

Figure \ref{fig:densityplots_vtauhat} considers a simulation with $n = 2500$. The figure displays a density plot for the multi-step estimator, $\widehat{V}_{\tau n}$ defined in \eqref{eq:multistep_estimator}, and a two-step debiased machine learning estimator computed as:
\begin{equation} \widehat{V}_{\tau n}^{\text{two-step}} := \frac{1}{n}\sum_{i =1}^{n}\varphi(Y_i,D_i,X_i,\widehat{\eta}_{-k_i}(X_i)),
\label{eq:twostep_dml}
\end{equation}
where $\widehat{\eta}_{-k}(x) = (\widehat{\tau}_{-k}(x),\widehat{\mu}_{0,-k}(x),p(x),\widehat{\tau}_{n,av})$ and $\widehat{\tau}_{n,av} = \frac{1}{n}\sum_{i=1}^n \widehat{\tau}_{-k_i}(X_i)$. When $V_{\tau} > 0$ the efficient influence function is non-degenerate. In high-heterogeneity regimes both converge to the same limiting distribution.\footnote{This is shown in Theorem \ref{thm:rootn_consistency_and_efficiency} for the multistep approach and can be shown for the two-step using standard arguments, e.g. \cite{chernozhukov2018double}.} However, when $V_{\tau} = 0$, the influence function is degenerate and they may converge at different rates. We see that the multi-step approach is much more precise. This can be explained by the fast boundary convergence rates derived in Theorem \ref{thm:superconsistency}. The two step approach can also produce negative estimates of $V_{\tau}$, which is an undesirable feature, whereas the multi-step estimator is always non-negative. Both estimators have higher bias when the dimension increases because there is more first-stage noise.

\begin{figure}[t]
	\centering
	\includegraphics[scale=0.45]{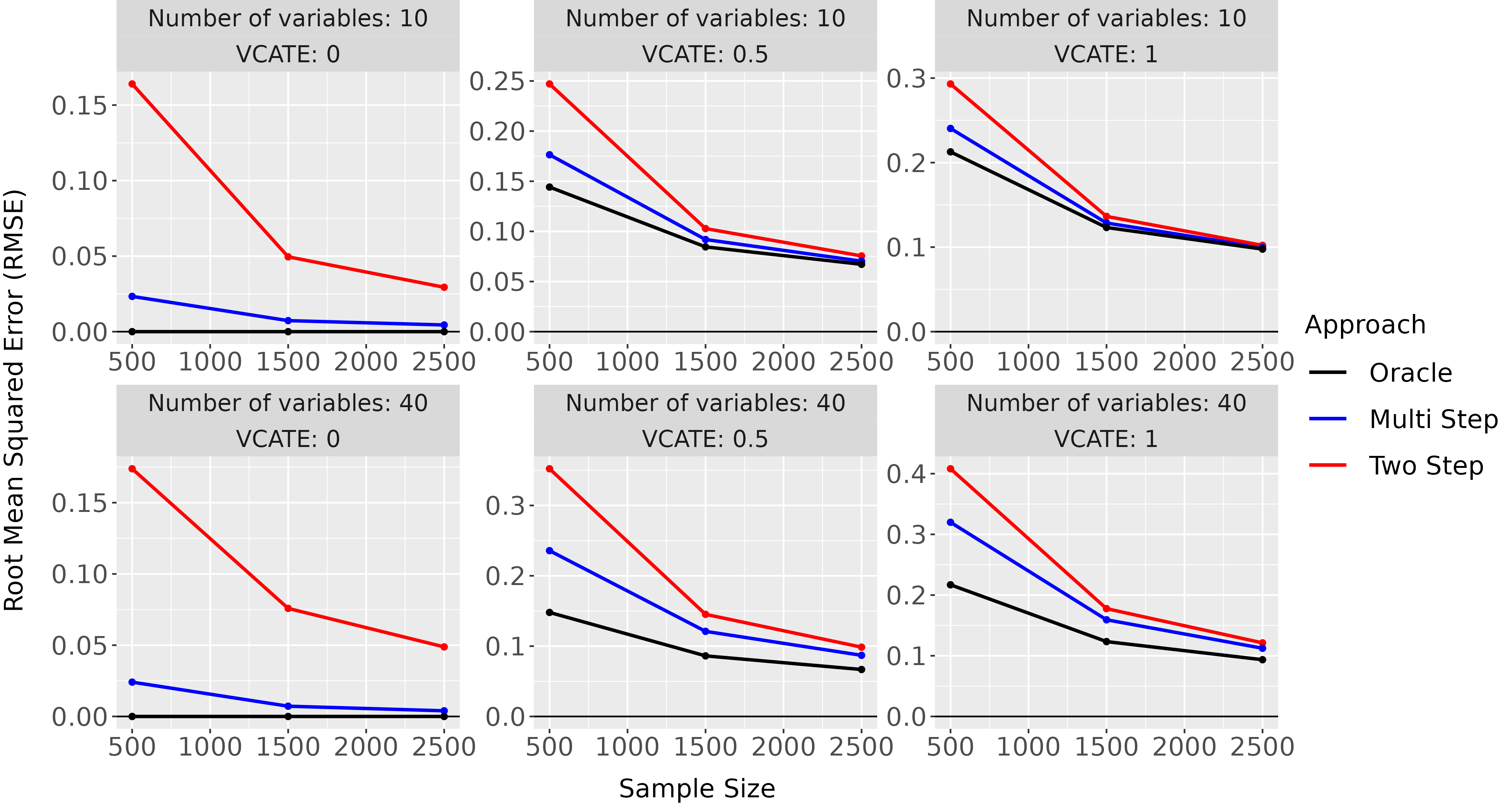}
	\caption{\scriptsize \textbf{(Root Mean Square Error)} The figure shows the root mean-square error of $\widehat{V}_{\tau n}$ (multi-step) and $\widehat{V}_{\tau n}^{\text{two-step}}$, defined in \eqref{eq:multistep_estimator} and \eqref{eq:twostep_dml}, respectively, for different sample sizes, $K=2$ folds, and a single split. The `oracle'' estimator is constructed by substituting $\widehat{\eta}_{-k}(X_i) = \eta(X_i)$ in \eqref{eq:twostep_dml}.  The horizontal panels show designs with homogeneity ($V_{\tau} = 0$, left), moderate heterogeneity ($V_{\tau} = 0.5$, middle) and high-heterogeneity ($V_{\tau} = 1$, right). The vertical panels have a low dimension case with $2J = 10$ (top), and a high-dimension case with $2J = 40$ (bottom).  }
	\label{fig:mse_by_vtau}	
\end{figure}
Figure \ref{fig:mse_by_vtau} plots the root mean-square error (RMSE) of $\widehat{V}_{\tau n}$ and $\widehat{V}_{\tau n}^{two-step}$ for different sample sizes. I compute the semiparametric efficiency bound by computing the RMSE of an ``oracle'' estimator that substitutes $\widehat{\eta}_{-k}(X_i) = \eta(X_i)$ in \eqref{eq:twostep_dml}. The results show that as the sample size increases, both estimators achieve a higher level of accuracy and their variance approaches the semi-parametric lower bound (the RMSE of the oracle). As expected by Corollary \ref{cor:bound_varefficient}, the semiparametric lower bound is zero at the boundary. The differences in RMSE shorten with higher $V_{\tau}$ and in lower dimensional settings $(2J = 10)$ (which have lower first-stage noise).

\begin{figure}[t]
	\centering
	\includegraphics[scale=0.45]{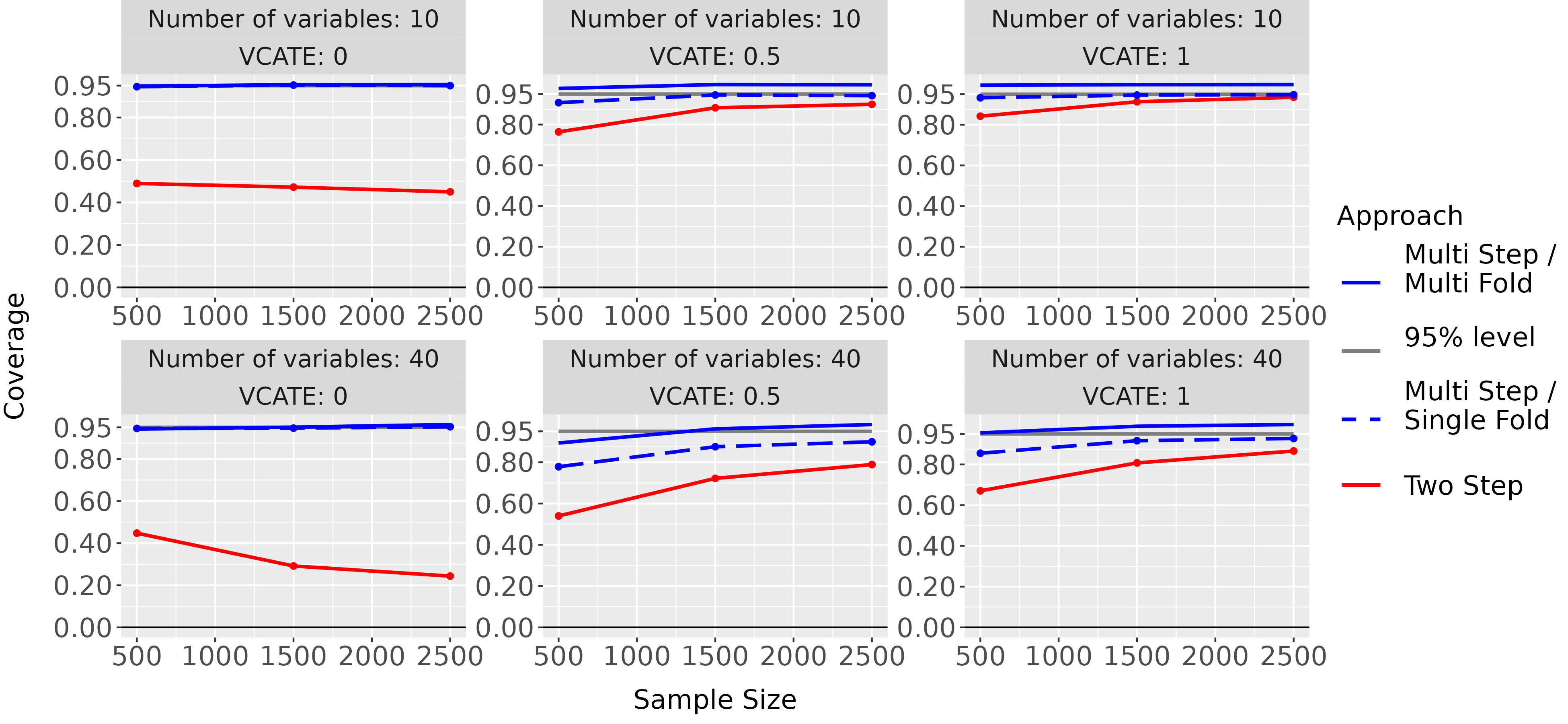}	
	\caption{\scriptsize \textbf{(Coverage of $V_{\tau}$)} The figure shows the 95\% coverage probabilities at different sample sizes, for the multi-fold CIs and single fold confidence intervals defined in \eqref{eq:degenerate_ci} and \eqref{eq:multisplit_cis_degenerate}, respectively, for $K = 2$ folds, single split. The two-step CIs are constructed as $(1/n)\sum_{i=1}^n \widehat{\varphi}_i \pm 1.96\sqrt{\smash[b]{\widehat{V}_{\varphi}/n}}$, where $\widehat{\varphi}_i$ is the summand  in \eqref{eq:twostep_dml} and $\widehat{V}_{\varphi}$ is an estimate of its sample variance. The horizontal panels show a regime with homogeneity ($V_{\tau} = 0$), moderate heterogeneity $(V_{\tau} = 0.5)$ and high-heterogeneity $(V_{\tau} = 1)$. The vertical panels show low dimension case with $2J = 10$ (top), a high-dimension case with $2J = 40$ (bottom). The dotted vertical lines denote the true value of the VCATE. 
		}
	\label{fig:coverage_by_vtau}		
\end{figure}

Figure \ref{fig:coverage_by_vtau} shows the coverage of $V_{\tau}$ for the different proposed confidence intervals (CIs). For the multi-step approach, I consider the single splits CIs in \eqref{eq:degenerate_ci} and the conservative multi-fold CIs from  and \eqref{eq:multisplit_cis_degenerate}. The two-step CIs are constructed as $\frac{1}{n}\sum_{i=1}^n \widehat{\varphi}_i \pm 1.96\sqrt{\widehat{V}_{\varphi}/n}$, where $\widehat{\varphi}_i$ is the summand  in \eqref{eq:twostep_dml} and $\widehat{V}_{\varphi}$ is an estimate of its sample variance. The coverage of the two-step approach is very low under homogeneity, and there is no improvement as sample size increases when $V_{\tau} = 0$. The coverage of the two-step estimator only improves with higher $n$, in high heterogeneity designs. By contrast, both multi-step approaches cover the parameter at the intended level, and coverage improves with higher sample size. For fixed $n$, coverage degrades for both cases when the number of covariates is higher.

Figure \ref{fig:localcoverage_by_vtau} explores the differences in covering the VCATE vs the pseudo-VCATE when $n = 2500$ and $2J = 10$ for a fine-grained set of values of $V_{\tau}$. Panel (a) reflects a dip in coverage close to the boundary. My theory predicts that the multistep CIs have exact coverage when $V_{\tau} = 0$, but may not cover uniformly close to the boundary (see discussion after Theorem \ref{thm:exact_coverage_foldlvci_pointwise}). The mulit-fold CIs have conservative coverage. Conversely, Figure \ref{fig:localcoverage_by_vtau},  Panel (b) shows the multi-step CIs always uniformly cover the pseudo-VCATE, as predicted by theory. This provides a robustness guarantee for how to interpret the CIs. The two-step approach has much lower coverage and no guarantees when $V_{\tau}  = 0$ in either panel.

\begin{figure}[!ht]
	\centering
	\includegraphics[scale=0.45]{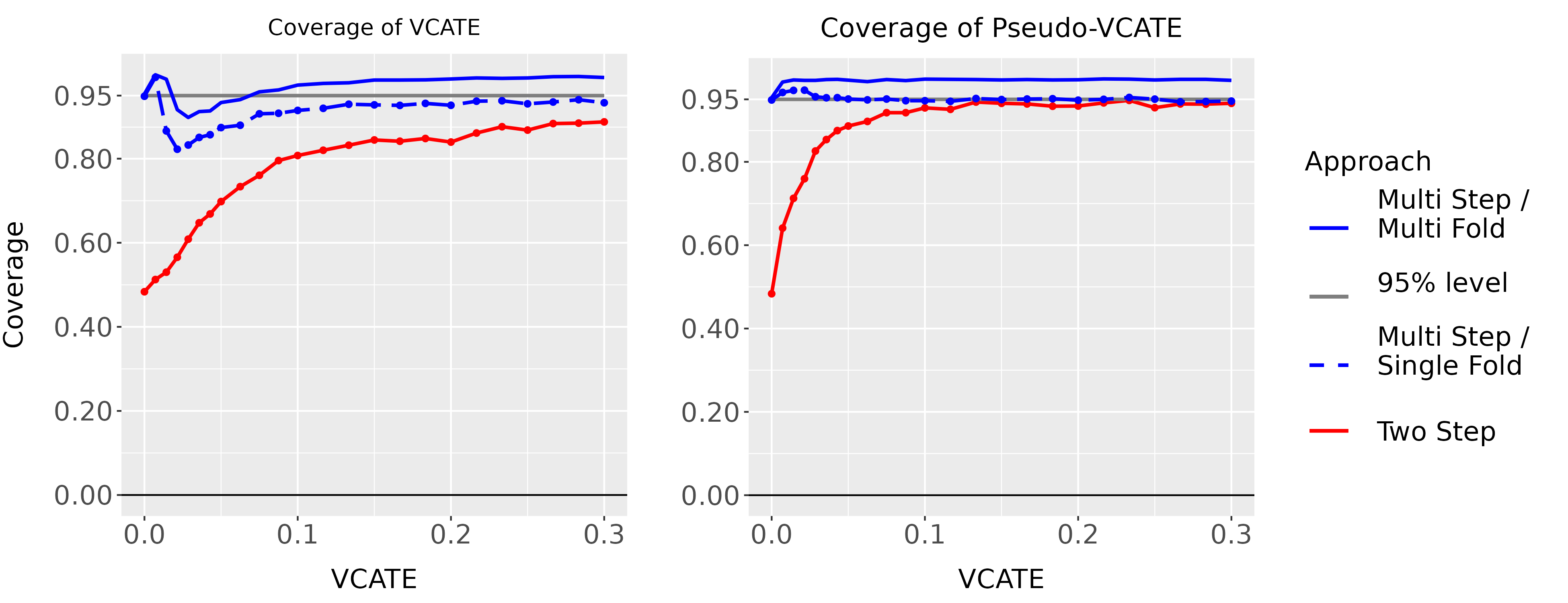}	
	\caption{\scriptsize \textbf{(Coverage of $V_{\tau}$ vs $V_\tau^*$)}	The figure shows the coverage of the VCATE (left) and the fold-specific pseudo-VCATE (right) for a set of fine-grained values of the true $V_{\tau}$, given $n=2500$, $K = 2$, single split, and $2J = 10$. The multi-fold CIs and single fold CIS, defined in \eqref{eq:degenerate_ci} and \eqref{eq:multisplit_cis_degenerate}, respectively. The two-step CIs are constructed as $(1/n)\sum_{i=1}^n \widehat{\varphi}_i \pm 1.96\sqrt{\smash[b]{\widehat{V}_{\varphi}/n}}$, where $\widehat{\varphi}_i$ is the summand  in \eqref{eq:twostep_dml} and $\widehat{V}_{\varphi}$ is an estimate of its sample variance. For the multi-fold and two-step approaches in the right panel, I report  coverage of the median $V_{\tau nk}^*$ across folds.}
	\label{fig:localcoverage_by_vtau}		
\end{figure}

\begin{figure}[h]
	\centering
	\includegraphics[scale=0.45]{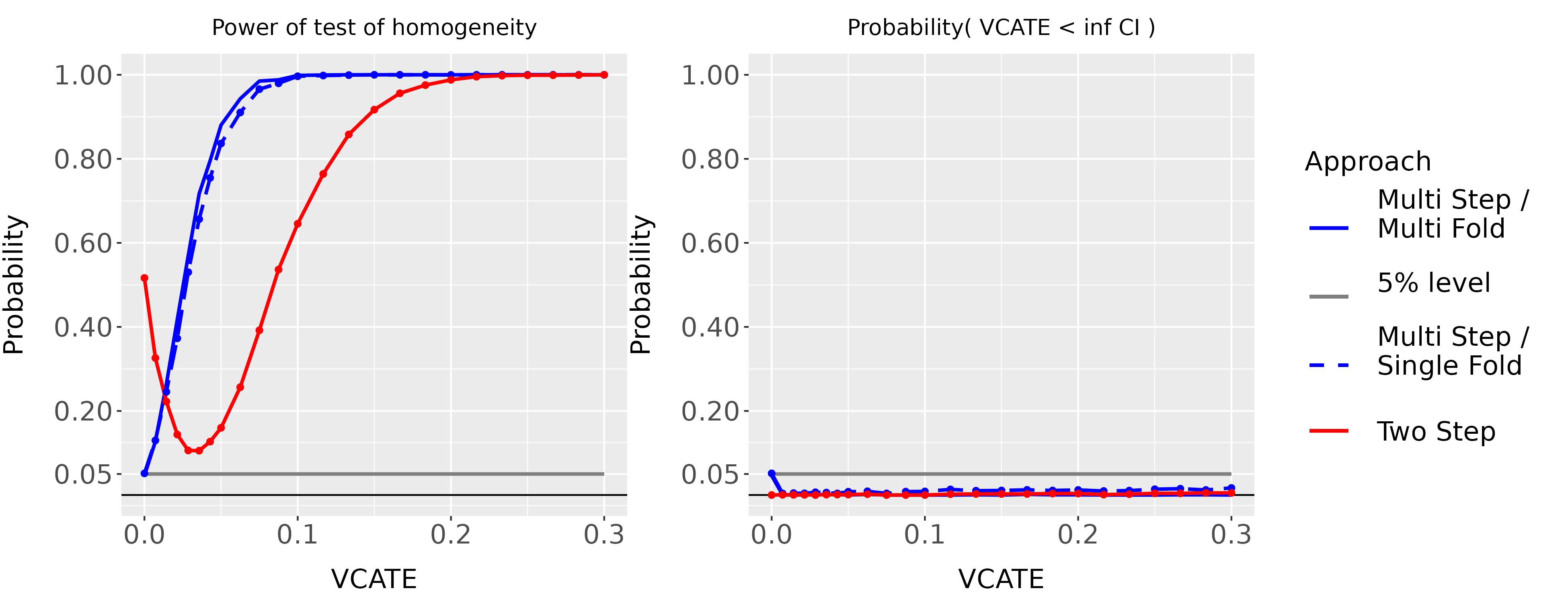}	
	\caption{\scriptsize \textbf{(Power and size)} Panel(a) shows the power of tests of homogeneity (whether zero is contained in the confidence interval) for a set of fine-grained values of the true $V_{\tau}$, given $n=2500$, $K = 2$, single split, and $2J = 10$. Panel (b) shows the size of one-sided tests, i.e., the probability that the VCATE is strictly below the CI bounds. The multi-fold CIs and single fold CIs are defined in \eqref{eq:degenerate_ci} and \eqref{eq:multisplit_cis_degenerate}, respectively. The two-step CIs are constructed as $(1/n)\sum_{i=1}^n \widehat{\varphi}_i \pm 1.96\sqrt{\smash[b]{\widehat{V}_{\varphi}/n}}$, where $\widehat{\varphi}_i$ is the summand  in \eqref{eq:twostep_dml} and $\widehat{V}_{\varphi}$ estimates its sample variance. }
	\label{fig:power_by_vtau}
\end{figure}

Figure \ref{fig:power_by_vtau}.(left) shows the power of tests of homogeneity in a simulation with $n = 2500$. The multi-step, single fold approach has correct size control and has local power, in line with the result of Lemma \ref{lem:powerlocal_alternatives}. The power of the test using the multifold approach is similar to using a single fold. The right panel shows the probability that the VCATE is strictly below the CI bounds. As predicted by theory, this probability is uniformly bounded by $\alpha = 0.05$ for the single and multifold approaches (see Corollaries \ref{cor:uniformity low} and \ref{cor:one_sidedtests_degenerate}, and Theorem \ref{thm:uniform_coverage_variational_lvci}). The two-step approach has a non-monotonic power curve with incorrect size. The size of one-sided tests in the right panel is uniformly bounded by $\alpha = 0.05$, though this may be partly the fact that $\widehat{V}_{\tau n}^{two-step}$ can take negative values and has a negative bias (see Figure \ref{fig:densityplots_vtauhat}).



\section{Empirical Example}

In this section, I illustrate my approach using data from a large-scale information experiment conducted by \cite{dizon2019parents}. The study, which covered 39 school districts, involved an intervention to provide low-income parents of at least two children with information about their children's school performance. Half the households where assigned to the information intervention and the rest were assigned to the control group. \cite{dizon2019parents} showed that, at baseline, parents faced large information gaps regarding their children's grades and class ranking. Even though schools produced a report card, 60\% of parents were unaware of their child's performance. Many parents reported that they did not receive the report card (children either lost them or did not take them home), or had trouble interpreting the report card structure, primarily due to low literacy levels. 

The intervention was designed to present details of their children's school performance in an easily accessible way.
\cite{dizon2019parents} showed that the information gaps (the difference between believed and true test scores) went down as a result of the intervention, and the amount of updating varied depending on students' initial test scores. \cite{dizon2019parents} also introduced a real-stakes scenario where parents received a series of lottery tickets for a scholarship paying for four years of high school. Parents had to decide how to allocate tickets between two siblings. If there were more than two siblings residing in the household, the survey team selected two at random. The results showed that parents allocated tickets towards their better performing child.

To test for heterogeneity, \cite{dizon2019parents} ran a linear regression of parental beliefs on initial scores, treatment, and an interaction as in \eqref{eq:regression}, and reported estimates $\bar{X} = 46.8$ (on a scale of 100) and $(\widehat{\beta}_1,\widehat{\beta}_2) = (-25.9,0.40)$ in their Tables 1 and 2, respectively. The coefficient $\widehat{\beta}_2$ captures how much the treatment effects vary (on a $0$ to $100$ scale) for a each additional point in students initial scores. The VCATE combines information about the coefficient and the initial variability in scores. From the data, I also estimate the variance of the control $\widehat{V}_x = 305.53$, and estimate the VCATE as $\widehat{\beta}_2^2\widehat{V}_x = 48.89$. Taking the square root and normalizing by the standard deviation of the outcome for the control group ($\widehat{V}_{(Y \mid D= 0)} = 311.72$), produces $0.40$. This means that the magnitude of treatment effect heterogeneity explained by scores is comparable to 40\% of the standard deviation of beliefs in the control group.

The VCATE can also help us understand the magnitude of treatment effect heterogeneity in the experiment using multiple covariates. I use LASSO for the first-stage predictions, two folds per split with 20 splits, and estimate $\{\widehat{\Omega}_{nk}\}_{k=1}^K$ using clustered standard errors at the household level and the formulas for the confidence intervals defined in \eqref{eq:multisplit_cis_degenerate}. The results are not very sensitive to the number of splits. For ease of exposition, I report point-estimates and CIs for $\sqrt{\frac{V_\tau}{\mathbb{V}(Y_0)}}$, which is the standard deviation of the CATE divided by the standard deviation of the outcome for the control group. I compute confidence intervals for the square root via the transformation proposed in \eqref{eq:sqrt_transform_ci}.


\label{sec:empirical_example}
\begin{table}[hbt!]
	\centering
	\begin{adjustbox}{width=\columnwidth,center}

	{\small
		\begin{tabular}{lcccccccc}
			\hline
			\\
			~ & N & Clusters & Estimate & 95 & \multicolumn{2}{c}{Welfare  Bounds} \\ 
			~ & ~ & (HH) & $\sqrt{V_\tau/\mathbb{V}(Y_0)}$ & \% C.I. &  Thm. \ref{thm:simple_regretbound} & Thm. \ref{thm:general_regretbound} \\
			
			\\	
			\hline \\
			\multicolumn{5}{l}{Panel (a): Endline Beliefs. $\frac{ATE}{\sqrt{\mathbb{V}(Y_0)}}$: -0.42 (0.03)}      \\ 
			&  &  &  &  &  &  &     \\ 
        Scores & 5244 & 2626 & 0.40 & [0.32, 0.48] & 0.201 & 0.081 \\ 
Parent Years of Education & 5208 & 2608 & 0.05 & [0.00, 0.13] & 0.023 & 0.001 \\ 
Scores + Parents' Education & 5208 & 2608 & 0.40 & [0.32, 0.48] & 0.199 & 0.079 \\ 
Above Median Educ. Expenses & 5244 & 2626 & 0.06 & [0.00, 0.15] & 0.028 & 0.002 \\ 
Respondent Variables & 4722 & 2365 & 0.03 & [0.00, 0.12] & 0.015 & 0.001 \\ 
Household Variables & 5244 & 2626 & 0.02 & [0.00, 0.10] & 0.009 & 0.000 \\ 
Student Variables & 4959 & 2532 & 0.11 & [0.04, 0.20] & 0.057 & 0.007 \\ 
All variables  & 4464 & 2278 & 0.40 & [0.31, 0.48] & 0.199 & 0.079 \\
			&  &  &  &  &  &  &     \\ 
			\multicolumn{5}{l}{Panel (b): Secondary School Lottery. $\frac{ATE}{\sqrt{\mathbb{V}(Y_0)}}$: 0.00 (0.00)}   &  &  &     \\ \hline
			&  &  &  &  &  &  &     \\ 
        Scores & 5258 & 2629 & 0.16 & [0.07, 0.24] & 0.078 & 0.078 \\ 
Parent Years of Education & 5222 & 2611 & 0.00 & [0.00, 0.00] & 0.000 & 0.000 \\ 
Scores + Parents' Education & 5222 & 2611 & 0.16 & [0.07, 0.24] & 0.078 & 0.078 \\ 
Above Median Educ. Expenses & 5258 & 2629 & 0.03 & [0.00, 0.07] & 0.015 & 0.015 \\ 
Respondent Variables & 4736 & 2368 & 0.00 & [0.00, 0.00] & 0.000 & 0.000 \\ 
Household Variables & 5258 & 2629 & 0.00 & [0.00, 0.00] & 0.000 & 0.000 \\ 
Student Variables & 4971 & 2535 & 0.06 & [0.00, 0.15] & 0.029 & 0.027 \\ 
All variables & 4476 & 2281 & 0.15 & [0.05, 0.24] & 0.076 & 0.074 \\  
			\hline
		\end{tabular}
	}
	\end{adjustbox}

	\caption[empirical_estimates]{\scriptsize \textbf{(Empirical Estimates)} Each panel computes the ATE and VCATE normalized by the standard deviation of each outcome for the control group. The $N$ varies depending on the missing values for the covariates and the outcome. Each line within panels (a) and (b) considers 8 different sets of covariates measured at baseline which include test scores, years of parental education, an indicator for whether annual educational expenditures the previous year (uniforms, fees, school supplies) are above the median, other respondent variables (gender, age, is literate, is farmer), household variables (number of kids, single-parent), and student variables (grade, age, gender, attendance). I estimated clustered covariance matrices at the household level.  The function $\mu_d(x)$ is computed using LASSO with 10-fold cross-validation, and the estimates are computed using 2-fold cross-fitting with 20 splits. The point estimates are the median values of $\widetilde{V}_{\tau n}$ across splits. I compute the bound in Column 7, using the ATE from the corresponding subsample.  }
	\vspace{-0.1in}
	\label{tab:empirical_estimates}
\end{table}


Table \ref{tab:empirical_estimates} computes the ATE and the $\sqrt{VCATE}$ for two outcomes (parental beliefs and lottery allocations) and 8 different sets of covariates. Panel (a) shows that, on average, parents downgrade their beliefs about test scores by 42\% of the standard deviation (SD) of the beliefs of the control group. The treatment effect heterogeneity explained by test scores is equivalent to 40\% of the standard deviation (SD) of the beliefs of the control group. This is statistically significant at the 5\% level and has a comparable magnitude to the ATE. The confidence intervals are relatively short in length. However, applying the bounds from Theorem \ref{thm:general_regretbound}  and Corollary \ref{cor:invariance_transformation} shows that differentiating treatment offers based on scores could further lower beliefs by at most 8.1\% SDs of the beliefs in the control group. In this case the ATE is already fairly high compared to $\sqrt{V_{\tau}}$, so in spite of the large heterogeneity, the marginals gains from targeting would be modest. Panel (b) presents the results for the secondary school lottery. The ATE is estimated precisely at zero, because the lottery tickets had to be divided as a zero sum between the siblings. The VCATE measures how much the dispersion in the allocation depends on the covariates. The standard deviation of the VCATE explained by initial scores is 16\% of the SD of the control group lottery allocation, and the maximum welfare gains are around 7.8\% SD.

The student variables (grade, age, gender, attendance, and educational expenditures) collectively explain 11\% of the SD of parental beliefs. This is statistically significant at the 5\% level. The magnitude is around a fourth of the variation for test scores, and the maximum welfare gain from targeting is 0.7\% of the SD of parental beliefs in the control group.  I find that other subsets of covariates do not produce statistically significant estimates of the VCATE at the 5\% level. The added welfare of personalizing treatment assignment using these covariates is also very low.

The estimates that use all the covariates are computed over a smaller subsample with non-missing values across all variables. Despite the large number of variables and the smaller sample, the estimates of the VCATE remain relatively stable across specifications. The VCATE computed from a rich set of respondent, household, and student covariates has a comparable magnitude to the VCATE that only includes student scores. The confidence intervals are also similar. The estimates of the maximum welfare gains from targeting using all covariates are 7.9\% SD for beliefs and 7.4\% SD for lottery outcomes, respectively, which are similar to the welfare gains computed using only scores.


\section{Conclusion}
\label{conclusion}
I propose an efficient estimator of the variance of treatment effects that can be attributed to baseline characteristics and propose novel adaptive confidence intervals that produce valid coverage. I analyze issues of non-standard inference that arise in this context, and how to address them. I also explore the economic significance of the VCATE for policymakers and researchers, by showing that the $\sqrt{VCATE}/2$ bounds the marginal gains of targeted policies. Overall, this paper proposes a broadly applicable approach to measure treatment effect heterogeneity in experiments.

\bibliographystyle{elsarticle-harv}
\bibliography{bibliographyml}

\counterwithin{assump}{section}

\newpage
\begin{appendices}

\section{Critical Values}
\label{sec:appendix_criticalvalues}
%
Define the correlation $\rho = \Omega_{12}/\sqrt{\Omega_{11}\Omega_{22}}$. By definition of the Cholesky decomposition, $(e_1'\Omega^{1/2}Z) = \sqrt{\Omega_{11}}Z_1,$ and $(e_2'\Omega^{1/2}Z) = \sqrt{\Omega_{22}(1-\rho^2)}Z_2 + \rho \sqrt{\Omega_{22}} Z_1$. Substituting these terms into the expression for $G$
\begin{align*}
	&G(n,V_{\tau}^*,\Omega,Z,\zeta) \\
	&= \underbrace{\frac{\Omega_{11}}{n}}_{\nu_1}Z_1^2 + \underbrace{2\left[ \zeta \sqrt{\frac{V_{\tau}^*\Omega_{11}}{n}} + \frac{V_{\tau}^*}{2\sqrt{n}} \rho \sqrt{\Omega_{22}}   \right]}_{\kappa_1} Z_1+\underbrace{\frac{V_{\tau}^*}{\sqrt{n}}\left[\sqrt{\Omega_{22}(1-\rho^2)} \right]}_{\kappa_2}Z_2.
\end{align*}
This has a quadratic form, $G(n,V_{\tau}^*,\Omega,z,\zeta) = \nu_1 \left(Z_1 + \frac{\kappa_1}{2 \nu_1 } \right)^2 + \kappa_2Z_2 - \frac{\kappa_1^2}{4 \nu_1 }$, which fits the form of a generalized chi-square \citep{das2021method}. To compute critical values we compute feasible analogs $(\widehat{\nu}_1,\widehat{\kappa}_1,\widehat{\kappa}_2)$ from an estimate of $\Omega$.


\section{Proofs Main Document}

\begin{proof}[Proof of Lemma \ref{lem:monotonicity}]
	Define $\tilde{\tau}(X') := \mathbb{E}[Y_1-Y_0 \mid X']$. Since $X$ is $X'$-measurable, then by the law of iterated expectations, $\mathbb{E}[\tilde{\tau}(X')\mid X] = \tau(X)$. By the law of total variance $V_{\tau}' = \mathbb{V}(\tau(X)) + \mathbb{E}[\mathbb{V}(\tilde{\tau}(X')\mid X)] \ge V_\tau$. We can prove the upper bound by setting $X' = Y_1 - Y_0$.
\end{proof}

\begin{proof}[Proof of Theorem \ref{thm:simple_regretbound}] The result is a special case of Theorem \ref{thm:general_regretbound}. The most adversarial distribution is one where $\tau_{av} = 0$.
\end{proof}

\begin{proof}[Proof of Theorem \ref{thm:general_regretbound}]
	
	Our goal is to find $\mathcal{R} := \sup_{\gamma \in \Gamma}\sup_{\pi \in \Pi} \ \mathcal{U}_\gamma(\pi)$, and to prove that $\mathcal{U}_\gamma(\pi) = \mathcal{R}$ for at least one $\gamma \in \Gamma$ and $\pi \in \Pi$.
	
	Define two random variables $T := \mathbb{E}_\gamma[Y_1 - Y_0 \mid X = x]$ and $M := \mathbbm{1}\{T >0 \}$. By adding/subtracting $\mathbb{E}_\gamma[Y_0]$ and applying the law of iterated expectations, $\mathcal{U}_\gamma(\pi) = \mathbb{E}_\gamma[\pi(X)T] -\max\{0,\mathbb{E}_\gamma[T] \} $. The optimal policy $\pi^*(X) = \mathbbm{1}\{T \ge 0\}$, which belongs to $\Pi$, i.e. the ``first-best''  \citep{kitagawa2018should}. This means that $\mathcal{R} = \sup_{\gamma \in \Gamma} \ \mathbb{E}_\gamma[\max\{0,T\}] - \max\{0,\mathbb{E}_\gamma[T]\}$.  \\
	
	\textbf{Step 1: (Problem  Equivalence) } The set of distributions $\Gamma$ is very large. Instead I focus on the problem over a set of equivalence classes. For $m \in \{0,1\}$, define the moments $p_m = \mathbb{P}_\gamma(M = m)$, $\tau_m = \mathbb{E}_\gamma[T \mid M = m]$, and $\omega_m = \mathbb{V}_\gamma(T \mid M = m)$. By definition, $\mathcal{U}_\gamma(\pi^*) = p_1 \tau_1 - \max\{ 0, \tau_{av} \}$, where $p_1\tau_1$ is the proportion of people that benefit from treatment times their conditional mean. Let $\Delta := \tau_1 - \tau_0$ be the mean difference between those that benefit from the program and those that do not. By definition $\Delta > 0$ because $\tau_1 > 0$ and $\tau_0 \le 0$. The conditional treatment effect is $\mathbb{E}_\gamma[T \mid M] =  \tau_0 + M \Delta$, and $ \tau_{av} =  \tau_0 + p_1 \Delta  $. Rearranging these expressions, $\tau_1 = \tau_{av} + (1-p_1)\Delta$ and $\mathcal{U}_\gamma(\pi^*) = p_1 (\tau_{av} + (1-p_1) \Delta) - \max\{ 0, \tau_{av} \}$.	By definition $\mathbb{V}_\gamma(\mathbb{E}_\gamma[T \mid M]) = \Delta^2p_1(1-p_1)$. By applying the law of total variance $V_{\tau} = p_1 \omega_1^2 + (1-p_1)\omega_0^2 + \Delta^2(p_1)(1-p_1)$. Then
	\begin{align}
	\begin{split}
	\mathcal{R} &= \sup_{\{p_1,\tau_0,\Delta,\omega_0^2,\omega_1^2\}}  p_1\tau_{av} + p_1(1-p_1) \Delta - \max\{ 0, \tau_{av} \}, \\
	s.t. \qquad & p_1 \in [0,1], \tau_0 \le 0, \tau_0 + \Delta > 0, \Delta > 0, \omega_1^2,\omega_0^2 \ge 0, \tau_0 + p_1\Delta = \tau_{av}, \\
	& p_1 \omega_1^2 + (1-p_1)\omega_0^2 + p_1(1-p_1)\Delta^2 = V_\tau.
	\end{split}
	\end{align}
	The values of $\tau_{av}$ and $V_{\tau}$ impose the following constraints on the feasible set:
	\begin{table}[!htbp]
		\centering
		\begin{tabular}{l|cc}
			& $V_{\tau} = 0$ & $V_{\tau} > 0$  \\
			\hline
			$\tau_{av} \le 0$ &  $p_1 = 0$, $\tau_0 = \tau_{av}$, $\omega_0 = 0$ & $p_1 \in [0,1)$ \\
			$\tau_{av} > 0$ & $p_1 = 1$, $\tau_1 = \tau_{av}$, $\omega_1 = 0$  & $p_1 \in (0,1] $  
		\end{tabular}
	\end{table}
	
	By the form of the objective, if $p_1 \in \{0,1\}$ then $\mathcal{R} = 0$. Therefore, if $V_{\tau} = 0$, then $\mathcal{R} = 0$ and the result of the theorem holds. Any value of the remaining parameters that satisfies the sign constraints will be feasible. Without loss, we focus on $V_{\tau} > 0$. \\

	\textbf{Step 2: (Optimum has binary support)} Consider a situation where $(p_1,\omega_1^2,\omega_0^2)$ are fixed and $p_1 \in (0,1)$. From the variance equation, $ \Delta^* = \sqrt{\frac{V_\tau - p_1\omega_1^2 - (1-p_1)\omega_0^2}{p_1(1-p_1)}}$, and from the mean $\tau_0^* = \tau_{av} - p_1\Delta^*$ and $\tau_1^* = \tau_{av} + (1-p_1) \Delta^*$. 
		
	A solution is feasible as long as $\omega_1^2,\omega_0^2 \ge 0$, $p_1\omega_1^2 + (1-p_1)\omega_0^2 \le V_\tau$, $\tau_0^* \le 0$, $\tau_1^* > 0$, and $\Delta^* =\tau_1^*-\tau_0^*> 0$. The optimal $\tau_0^*$ is strictly increasing in $(\omega_1^2,\omega_0^2)$ and $\tau_1^*$ is strictly decreasing. If a given value of $(\omega_1^2,\omega_0^2)$ is feasible, then another candidate which shrinks it to zero will still satisfy the constraints. Moreover, the objective function is strictly increasing in $\Delta^*$, which in turn is strictly decreasing in $(\omega_1^2,\omega_0^2)$. Therefore the optimum is $(\omega_1^2,\omega_0^2) = (0,0)$, i.e. binary support for the CATE, $\Delta^* = \sqrt{V_\tau / p_1(1-p_1)}$, $\tau_0^* = \tau_{av} - \sqrt{p_1/(1-p_1)}\sqrt{V_{\tau}}$, and $\tau_1^* = \tau_{av} + \sqrt{(1-p_1)/p_1}\sqrt{V_{\tau}}$.
	
	Only some values of $p_1$ are feasible, i.e. satisfy the constraints $\Delta^* > 0$, $\tau_0^* \le 0$, and $\tau_1^* > 0$. The function $p_1/(1-p_1)$ is strictly increasing in $p$ with a range in $(0,\infty)$. The feasible values depend on the sign of $\tau_{av}$. If (i) $\tau_{av} \le 0$, then $p_1 \in \left(0,\frac{V_{\tau}}{\tau_{av}^2 + V_{\tau}}\right]$. If (ii) $\tau_{av} > 0$, $p_1 \in \left[\frac{\tau_{av}^2}{\tau_{av}^2 + V_{\tau}},1\right)$. \\
	
	\textbf{Step 3: (Solve points of support)}
	 For case (i) the objective function is $ \mathcal{R} = \max_{p_1\in \left[0,\frac{V_{\tau}}{\tau^2 + V_{\tau}}\right]} \quad p_1 \tau_{av} + \sqrt{p_1(1-p_1)}\sqrt{V_\tau}$. The unique solution to the FOC is $p_1^* = (1/2) - (1/2)\sqrt{\tau_{av}^2 /(\tau_{av}^2 + V_{\tau})}$, and is interior because  $\sqrt{\tau_{av}^2/(\tau_{av}^2 + V_{\tau})} \ge \tau_{av}^2 / (\tau_{av}^2 + V_{\tau})$. By strict concavity, it is a global optimum, and $\mathcal{R} = \tau_{av}/2 -(\tau_{av}|\tau_{av}|)/(2 \sqrt{\tau_{av}^2 +V_\tau}) + V_\tau/(2\sqrt{\tau_{av}^2 + V_\tau})$. Since $\tau_{av} \le 0$, this can be simplified to $\frac{1}{2}(-|\tau_{av}| + \sqrt{\tau_{av}^2 + V_{\tau}})$.	For case (ii), $\tau_{av} > 0$, the objective function is $(p_1 - 1)\tau_{av} + \sqrt{p_1(1-p_1)}\sqrt{V}_{\tau}$, subject to $1 \ge p_1 \ge \frac{\tau_{av}^2}{\tau_{av}^2+V_\tau}$. The unique solution to the FOC is $p_1^* =  (1/2) + (1/2)\sqrt{\tau_{av}^2 /(\tau_{av}^2 + V_{\tau})}$, it is interior because $p_1^* \ge \frac{1}{2} + \frac{1}{2}\frac{\tau_{av}^2}{\tau_{av}^2 + V_{\tau}} \ge \frac{\tau_{av}^2}{\tau_{av}^2+V_\tau}$, and produces the desired $\mathcal{R}$.

\end{proof}

\begin{proof}[Proof of Corollary \ref{cor:invariance_transformation}] The potential outcomes under a linear transformation are $\kappa_1 + \kappa_2 Y_0$ and $\kappa_1 + \kappa_2 Y_1$, respectively. The treatment effect is $\kappa_2 (Y_1 - Y_0)$, which does not depend on $\kappa_1$, and the transformed CATE is $\kappa_2\tau(x)$. The ATE is $\kappa_2\tau_{av}$ and the VCATE is $\kappa^2V_{\tau}$. The result follows from substituting the transformed values into Theorem \ref{thm:general_regretbound} and factorizing $|\kappa_2|$.

\end{proof}

	\begin{proof}[Proof of Lemma \ref{lem:decomposition_remainder_lvci_reg}	]
		We decompose $\widehat{V}_{\tau n} = \widehat{\beta}_{2n}^2 \widehat{V}_{xn}$ into components that map into those of Assumption \ref{assump:clt_regcoef}, by centering the key terms.
		\begin{equation}
		\label{eq:decomposition_hatVtau_star_expanded_simple}
		\begin{aligned}
		&\widehat{V}_{\tau n} - V_{\tau n}^* =  \widehat{\beta}_{2n}^2\widehat{V}_{xn} - \beta_{2n}^2V_{xn} = (\widehat{\beta}_{2n}-\beta_{2n} + \beta_{2n})^2\frac{V_{xn}}{V_{xn}}(\widehat{V}_{xn} - V_{xn} + V_{xn}) - \beta_{2n}^2V_{xn} \\
		& =[\sqrt{V}_{xn}(\widehat{\beta}_{2n}-\beta_{2n})]^2+ 2[\sqrt{V_{xn}}\beta_{2n}][\sqrt{V}_{xn}(\widehat{\beta}_{2n}-\beta_{2n})]+ [V_{xn}\beta_2^2]\left(\frac{\widehat{V}_{xn}}{V_{xn}}-1\right)  \\
		&\quad 
		+ [\sqrt{V_{xn}}(\widehat{\beta}_{2n}-\beta_{2n})]^2\left( \frac{\widehat{V}_{xn}}{V_{xn}}-1\right) + 2[\sqrt{V_{xn}}\beta_{2n}][\sqrt{V_{xn}}(\widehat{\beta}_{2n}-\beta_{2nk})]\left( \frac{\widehat{V}_{xn}}{V_{xn}}-1\right).
		\end{aligned}
		\end{equation}
		Let $\widehat{Z}_{n} = n^{1/2}\Omega_{n}^{-1/2}[ \sqrt{V}_{xn}(\widehat{\beta}_{2n}-\beta_{2n}),\widehat{V}_{xn}/V_{xn}-1]'$ be a normalized statistic, and let $e_1 = [1,0]'$, $e_2 = [0,1]'$ be vectors that select the first and second coordinates, respectively. We can always find a $\zeta_n \in \{-1,1\}$ that solves $\zeta_{n}\sqrt{V_{\tau n}} = \beta_{2n}\sqrt{V_{xn}}$. If $V_{\tau n} > 0$ this is the sign of $\beta_{2n}$ but otherwise any $\zeta_n \in \{-1,1\}$ will solve the equation. By substituting the definition of $\widehat{Z}_n$, $V_{\tau n} = \beta_{2n}^2V_{xn}$, and 
		\begin{align}
		\begin{split}
			\widehat{V}_{\tau n} - V_{\tau n}^*
			&= \frac{(e_1'\Omega_n\widehat{Z}_n)^2}{n} + 2 \zeta_n \sqrt{V_{\tau n}^*}\left( \frac{e_1'\Omega_n^{1/2}\widehat{Z}_n}{\sqrt{n}}\right) + V_{\tau n}^*\frac{e_2'\Omega_n^{1/2}\widehat{Z}_n}{\sqrt{n}} \\
			&\quad+ \frac{ (e_1'\Omega_n^{1/2}\widehat{Z}_n)^2(e_2'\Omega_n^{1/2}\widehat{Z}_n)}{n^{3/2}} + 2 \zeta_n \sqrt{V_{\tau n}^*}\frac{(e_1'\Omega_n^{1/2}\widehat{Z}_n)(e_2'\Omega_n^{1/2}\widehat{Z}_n)}{n}.
			\end{split}
			\label{eq:decomposition_empiricalprocess_proof}
		\end{align}
		By Assumption \ref{assump:clt_regcoef}, $\widehat{Z}_n \to Z_n + o_p(1)$. Moreover, since $\Omega_n$ has bounded eigenvalues, $\Omega_n^{1/2}\widehat{Z}_n = \Omega_n^{1/2}Z_n + o_p(1)$. This means that 
		
		\begin{equation}
		\begin{aligned}
		\widehat{V}_{\tau n} - V_{\tau n}^*&= \frac{(e_1'\Omega_n^{1/2}Z_n)^2}{n} + 2 \zeta_n \sqrt{\frac{V_{\tau n}^*}{n}}e_1'\Omega_n^{1/2}Z_n + \frac{V_{\tau n}^*}{\sqrt{n}}e_2'\Omega_n^{1/2}Z_n \ + \ \text{Residual}_n, 
	    \end{aligned}	
		\label{eq:decomposition_empiricalprocess_proof2}
		\end{equation}
		where the residual is $o_p\left(n^{-1}\right) +  o_p\left(\sqrt{V_{\tau n}^*/n}\right)+ o_p\left(V_{\tau n}^*/\sqrt{n }\right) + O_p(n^{-3/2}) + o_p(n^{-1}\sqrt{V_{\tau n}^*})$. The fourth and fifth terms of the residual are $o_p(n^{-1})$ and $o_p(n^{-1/2}\sqrt{V_{\tau n}^*})$, and hence asymptotically negligible.  The leading term in \eqref{eq:decomposition_empiricalprocess_proof2} is $O_p\left( \max\left\{1/n,\sqrt{V_{\tau n}^*/n}\right\} \right)$.

	\end{proof}

\begin{proof}[Proof of Lemma \ref{lem:variance_eff_influence}]

Let $\varphi_i = \varphi(Y_i,D_i,X_i,\eta)$ be a realization of the influence function in \eqref{eq:efficient_influence} and compute $\mathbb{V}[\varphi_i \mid D_i=d,X_i =x] = 4 (\tau(x)-\tau_{av})^2\left[ \frac{d\sigma_1^2(x)}{p(x)^2} + \frac{(1-d)\sigma_0^2(x)}{(1-p(x))^2} \right]$ and $\mathbb{E}[\varphi_i \mid D_i=d,X_i =x] =  (\tau(x)-\tau_{av})^2$.

By the law of iterated expectations $\mathbb{E}[\varphi_i \mid X_i=x]=(\tau(x)-\tau_{av})^2$. By applying the law of total variance recursively, $\mathbb{V}(\varphi_i) = \mathbb{V}(\mathbb{E}[\varphi_i \mid X_i]) + \mathbb{E}[\mathbb{E}[\mathbb{V}(\varphi_i \mid D_i,X_i)\mid X]] + \mathbb{E}[\mathbb{V}(\mathbb{E}[\varphi_i \mid D_i,X_i]\mid X_i)]$. This produces, $ \mathbb{V}(\varphi_i) = \mathbb{V}((\tau(X_i)-\tau_{av})^2) +4 \mathbb{E}\left[(\tau(X_i)-\tau_{av})^2\left( \frac{\sigma_1^2(X_i)}{p(X_i)} + \frac{\sigma_0^2(X_i)}{(1-p(X_i))}\right)\right]$.
\end{proof}

\begin{proof}[Proof of Corollary \ref{cor:bound_varefficient}]
	$\mathbb{V}((\tau(X_i)-\tau_{av})^2) \le \mathbb{E}[(\tau(X_i)-\tau_{av})^4] \le \kappa^2 V_{\tau}$. By the Cauchy-Schwarz inequality, the term $\mathbb{E}\left[(\tau(X_i)-\tau_{av})^2\left( \frac{\sigma_1^2(X_i)}{p(X_i)} + \frac{\sigma_0^2(X_i)}{(1-p(X_i))}\right)\right] $ is bounded by $\kappa V_{\tau}\sqrt{\mathbb{E}\left[ \left( \frac{\sigma_1^2(X_i)}{p(X_i)} + \frac{\sigma_0^2(X_i)}{(1-p(X_i))}\right)^2\right]}$.
\end{proof}

\begin{proof}[Proof of Lemma \ref{lem:equivalence_efficient_simple}]
	Let $\mathcal{Q}(\theta)$ be the Jacobian of the least squares problem, defined in \eqref{eq:foc_adjustedols}. Substituting $\widehat{W}(X_i,D_i)'e_4 = (D_i-p(X_i))\widehat{S}(X_i)$ and $\lambda(X_i) = [p(X_i)(1-p(X_i))]^{-1}$ into the definition, $\beta_2\mathcal{Q}(\theta)e_4 = \frac{1}{n}\sum_{i=1}^n \frac{[D_i-p(X_i)]}{p(X_i)(1-p(X_i))}\beta_2\widehat{S}(X_i)(Y_i - \widehat{W}(X_i,D_i))$. By the definition of the nuisance functions in $\tilde{\eta}_\theta(x)$ in \eqref{eq:pseudonuisancefunctions}, $\widehat{W}(X_i,D_i)'\theta = \tilde{\mu}_{0,\theta}(X_i) + D_i\tilde{\tau}_\theta(X_i)$, $\tilde{\tau}_\theta(x) = \beta_1 + \beta_2 \widehat{S}(X_i)$, and $\tilde{\tau}_\theta = \beta_1 + \left[\frac{1}{n}\sum_{i=1}^n \widehat{S}(X_i)\right]\beta_2$. Furthermore, since $\frac{1}{n}\sum_{i=1}^n\widehat{S}(X_i) = 0$, then $\beta_2\widehat{S}(X_i) = \tilde{\tau}_\theta(X_i)-\tilde{\tau}_{av,\theta}$ and
	$$\beta_2\mathcal{Q}(\theta)e_4 = \frac{1}{n}\sum_{i=1}^n \frac{[D_i-p(X_i)]}{p(X_i)(1-p(X_i))}\beta_2\widehat{S}(X_i)(Y_i -  \tilde{\mu}_{0,\theta}(X_i) - D_i\tilde{\tau}_\theta(X_i)).$$
	Notice that $D_i - p(X_i) = D_i(1-p(X_i)) + p(X_i)(1-D_i)$. This implies that $\frac{D_i-p(X_i)}{p(X_i)} = \frac{D_i}{p(X_i)} - \frac{(1-D_i)}{1-p(X_i)}$. The next step is to substitute the estimated parameters. Then $ \frac{1}{n}\sum_{i=1}^n \varphi(Y_i,D_i,X_i,\tilde{\eta}_{\widehat{\theta}_n})  =  \frac{1}{n}\sum_{i=1}^n (\tilde{\tau}_{\widehat{\theta}_n}(X_i)-\tilde{\tau}_{av,\widehat{\theta}_n})^2 + \widehat{\beta}_{2n}\mathcal{Q}(\widehat{\theta}_n)e_4$. The first term simplifies to $\widehat{\beta}_{2n}^2\left[\frac{1}{n}\sum_{i=1}^n \widehat{S}(X_i)^2\right] = \widehat{\beta}_{2n}^2\widehat{V}_{x n}$. Finally, $\mathcal{Q}(\widehat{\theta}_n) = 0$ implies that the second term is zero. 
	
\end{proof}

\begin{proof}[Proof of Lemma \ref{lem:projection_weak}]
	
	The parameter $\theta = (\alpha_1,\alpha_2,\beta_1,\beta_2) \in \Theta^*$ solves
	\begin{equation}
		\label{eq:ols_id_condition_proof}
	\mathbb{E}\left[\lambda(X)W(X,D)Y\right] - \mathbb{E}\left[\lambda(X)W(X,D)W(X,D)'\right]\theta = 0_{4 \times 1}.
	\end{equation}
	If there exists a $\theta \in \mathbb{R}^4$ such that $\mathbb{E}[Y\mid D=d,X=x] = \mu_d(x) = W(x,d)'\theta$, then by the law of iterated expectations $\mathbb{E}\left[\lambda(X)W(X,D)Y \right] = \mathbb{E}[\lambda(X)W(X,D)W(X,D)'\theta]$. Such a $\theta$ would automatically satisfy \eqref{eq:ols_id_condition_proof}. 	Now suppose that $S(x) = \tau(x) - \tau_{av}$, $M(x) = \mu_0(x) + p(x)\tau(x)$, and $\theta = (0,1,\tau_{av},1)$. For this choice, $
		W(x,d)'\theta = 0 + (\mu_0(x)+p(x)\tau(x)) + (d-p(x))\tau_{av} + (d-p(x))(\tau(x)-\tau_{av})$.	This simplifies to $W(x,d)'\theta = \mu_0(x) + d\tau(x) = \mu_d(x)$ and $\theta$ solves \eqref{eq:ols_id_condition_proof}.

\end{proof}

\begin{proof}[Proof of Theorem \ref{thm:clt_lvci}.(i)]
	Define a vector $\uline{Y}_{nk} \in \mathbb{R}^{n_k}$ with the outcomes in fold $\mathcal{I}_{nk}$, $\uline{W}_{nk}$ be an $n_k \times 4$ matrix of regressors, $\uline{U}_{nk} \in \mathbb{R}^{n_k}$  be a vector of errors, and $\Lambda_{nk}$ be a $n_k \times n_k$ diagonal matrix with entries $\{ \lambda(X_i) \}$. Let $\uline{\widehat{W}}_{nk}$ be an $n_k \times 4$ matrix of generated regressors defined in \eqref{eq:thetahat_dml}, and let $\Pi_{nk}$ be a $4 \times 4$ diagonal matrix with entries $(1,1,1,V_{xnk}^{-1/2})$. Let $\theta_{nk}^* := (\alpha_{1nk},\alpha_{2nk},\beta_{1nk},\sqrt{V}_{xnk}\beta_{2nk})$, and define an infeasible estimator $\widetilde{\theta}_{nk}^* :=  \Pi_{nk}^{-1}\widetilde{\theta}_{nk} := \left(\Pi_{nk}'\uline{W}_{nk}'\Lambda_{nk}\uline{W}_{nk}\Pi_{nk}/n_k \right)^{-1}\left(\Pi_{nk}'\uline{W}_{nk}'\Lambda_{nk}Y_{nk}/n_k\right)$.
	
	Define $Q_{ww,nk} := \mathbb{E}[\lambda(X_i)W_iW_i' \mid \mathcal{I}_{-nk}]$. By Lemma \ref{lem:designmatrix_dml},
	\begin{equation}
	\label{eq:designmatrix_tildew_normalized}
	\Pi_{nk}'Q_{ww,nk}\Pi_{nk} =  \begin{pmatrix} \mathbb{E}[\lambda(X_i) \mid \mathcal{I}_{-nk}] & \mathbb{E}[\lambda(X_i)M_{-k}(X_i) \mid \mathcal{I}_{-nk}] & 0 & 0 \\
	\mathbb{E}[\lambda(X_i)M_{-k}(X_i) \mid \mathcal{I}_{-nk}] & \mathbb{E}[\lambda(X_i)M_{-k}(X_i)^2 \mid \mathcal{I}_{-nk}] & 0 & 0 \\
	0 & 0 & 1 & 0\\ 
	0 & 0 & 0 & 1\\ \end{pmatrix}.
	\end{equation}
	By Assumptions \ref{assump:unconfound_sutva}.(ii), \ref{assump:momentbounds_dml}.(i) and (ii), $\Pi_{nk}'Q_{ww,nk}\Pi_{nk}$  has bounded eigenvalues. Algebraically, $\frac{\Pi_{nk}'\uline{W}_{nk}'\Lambda_{nk}\uline{W}_{nk}\Pi_{nk}}{n_k} = \frac{1}{n_k}\sum_{i \in \mathcal{I}_{nk}}\Pi_{nk}'\lambda(X_i)W_iW_i'\Pi_{nk}$. Since the terms are conditionally i.i.d, it converges to \eqref{eq:designmatrix_tildew_normalized}. By Assumptions \ref{assump:momentbounds_dml}.(iii) and (iv) $\lambda(X_i)\Pi_{nk}W_iU_i$ is bounded in the $L_2$ norm, and $\sqrt{n_k}(\widehat{\theta}_{nk}^*-\theta_{nk}) = (\Pi_{nk}'Q_{ww,nk}\Pi_{nk})^{-1}\frac{\Pi_{nk}'\uline{W}_{nk}'\Lambda_{nk}\uline{U}_{nk}}{\sqrt{n_k}} + o_p(1)$. Also, $V_{xnk}^{-1}\widehat{V}_{xnk} - 1 
	= \frac{1}{n_k}\sum_{i \in \mathcal{I}_{nk}}[V_{xnk}^{-1}S_{-k}(X_i)^2 -1 ] - \left[\frac{1}{n_k}\sum_{i \in \mathcal{I}_{nk}}V_{xnk}^{-1/2}S_{-k}(X_i) \right]^2$. The second term is $O_p(n^{-1})$ because the summand has mean zero and unit variance.
	$$ A_{i} := \begin{pmatrix} \lambda(X_i)\Pi_{nk}W_iU_i \\ V_{xnk}^{-1}S_{-k}(X_i) - 1 \end{pmatrix}, \quad J_{nk}^* := \begin{bmatrix} (\Pi_{nk}'Q_{ww,nk}\Pi_{nk})^{-1} & \uline{0}_4 \\ \uline{0}_4' & 1 \end{bmatrix}.$$
	By \eqref{eq:designmatrix_tildew_normalized}, $(\Pi_{nk}'Q_{ww,nk}\Pi_{nk})^{-1}$ has a block-diagonal structure, with a one in the bottom-right cell, and hence $\Upsilon [J_{nk}^*]^{-1} = \Upsilon$. This means that
	\begin{equation}
		\label{eq:pre_clt_sum_dml_reg}
		\begin{aligned}  \sqrt{n_k}\Upsilon\begin{pmatrix} \widetilde{\theta}_n^* - \theta_n^* \\ V_{xnk}^{-1}\widehat{V}_{xnk} - 1\end{pmatrix} &= \Upsilon\left[J_{nk}^*\right]^{-1}\left[\frac{1}{\sqrt{n}_k}\sum_{i \in \mathcal{I}_{nk}} A_i \right]+ o_p(1) \\
			&=\frac{1}{\sqrt{n_k}} \sum_{i \in \mathcal{I}_{nk}} \begin{pmatrix} V_{xnk}^{-1/2}\lambda(X_i)(D_i-p(X_i))S_{-k}(X_i)U_{i} \\ V_{xnk}^{-1/2}S_{-k}(X_i)^2 - 1  \end{pmatrix} + o_p(1).
		\end{aligned}
	\end{equation}

	Define  $\widehat{\tau}_{nk,av} := \tfrac{1}{n_k}\sum_{i \in \mathcal{I}_{nk}} \widehat{\tau}_{-nk}(X_i)$ and $\tau_{nk,av} :=  \mathbb{E}[\widehat{\tau}_{-nk}(X_i) \mid \mathcal{I}_{-nk}]$. Algebraically, $\uline{\widehat{W}}_{nk} = \uline{W}_{nk}(I_{4 \times 4} + \Delta_{nk})$, where $\Delta_{nk} = [\uline{0}_4,\uline{0}_4,\uline{0}_4, e_3 (\widehat{\tau}_{nk,av}-\tau_{nk,av})]$ and $e_3 = [0,0,1,0]'$. The bias in centering the fourth column of $\widehat{W}_{nk}$ is a scalar times the third regressor. Substituting $\uline{\widehat{W}}_{nk} = \uline{W}_{nk}\Pi_{nk}\Pi_{nk}^{-1}(I + \Delta_{nk})$,
	\begin{equation} \widehat{\theta}_{nk}^* 
	= \left(\Pi_{nk}^{-1}(I+\Delta_{nk})\Pi_{nk}\right)^{-1}\left( \Pi_{nk}'\uline{W}_{nk}'\Lambda_{nk}\uline{W}_{nk}\Pi_{nk} \right)^{-1}\left( \Pi_{nk}'\uline{W}_{nk}'\Lambda_{nk}\uline{Y}_{nk} \right),
	\label{eq:decomp_fesasibleestimator}
	\end{equation}
	which simplifies to  $\widehat{\theta}_{nk}^*  = \left(\Pi_{nk}^{-1}(I+\Delta_{nk})\Pi_{nk}\right)^{-1}\tilde{\theta}_{nk}^*$. Since $\Delta_{nk}\Delta_{nk} = 0_{4 \times 4}$, then $\left( \Pi_{nk}^{-1}(I+\Delta_{nk})\Pi_{nk}\right)^{-1} =  \left( I + V_{xnk}^{-1/2}\Delta_{nk}\right)^{-1} = I-V_{xnk}^{-1/2}\Delta_{nk}$, and	
		\begin{equation}
	\label{eq:algebraic_equivalence_hatbeta_2n}
	\begin{pmatrix} \widehat{\beta}_{2n}^* \\ V_{xnk}^{-1}\widehat{V}_{xnk} \end{pmatrix} := \Upsilon \begin{pmatrix} \left( \Pi_{nk}^{-1}\widetilde{\Delta}_{nk}\Pi_{nk}\right)^{-1}  & \uline{0}_{4 \times 1}  \\
	\uline{0}_{1 \times 4} & 1 \end{pmatrix}\begin{pmatrix} \widetilde{\theta}_{nk}^* \\ V_{xnk}^{-1}\widehat{V}_{xnk} \end{pmatrix} = \Upsilon \begin{pmatrix} \widetilde{\theta}_{nk}^* \\ V_{xnk}^{-1}\widehat{V}_{xnk} \end{pmatrix}.	\end{equation}  
	The second equality follows from the fact that the selection matrix $\Upsilon$ only extracts the fourth and fifth rows of the vector. We can plug-in \eqref{eq:algebraic_equivalence_hatbeta_2n} and the identity $\Upsilon[\theta_{nk}^{*'},1]' = [\beta_{2nk}^*,1]'$ into \eqref{eq:pre_clt_sum_dml_reg}. By Assumptions \ref{assump:boundsquantities} and \ref{assump:randomsampling}, the observations are conditionally i.i.d. given $\mathcal{I}_{-nk}$, and $\Omega_{nk}$ has bounded eigenvalues. Then by the Lindeberg-Feller CLT,
	\begin{equation}
	\label{eq:normalconvergence_dmlblp_proof}
	\sqrt{n_k}\Omega_{nk}^{-1/2}\begin{pmatrix} \widehat{\beta}_{2nk}^*-\beta_{2nk}^* \\
	V_{xnk}^{-1}\widehat{V}_{xnk} - 1 \end{pmatrix} \mid \mathcal{I}_{-nk} \to^d \mathcal{N}(0_{2 \times 1},I_{2\times 2}).
	\end{equation} 

	\end{proof}

	\begin{proof}[Proof of Theorem \ref{thm:clt_lvci}.(ii)]		
		
	Let $\tilde{\Delta}_{nk} = I+\Delta_{nk}$, where the non-zero term in $\Delta_{nk}$ is an average with mean zero and variance $V_{xnk}$. Under the assumptions of part (a), $\widehat{\Pi}_{nk}\Pi_{nk} = I+o_p(1)$ and $\Pi_{nk}\tilde{\Delta}_{nk}\Pi_{nk}^{-1} = I + V_{xnk}^{-1/2}\Delta_{nk} = I + o_p(1)$. Decomposing $\widehat{\Pi}_{nk}\uline{\widehat{W}}_{nk}'\Lambda_{nk}\uline{\widehat{W}}_{nk}\widehat{\Pi}_{nk}$ in a similar way to \eqref{eq:decomp_fesasibleestimator},
	$$(\widehat{\Pi}_{nk}\tilde{\Delta}_{nk}'\Pi_{nk}^{-1})(\Pi_{nk}'\uline{W}_{nk}'\Lambda_{nk}\uline{W}_{nk}\Pi_{nk})(\Pi_{nk}^{-1}\tilde{\Delta}_{nk}\widehat{\Pi}_{nk}) = \Pi_{nk}'Q_{ww,nk}\Pi_{nk} + o_p(1)$$.
   Hence $\widehat{J}_{nk} = J_{nk}^* + o_p(1)$. Substituting $\uline{\widehat{W}}_{nk}$, into the upper-left block of $\widehat{H}_{nk}$,  
   \begin{equation} (\widehat{\Pi}_{nk}\tilde{\Delta}_{nk}'\Pi_{nk}^{-1}) \left[\frac{1}{n_k}\sum_{i \in \mathcal{I}_{-nk}} \widehat{U}_i^2\lambda(X_i)^2 \Pi_{nk}'W_iW_i'\Pi_{nk} \right](\Pi_{nk}^{-1}\tilde{\Delta}_{nk}\widehat{\Pi}_{nk}).
   \label{eq:variance_se_H_decomp}
   \end{equation}
	As before the outer terms are $I + o_p(1)$. For the inner terms we apply \eqref{eq:decomp_fesasibleestimator}, $\widehat{U}_{i} = Y_i -\widehat{W}_i'\widehat{\theta}_{nk} = Y_i - W_i'\tilde{\Delta}_{nk}\widehat{\theta}_{nk} = Y_i - W_i'\Pi_{nk}\Pi_{nk}^{-1}\tilde{\Delta}_{nk}\Pi_{nk}\Pi_{nk}^{-1}\widehat{\theta}_{nk} = Y_i - W_i'\Pi_{nk}'\widehat{\theta}_{nk}^* = U_i - W_i'\Pi_{nk}(\widehat{\theta}_{nk}^*-\theta_{nk}^*)$.  The inner term of \eqref{eq:variance_se_H_decomp} is $\mathbb{E}[U_i^2\lambda(X_i)^2 \Pi_{nk}'W_iW_i'\Pi_{nk} \mid \mathcal{I}_{-nk}] + o_p(1)$ under Assumption  \ref{assump:unconfound_sutva}.(ii) (overlap) and \ref{assump:momentbounds_dml} (bounds on moments of $U_i$ and $\Pi_{nk}W_i$). In particular, \ref{assump:momentbounds_dml}.(iv) ensures that $\mathbb{E}[\Vert \Pi_{nk}W_i \Vert^4 \mid \mathcal{I}_{-nk}]$ is uniformly bounded.
	
	Define $\bar{S}_{nk} := \frac{1}{n_k}\sum_{i \in \mathcal{I}_{-nk}} S_{-k}(X_i)$ and $\widehat{S}_{-k}(X_i) = \widehat{\tau}_{-k}(X_i) - \frac{1}{n_k}\sum_{i\in \mathcal{I}_{-nk}}\widehat{\tau}_{-k}(X_i)$. Adding/subtracting the mean,  $\widehat{S}_{-k}(X_i) = S_{-k}(X_i) - \bar{S}_{nk}$. Algebraically, $\frac{1}{n_k}\sum_{i \in \mathcal{I}_{nk}}\widehat{T}_{i}^2 = \widehat{V}_{xnk}^{-2}\left( \frac{1}{n_k}\sum_{i \in \mathcal{I}_{nk}}[\widehat{S}_{-k}(X_i)^2 -\widehat{V}_{xnk}]^2\right) = \widehat{V}_{xnk}^{-2} \left( \frac{1}{n_k}\sum_{i \in \mathcal{I}_{nk}}\widehat{S}_{-k}(X_i)^4 - \widehat{V}_{xnk}^2 \right)$. This simplifies to $ \widehat{V}_{xnk}^{-2} \left( \frac{1}{n_k}\sum_{i \in \mathcal{I}_{nk}}[S_{-k}(X_i)-\bar{S}_{nk}]^4  \right) - 1 $. By a binomial expansion,
		\begin{align*}
		\frac{1}{n_k}\sum_{i \in \mathcal{I}_{nk}}\frac{S_{-k}(X_i)^4}{\widehat{V}_{xnk}^{2}}
		&\quad + \underbrace{\left(\frac{V_{xnk}}{\widehat{V}_{xnk}}\right)^2}_{\to^p 1}  \sum_{\ell = 1}^4 {4 \choose \ell } (-1)^{\ell}\underbrace{\frac{\bar{S}_{nk}^{\ell}}{V_{xnk}^{\ell/2}}}_{o_p(1)}\underbrace{\left[\frac{1}{n_k}\sum_{i \in \mathcal{I}_{nk}}V_{xnk}^{2-\ell/2}S_{-k}(X_i)^{4-\ell}\right]}_{O_p(1)} - 1.
	\end{align*}
	By Assumption \ref{assump:momentbounds_dml}.(iv), the bound on $V_{xnk}$, and the moment bounds for $V_{xnk}^{-1/2}S_{-k}(X_i)$, then for $\ell \in \{1,2,3,4\}$, $\left[\frac{1}{n_k}\sum_{i \in \mathcal{I}_{nk}}V_{xnk}^{2-\ell/2}S_{-k}(X_i)^{4-\ell}\right]$ is $O_p(1)$. By the weak law of large number, $V_{xnk}^{-1/2}\bar{S}_{nk} = o_p(1)$ and by \eqref{eq:normalconvergence_dmlblp_proof}, $\widehat{V}_{xnk}/V_{xnk} \to^p 1$. Hence $ \frac{1}{n_k}\sum_{i \in \mathcal{I}_{nk}}\widehat{T}_{i}^2 = \mathbb{E}[V_{xnk}^{-1}S_{-k}(X_i)^4 \mid \mathcal{I}_{-nk}]-1 + o_p(1) = \mathbb{V}(V_{xnk}^{-1/2}S_{-k}(X_i)^2) + o_p(1)$. This shows that the diagonals of $\widehat{H}_{nk}$ converge to their population analogs. The proof of convergence for the off-diagonals is similar. By the form of $Q_{ww,nk}$ in \eqref{eq:designmatrix_tildew_normalized}, $\Upsilon J_{nk}^{-1} = \Upsilon$. Hence $\Upsilon\widehat{J}_{nk}^{-1}\widehat{H}_{nk} \widehat{J}_{nk}\Upsilon' = \Upsilon J_{nk}^{^*-1}H_{nk} J_{nk}^*\Upsilon' + o_p(1) = \Upsilon H_{nk} \Upsilon' + o_p(1) = \Omega_{nk} + o_p(1)$.
		
	\end{proof}

\begin{proof}[Proof of Lemma \ref{lem:convergence_lvci_to_vci}]
	
	We start by proving that $\Vert V_{\tau}(\gamma) - V_{\tau}^*(\gamma, \mathcal{I}_{-nk}) \Vert \le V_{\tau}$. Since the second term of \eqref{eq:defn_pseudovcate} is always non-negative, $V_{\tau}^*(\gamma,\mathcal{I}_{-nk}) \le V_{\tau}(\gamma)$.  Also, since $\beta_1 =\tau_{\gamma,av}$ and $\beta_2 = 0$ are part of the feasible set, then the second term is at most $V_{\tau}(\gamma) = \mathbb{E}_\gamma[(\tau_\gamma(X)-\tau_{\gamma,av})^2]$. This shows that $V_{\tau}^*(\gamma,\mathcal{I}_{-nk}) \ge 0$.
	
	To prove the second bound, we examine the solution to the best linear projection. Since the folds are independent, conditioning on $\mathcal{I}_{-nk}$ is only important to be able to handle $S_{-k}(x)$ as a deterministic function. For ease of exposition, let $F$ denote the conditional distribution of $\gamma$ given $\mathcal{I}_{-nk}$, and let $\Vert g \Vert_{F,2} = \sqrt{\mathbb{E}_F[\Vert g(X) \Vert^2]}$ denote the $L_2$ norm.
	 Since $\mathbb{E}_F[S_{-k}(X)] = 0$, the optimal solution is $\beta_{1F}^* = \tau_{F,av}$. When $\mathbb{V}_F(S_{-k}(X)) = 0$, the optimal $\beta_{2F}^*$ is indeterminate, and $V_{\tau}^*(\gamma,\mathcal{I}_{-nk}) = 0$. Otherwise, 
	$$(\beta_{2F}^* - 1)=  \frac{Cov_F(S_{-k}(X),\tau_{F}(X))}{\mathbb{V}_F(S_{-k}(X))}- 1 =  \frac{Cov_F(S_{-k}(X),\tau_{F}(X)-S_{-k}(X))}{\mathbb{V}_F(S_{-k}(X))}.$$
	By the Cauchy-Schwarz inequality, $(\beta_{2F}^*-1)^2\mathbb{V}_F(S_{-k}(X)) \le \mathbb{E}_{F}[(\tau_\gamma(X)-\tau_{F,av}-S_{-k}(X))^2]$. Rearranging \eqref{eq:defn_pseudovcate} and substituting the optimum:
	\begin{align*} V_{\tau} - V_{\tau}^* =   \mathbb{E}_F[(\tau_\gamma(X)-\tau_{av}-S_{-k}(X) - (\beta_2^*-1)S_{-k}(X))^2].
	\end{align*}
	By the triangle inequality for the $L_2$ norm, and the above Cauchy-Schwartz bound,
	\begin{equation}
	\label{eq:secondbound_pseudovtau_diff}
	V_{\tau}(\gamma) - V_{\tau}^*(\gamma,\mathcal{I}_{-nk}) \le \left[ 2 \Vert \tau_{F}(X)-\tau_{F,av} -S_{-k}(X)\Vert_{F,2} \right]^2.
	\end{equation}
	Substituting $S_{-k}(x) = \widehat{\tau}_{-k}(x)-\mathbb{E}_{F}[\widehat{\tau}_{-k}(X)]$, \eqref{eq:secondbound_pseudovtau_diff} can be rewritten as $4 \Vert (\tau_\gamma(X) -\widehat{\tau}_{-k}(X)) - (\mathbb{E}_F[\tau_\gamma(X)]-\mathbb{E}_{F}[\widehat{\tau}_{-k}(X)]) \Vert^2 \le 16 \Vert \widehat{\tau}_{-k} - \tau_{\gamma} \Vert_{F,2}^2$. By the law of iterated expectations $\omega(\gamma) := \mathbb{E}_{\gamma}[\Vert \widehat{\tau}_{-k}(X) - \tau_{\gamma}(X) \Vert^2] = \mathbb{E}_{\gamma}[\Vert \widehat{\tau}_{-k} - \tau_{\gamma} \Vert_{F,2}^2]$. Combining the two bounds and applying Jensen's inequality,
	$$ \mathbb{E}_\gamma[|V_{\tau}(\gamma)-V_{\tau}^*(\gamma,\mathcal{I}_{-nk})|] \le \mathbb{E}_\gamma[\min\{16 \Vert \widehat{\tau}_{-k} - \tau \Vert_{F,2}^2,V_{\tau}(\gamma)\}] \le \min\{16\omega(\gamma),V_{\tau}(\gamma)\}.$$

\end{proof}

\begin{proof}[Proof of Theorem \ref{thm:superconsistency} ]
	Decompose $\Delta_{nk} = \Delta_{nk1} + \Delta_{nk2}$, 	where $\Delta_{nk1} := V_{\tau}(\gamma_n) - V_{\tau}^*(\gamma_n,\mathcal{I}_{-nk})$, and $ \Delta_{nk2} :=  \widehat{V}_{\tau nk}-V_{\tau}^*(\gamma_n,\mathcal{I}_{-nk})$. First, by Lemma \ref{lem:convergence_lvci_to_vci} and Assumption \ref{assump:convergencenuisance}, $\mathbb{E}_{\gamma_n}[\Vert \Delta_{nk1} \Vert] \le \min\{16\omega(\gamma_n)^2,V_{\tau}(\gamma_n)\} = o(n_k^{-1/2})$. By Markov's inequality, $\Delta_{nk1} = o_p(n_k^{-1/2})$. Second, conditional on a sequence $\{\mathcal{I}_{-nk}\}_{n=1}^{\infty}$, by Theorem \ref{thm:clt_lvci} and Lemma \ref{lem:decomposition_remainder_lvci_reg},  $\Vert \Delta_{nk2} \Vert = O_p\left( \max\left\{\frac{1}{n_k},\sqrt{V_\tau^*(\gamma_n,\mathcal{I}_{-nk})}/\sqrt{n_k} \right\} \right)$. Since $V_{\tau}^*(\gamma_n,\mathcal{I}_{-nk}) \le V_{\tau}(\gamma_n) = o(1)$, then almost surely, for all $\varepsilon > 0$, $\psi_{n,\varepsilon}(\mathcal{I}_{-nk}) := \mathbb{P}_{\gamma_n}(n_k^{1/2}\Vert \Delta_{nk2} \Vert > \varepsilon  \mid \mathcal{I}_{-nk}) \to 0$. By iterated expectations and the bounded convergence theorem, $\lim_{n_k \to \infty} \mathbb{P}_{\gamma_n}(n_k^{1/2}\Vert \Delta_{nk2} \Vert < \varepsilon ) = \lim_{n_k \to \infty} \mathbb{E}_{\gamma_n}[\psi_{n,\varepsilon}(\mathcal{I}_{-nk})]  =  \mathbb{E}_{\gamma_n}[\lim_{n_k \to \infty}\psi_{n,\varepsilon}(\mathcal{I}_{-nk})] = 0$.
	
	If, in addition, $n_k^{1/2 + \rho}V_{\tau}(\gamma_n)\to 0$ for $\rho \in [0,1/2)$, then $n_k^{1/2+\rho}\Delta_{nk1} = o_p(1)$ (the second bound dominates). Since $n_k^{\rho}V_{\tau}(\gamma_n) \to 0$ as well and $\rho < 1/2$, then conditional on $\mathcal{I}_{-nk}$, $n_k^{1/2+\rho}\Delta_{nk2} = o_p(1)$. We can apply the bounded convergence theorem once again to show that $\mathbb{P}_{\gamma_n}(n_k^{1/2+\rho}\Vert \Delta_{nk2} \Vert < \varepsilon ) = o(1)$.
	
\end{proof}

\begin{proof}[Proof of Theorem \ref{thm:rootn_consistency_and_efficiency}]
	Consider a sequence of distributions $\{\gamma_n\}_{n=1}^\infty$, with associated nuisance functions $(\tau_n(x),\mu_{0n}(x),p_n(x),\tau_{n,av})$.
	
	\textbf{Case 1 (Near Homogeneity):} $V_{\tau}(\gamma_n) \to 0$. By the triangle inequality,	
	\begin{align}
	\label{eq:triangleineq_vhattau_vtaun_nearhomogeneity_dml}
	\begin{split}
	\sqrt{n_k}\left\Vert \widehat{V}_{\tau nk}- \frac{1}{n_k}\sum_{i \in \mathcal{I}_{nk}}\varphi_i \right\Vert 
	&\quad\le \underbrace{\sqrt{n_k} \left\Vert \widehat{V}_{\tau nk} -V_{\tau}(\gamma_n) \right\Vert}_{\xi_{n1}} +
	\underbrace{\sqrt{n_k} \left\Vert \frac{1}{n_k}\sum_{i \in \mathcal{I}_{nk}}\varphi_i - V_{\tau}(\gamma_n) \right\Vert}_{\xi_{n2}}.
	\end{split}
	\end{align}\
	By Theorem \ref{thm:superconsistency}, $\xi_{n1} = o_p(1)$. Since the $\varphi_i$ are i.i.d with mean $V_{\tau}(\gamma_n)$, then $\mathbb{E}_{\gamma_n}[\xi_{n2}^2] = \mathbb{V}_{\gamma_n}(\varphi_i)$. Let $\sigma_{dn}^2(x) = \mathbb{V}_{\gamma_n}(Y_i \mid X_i=x,D_i=d)$ and $p_n(x) = \mathbb{P}_{\gamma_n}(D_i \mid X_i =x)$. By Lemma \ref{lem:variance_eff_influence} and Corollary \ref{cor:bound_varefficient}, 
	$ \mathbb{V}_{\gamma_n}(\varphi_i) \le  \kappa^2 V_{\tau}(\gamma_n)^2+ 4\kappa V_{\tau}(\gamma_n)\sqrt{\mathbb{E}_{\gamma_n}\left[ \left( \frac{\sigma_{1n}^2(X_i)}{p_n(X_i)}+\frac{\sigma_{0n}^2(X_i)}{1-p_n(X_i)} \right) \right]}$. Since $V_{\tau}(\gamma_n) \to 0$ as $n \to \infty$, then $\mathbb{E}[\xi_{n2}^2] \to 0$ and hence $\xi_{n2} = o_p(1)$. \\
	
	\textbf{Case 2 (Strong Heterogeneity):} $V_{\tau}(\gamma_n) \to V_{\tau} \in (0,\infty)$.  For fixed $(x,d)$, the cross-fitted regressors are $\widehat{W}(x,d) = [ 1, \widehat{\mu}_{0,-k}(x) + p_n(x)\widehat{\tau}_{-k}(x), (d-p_n(x)), (d-p(x))\widehat{S}_{-k}(x)]$. Let $e_\ell $ be a $4 \times 1$ vector with a one in the $\ell^{th}$ coordinate and zero otherwise. Analogous to \eqref{eq:pseudonuisancefunctions}, define the regression adjusted nuisance functions as
	\begin{align}
		\widehat{\eta}_{-k,\widehat{\theta}_{nk}}(x) &=  [e_1 (\widehat{W}(x,1) -\widehat{W}(x,0))' +e_2 \widehat{W}(x,0)']\widehat{\theta}_{nk} + e_3 p_n(x) + e_4 \widehat{\tau}_{nk,av}.
	\label{eq:pseudonuisancefunctions_dml}
	\end{align}
	By applying Lemma \ref{lem:equivalence_efficient_simple} to the subsample in $i \in \mathcal{I}_{nk}$,
	$$ \widehat{V}_{\tau n k} = \frac{1}{n_k}\sum_{i \in \mathcal{I}_k} \varphi(Y_i,D_i,X_i,\widehat{\eta}_{-k,\widehat{\theta}_{nk}}), $$	
	where $\eta_{-k,\widehat{\theta}_{nk},r}(x) := \eta(x) + r (\eta_{-k,\widehat{\theta}_{nk}}(x)-\eta(x))$. The true nuisance function $\eta$ depends on the distribution indexed by $n$, but we drop the subscript to simplify the notation. By a second-order term in the Taylor expansion around $r = 0$, for some $\tilde{r} \in (0,1)$,
	
	\begin{align}
	\label{eq:taylorexpansion_dml}
	\begin{split}
	&\sqrt{n_k}(\widehat{V}_{\tau n k}-V_\tau(\gamma_n)) = \frac{1}{\sqrt{n_k}}\sum_{i\in \mathcal{I}_{nk}} [\varphi(Y_i,D_i,X_i,\eta)-V_{\tau}(\gamma_n)]\\
	&\qquad + \frac{1}{\sqrt{n_k}}\sum_{i\in \mathcal{I}_{nk}} \frac{\partial \varphi(Y_i,D_i,X_i,\eta)}{\partial \eta'}(\eta_{-k,\widehat{\theta}_{nk}}(X_i)-\eta(X_i)) \\
	&\qquad + \frac{1}{n_k}\sum_{i=1}^n \sqrt{n_k} (\eta_{-k,\widehat{\theta}_{nk}}(X_i)-\eta(X_i))'\frac{\partial^2 \varphi(Y_i,D_i,X_i,\eta_{-k,\widehat{\theta}_{nk},\tilde{r}})}{\partial \eta'}(\eta_{-k,\widehat{\theta}_{nk}}(X_i)-\eta(X_i)).
	\end{split}
	\end{align}
	Our ultimate goal is to show that the second and third terms of the expansion are $o_p(1)$. To keep the notation concise, let $\varphi_i(\eta) := \varphi(Y_i,D_i,X_i,\eta)$. One of the main challenges is that the nuisance functions are estimated in multiple steps, combining information from the  $\mathcal{I}_{nk}$ and $\mathcal{I}_{-nk}$ subsamples. The key is to decompose these different sources of uncertainty. Define $(\widehat{\lambda}_{nk}-\lambda_{nk}) := [(\widehat{\theta}_{nk}-\theta_{n}),\widehat{\tau}_{nk,av}(\widehat{\theta}_{nk}-\theta_{n}),(\widehat{\tau}_{nk,av}-\tau_{n,av})]'$, where $\theta_n = (0,1,\tau_{n,av},1)$.  By Lemma \ref{lem:decomposition_nuisance}, there exist matrices $\Psi_{-nk}(x)$ and $\Delta_{-nk}(x)$ that are $(x,\mathcal{I}_{-nk})$-measurable, such that for all $x \in \mathcal{X}$,
	\begin{equation}
	\label{eq:decomp_nuisance_dml_eff_proof}
	\widehat{\eta}_{-k,\widehat{\theta}_{nk}}(x) - \eta(x) = \Psi_{-nk}(x)(\widehat{\lambda}_{nk}-\lambda_{nk}) + \Delta_{-nk}(x).
	\end{equation}
	Let $e_\ell$ be a $4 \times 1$ elementary basis vector. Lemma \ref{lem:decomposition_nuisance} shows that  $e_3'[\widehat{\eta}_{-k,\widehat{\theta}_{nk}}(x) - \eta(x)] = 0$ and $ e_3' \Delta_{-nk}(x) = 0$ because the experimental propensity scores are known. Lemma \ref{lem:decomposition_nuisance} also shows that $ e_4' \Delta_{-nk}(x) = 0$ by construction. Substituting \eqref{eq:decomp_nuisance_dml_eff_proof},
	\begin{align*}
	&\frac{1}{\sqrt{n_k}}\sum_{i \in \mathcal{I}_{nk}}\frac{\partial \varphi_i(\eta)}{\partial \eta'}(\eta_{-k,\widehat{\theta}_{nk}}(X_i)-\eta(X_i)) 
	=\frac{1}{\sqrt{n_k}}\sum_{i \in \mathcal{I}_{nk}} \sum_{\ell =1}^4 \frac{\partial \varphi_i(\eta)}{\partial \eta_{\ell}}e_\ell'(\eta_{-k,\widehat{\theta}_{nk}}(X_i)-\eta(X_i)) \\
	&= \underbrace{\left[\frac{1}{\sqrt{n_k}}\sum_{i \in \mathcal{I}_{nk}}\sum_{\ell \ne 3}\frac{\partial \varphi_i(\eta)}{\partial \eta_{\ell}}e_\ell' \Psi_{-nk}(X_i)\right]}_{\Lambda_{1nk}}(\widehat{\lambda}_{nk}-\lambda_{nk}) + \underbrace{\left[\frac{1}{\sqrt{n_k}}\sum_{i \in \mathcal{I}_{nk}}\sum_{\ell =1}^2\frac{\partial \varphi_i(\eta)}{\partial \eta_{\ell}}e_\ell'\Delta_{-nk}(X_i) \right]}_{\Lambda_{2nk}}.
	\end{align*}
	Lemma \ref{lem:convergence_regression_parameters} implies that $\widehat{\lambda}_{nk}-\lambda_{nk} = o_p(1)$.\\
	
	(i) Prove that  $\Lambda_{1nk} = O_p(1)$: By Lemma \ref{lem:derivative_influence_function}.(a) the conditional mean of $\frac{\partial \varphi_i(\eta)}{\partial (\eta_{1},\eta_2,\eta_4)'}$ given $(x,\mathcal{I}_{-nk})$ is  $[0,0,-2(\tau(x)-\tau_{n,av})]$ and by Lemma  \ref{lem:decomposition_nuisance}, $e_4'\Psi_{-nk} = c'$. Also, $\mathbb{E}_{\gamma_n}[ \tau_n(X_i) \mid \mathcal{I}_{-nk}] = \mathbb{E}_{\gamma_n}[ \tau_n(X_i)] = \tau_{n,av}$ by fold independence.
	By the law of iterated expectations, $\mathbb{E}_{\gamma_n}\left[ \frac{\partial \varphi_i(\eta)}{\partial \eta_\ell }e_\ell' \Psi_{-nk}(X_i) \mid \mathcal{I}_{-nk}\right] = 0$ for $\ell \in \{1,2,4\}$. 	By Assumption \ref{assump:regularityconditions}, $\mathbb{E}_{\gamma_n}[\Vert \eta(X_i)\Vert^4]$ and $\mathbb{E}_{\gamma_n}\left[ \left\Vert U_i \right\Vert^4 \right]$ are uniformly bounded by a constant $C < \infty$. By Assumption \ref{assump:unconfound_sutva}.(ii), $p_n(x)$ is contained in $[\delta,1-\delta]$.  Applying Lemma  \ref{lem:derivative_influence_function}.(b), $\mathbb{E}_{\gamma_n}\left[ \left\Vert \partial \varphi_i(\eta)/\partial \eta_\ell \right\Vert^4 \right]^{1/4} \le (16/\delta)C < \infty$. By Lemma \ref{lem:decomposition_nuisance}.(d),  and the triangle inequality, $\mathbb{E}_{\gamma_n}\left[ \left\Vert e_\ell' \Psi_{-nk}(X_i)\right\Vert^4\right]^{1/4} \le C [1+\mathbb{E}_{\gamma_n}[\Vert \widehat{\eta}_{-k}(X_i)\Vert^4]^{1/4}]$. By the bound in Assumption \ref{assump:regularityconditions}.(iii) and the Cauchy-Schwartz inequality, $\mathbb{E}_{\gamma_n}\left[\Vert \frac{\partial \varphi_i(\eta)}{\partial \eta_{\ell}}e_\ell' \Psi_{-nk}(X_i)\Vert \right]<\infty$.  We can combine these moment bounds to apply Lemma \ref{lem:convergence_standardizedsums}.(a), hence $\Lambda_{1nk}= O_p(1)$. \\
	
	(ii) Prove that $\Lambda_{2nk} = o_p(1)$: By Lemma \ref{lem:derivative_influence_function}.(a), $\mathbb{E}_{\gamma_n}\left[ \frac{\partial \varphi_i(\eta)}{\partial \eta_{\ell}}e_\ell' \Delta_{-nk}(X_i) \mid \mathcal{I}_{-nk}\right] = 0 $ for $\ell \in \{1,2\}$. Using similar arguments $\mathbb{E}_{\gamma_n}\left[ \left\Vert \partial \varphi_i(\eta)/\partial \eta_{\ell}\right\Vert^4\right] < C$. By Lemma \ref{lem:decomposition_nuisance}.(b),  $\Vert \Delta_{nk}(x) \Vert \le C \times \Vert \widehat{\eta}_{-k}(x) - \eta(x) \Vert$. By Assumption \ref{assump:regularityconditions}.(ii), $\mathbb{E}_{\gamma_n}\left[ \Vert \widehat{\eta}_{-k}(x) - \eta(x) \Vert^4\right] = o(1)$, which means that $\mathbb{E}_{\gamma_n}\left[\left\Vert \frac{\partial \varphi_i(\eta)}{\partial \eta_{\ell}}e_\ell' \Psi_{-nk}(X_i)\right\Vert^2\right]$ is $o(1)$. By Lemma  \ref{lem:convergence_standardizedsums}.(b), $\Lambda_{2nk} = o_p(1)$.\\

	Let $\Xi = [e_1,e_2,e_4]$ be a $4 \times 3$ matrix such that $\Vert \Xi \Vert \le 1$. Let $\Lambda_{3nk}$ be the second-order on the right-hand side of \eqref{eq:taylorexpansion_dml}. Since the propensity score is known, 	
	\begin{align*}
	&	\Lambda_{3nk} :=  \\
	& \frac{1}{n_k}\sum_{i \in \mathcal{I}_{nk}} \sqrt{n_k} (\eta_{-k,\widehat{\theta}_{nk}}(X_i)-\eta(X_i))'\Xi\frac{\partial^2 \varphi_i(\eta_{-k,\widehat{\theta}_{nk},\tilde{r}})}{\partial (\eta_1,\eta_2,\eta_4)\partial (\eta_1,\eta_2,\eta_4)'}\Xi'(\eta_{-k,\widehat{\theta}_{nk}}(X_i)-\eta(X_i)).
	\end{align*}
	(iii) Prove that $\Lambda_{3nk} = o_p(1)$: By Lemma \ref{lem:derivative_influence_function}.(d), $\left\Vert \frac{\partial^2 \varphi_i(\eta_{-k,\widehat{\theta}_{nk},\tilde{r}})}{\partial (\eta_1,\eta_2,\eta_4)\partial (\eta_1,\eta_2,\eta_4)'}\right\Vert \le 18 \times \sqrt{3}/\delta$. For scalars, $a,b \in \mathbb{R}$, $(a+b)^2 \le 4(a^2 + b^2)$. By \eqref{eq:decomp_nuisance_dml_eff_proof} and the triangle inequality,  $ \Vert \Lambda_{3nk} \Vert \le 4\times \frac{18\sqrt{3}}{\delta}\left( \sqrt{n_k}\Vert \widehat{\lambda}_{nk}-\lambda_{nk}\Vert^2 \left[ \frac{1}{n_k}\sum_{i \in \mathcal{I}_{nk}} \Vert \Psi_{-nk}(X_i)\Vert^2 \right] + \left[ \frac{1}{n_k}\sum_{i \in \mathcal{I}_{nk}} \sqrt{n_k}\Vert \Delta_{-nk}(X_i)\Vert^2 \right] \right)$.  By Lemma \ref{lem:convergence_regression_parameters}, $\sqrt{n_k}\Vert \widehat{\lambda}_{nk}-\lambda_{nk}\Vert^2 = o_p(1)$. By Lemma \ref{lem:decomposition_nuisance}.(d), $\Vert \Psi_{-nk}(x) \Vert \le C[1+ 2\Vert \widehat{\eta}_{-k}(x) \Vert]$. By the moment bound in Assumption \ref{assump:regularityconditions}.(i) and Markov's inequality, then $\frac{1}{n_k}\sum_{i \in \mathcal{I}_{nk}} \Vert \Psi_{-nk}(X_i)\Vert^2 = O_p(1)$. By Lemma \ref{lem:decomposition_nuisance}.(e), $\Vert \Delta_{-nk}(x) \Vert \le C \times \Vert \widehat{\eta}_{-k}(x) - \eta(x) \Vert$. By Assumption \ref{assump:regularityconditions}.(iii), $\sqrt{n_k}\mathbb{E}_{\gamma_n}[\Vert\widehat{\eta}_{-k}(X_i)-\eta(X_i)\Vert^2] = o(1)$, then by Lemma \ref{lem:convergence_standardizedsums}.(c), $\frac{1}{n_k}\sum_{i \in \mathcal{I}_{nk}} \sqrt{n_k}\Vert \Delta_{-nk}(X_i)\Vert^2 = o_p(1)$. Combining these results, $\Lambda_{3nk} = o_p(1)$.
	
\end{proof}

\begin{proof}[Proof of Theorem \ref{thm:exact_coverage_foldlvci}]	
	Let $\rho(\gamma,\mathcal{I}_{-nk}) = \mathbb{P}_{\gamma}(V_{\tau}^*(\gamma,\mathcal{I}_{-nk}) \in \widehat{CI}_{\alpha nk} \mid \mathcal{I}_{-nk})$ denote the conditional probability that the pseudo-VCATE is contained in the confidence interval. Let $\mathcal{F}(\gamma)$ be the support of $\mathcal{I}_{-nk}$, and let $\mathcal{F}(\gamma,t) \subseteq \mathcal{F}(\gamma)$. Almost surely,
	\begin{align} \inf_{\mathcal{I}_{-nk} \in  \mathcal{F}(\gamma)} \rho(\gamma,\mathcal{I}_{-nk})
	&\le \mathbb{E}_{\gamma}[\rho(\gamma,\mathcal{I}_{-nk})] \le \sup_{\mathcal{I}_{-nk}\in \mathcal{F}(\gamma,t)}\rho(\gamma,\mathcal{I}_{-nk}) + \mathbb{P}_{\gamma}(\mathcal{I}_{-nk} \notin \mathcal{F}(\gamma,t)).
	\label{eq:inequality_uniformcoverage}
	\end{align}
	The left inequality considers the worst-case coverage. The right inequality applies the law of iterated expectations by the event $\mathcal{I}_{-nk} \in \mathcal{F}^*(\gamma,t)$, then bounds $\mathbb{P}_{\gamma}(\mathcal{I}_{-nk} \in \mathcal{F}(\gamma,t)) \le 1$ and $\rho(\gamma,\mathcal{I}_{-nk}) \le 1$ to simplify the expressions.  Applying limits,
	\begin{equation} \underset{n \to \infty}{\lim\inf}\inf_{\gamma \in \Gamma} \inf_{\mathcal{I}_{-nk} \in  \mathcal{F}(\gamma)} \rho(\gamma,\mathcal{I}_{-nk}) \le \underset{n \to \infty}{\lim\inf}\inf_{\gamma \in \Gamma} \mathbb{P}_{\gamma}\left( V_{\tau}^*(\gamma,\mathcal{I}_{-nk}) \in \widehat{CI}_{\alpha n k} \right),
	\label{ew:bound_modifieddistribution}
	\end{equation}
	and an analogous result for the upper bound. The data in $\mathcal{I}_{-nk}$ only affects fold $k$ through the estimated nuisance functions $\widehat{\eta}_{-k}(x)$. To prove uniform coverage we will derive the bounds for a class of distributions where the $(Y_i,D_i,X_i)$ in fold $k$ is distributed as $\gamma$ and the plug-in nuisance functions are deterministic. We will define $\mathcal{F}(\gamma,t)$ as the set where $|V_{\tau}^*(\gamma,\mathcal{I}_{-nk})\Omega_{nk,12}| \le t$ for some fixed $t > 0$.	
	
	Let $\{\gamma_n\}_{n=1}^{\infty}$ and $\{\mathcal{I}_{-nk}\}_{n=1}^{\infty}$ denote a sequence of distributions and data realizations, respectively.  Theorem \ref{thm:clt_lvci} verifies that under  Assumptions \ref{assump:unconfound_sutva},  \ref{assump:momentbounds_dml}, \ref{assump:boundsquantities}, and \ref{assump:randomsampling}, the conditional CDFs $\{F_{\gamma_n\mid \mathcal{I}_{-nk}} \}_{n=1}^{\infty}$ satisfy the high-level Assumption \ref{assump:clt_regcoef}, almost surely. Furthermore, $\widehat{CI}_{\alpha nk}$ satisfies the form in \eqref{eq:feasible_confidenceset}, substituting $(\widehat{V}_{\tau nk},\widehat{\Omega}_{nk})$. Therefore, Lemma \ref{lem:coverage_generic}.(a) shows that the uniform lower bound on the asymptotic coverage probability is $(1-\alpha)$. If Assumption \ref{assump:exactcoverage} also holds, then for all $t > 0$, $\underset{n \to \infty}{\lim\sup}\sup_{\gamma \in \Gamma} \mathbb{P}_{\gamma}(\mathcal{I}_{-nk} \notin \mathcal{F}^*(\gamma,t)) = 0$. Hence by Lemma \ref{lem:coverage_generic}.(ii), the upper bound is $(1-\alpha)$ plus a term that can be made arbitrarily small by choosing $t$ close to zero.

\end{proof}

\begin{proof}[Proof of Lemma \ref{lem:verifyexactcoverage}]
	Consider an arbitrary sequence of distributions $\{\gamma_n\}_{n=1}^{\infty} \in \Gamma^{\infty}$ and a convergent subsequence where $V_{\tau}(\gamma_{n_\ell}) \to V_{\tau}$. Since $V_{\tau}^*(\gamma_{n_\ell},\mathcal{I}_{-n_\ell k})\le V_{\tau n_\ell}$ and $\Omega_{n_\ell k}$ has bounded eigenvalues, if $V_{\tau} = 0$ then $\lim_{\ell \to \infty} \mathbb{P}_{\gamma_{n_\ell}}(\sqrt{V_{\tau}^*(\gamma_{n_\ell},\mathcal{I}_{-n_\ell k})}|\Omega_{n_\ell k,12}| > t) = 0$. Now suppose that $V_{\tau} > 0$. By Assumption \ref{assump:convergencenuisance}, the nuisance functions converge to their true value. Then $\widehat{\tau}_{-k}(x)$ converges point-wise to $\tau(x)$, and by the moment bound in \ref{assump:momentbounds_dml}.(v) and the dominated convergence theorem, $V_{xn_\ell k} = \mathbb{V}_{\gamma_{n_\ell}}(\widehat{\tau}_{-k}(X) \mid \mathcal{I}_{-n_\ell k}) = V_{\tau}(\gamma_{n_\ell})+o_p(1) \ge V_{\tau} + o_p(1)$. By \eqref{eq:omegank},	
	\begin{equation}
	\label{eq:omegank_proof1}
	\Omega_{n_\ell k} = \mathbb{V}_{\gamma_{n_\ell}}\left(\begin{bmatrix} \lambda_{n_\ell}(X_i)(D_i-p_{n_\ell}(x_i))V_{xn_\ell k}^{-1/2}S_{-k}(X_i) U_{i} \\ 
	V_{xn_\ell k}^{-1}S_{-k}(X_i)^2
	\end{bmatrix} \mid \mathcal{I}_{-n_\ell k} \right),
	\end{equation}
	where $U_{in_\ell } = Y_{i} -  W_{i}'\theta_{n_\ell k}$. By Assumption \ref{assump:regularityconditions}, for fixed $\{D_i=d,X_i = x\}$, $W_{i}'$ point-wise converges to $W_{i}^{*'} := [1, \mu_{0n_\ell}(x) + p_n(x)\tau_{n_\ell}(x), (d-p_{n_\ell}(x)) , (d-p_{n_\ell}(x))\tau_{n_\ell}(x)]$. By Lemma \ref{lem:convergence_regression_parameters}, $\theta_{n_\ell k} \to \theta_{n_\ell} := [0,1,\tau_{n_\ell,av},1]$. By Lemma \ref{lem:projection_weak}, $W_i^{*'}\theta_{n_\ell} = \mu_{d,n_\ell}(x)$, and hence $U_{in_\ell}^* := Y_i - W_i^{*'}\theta_{n_\ell}$. Since $\Omega_{n_\ell k}$ is almost surely bounded by Assumption \ref{assump:momentbounds_dml} over random partitions $\mathcal{I}_{-n_\ell k}$, then applying the dominated convergence theorem,
	\begin{equation}
	\label{eq:omegank_proof2}
	\Omega_{n_\ell k} = \mathbb{V}_{\gamma_{n_\ell}}\left(\begin{bmatrix} \lambda_{n_\ell}(X_i)(D_i-p_{n_\ell }(x_i))V_{\tau}^{-1/2}(\tau_{n_\ell }(X_i)-\tau_{n_\ell ,av}) U_{in_\ell}^* \\ 
	V_{\tau}^{-1}(\tau_{n_\ell}(X_i)-\tau_{n_\ell,av})^2
	\end{bmatrix} \mid \mathcal{I}_{-n_\ell k} \right) + o_p(1),
	\end{equation}
	Since $\mathbb{E}_{\gamma_n}[U_{in_\ell}^* \mid D_i = d,X_i = x,\mathcal{I}_{-n_\ell k}] = 0$ and the second component of \eqref{eq:omegank_proof2} only depends on $X_i$. By iterated expectations, $\Omega_{n_\ell k,12} = o_p(1)$ and the limiting probability is $\lim_{\ell \to \infty} \mathbb{P}_{\gamma_{n_\ell}}(\sqrt{V_{\tau}^*(\gamma_{n_\ell},\mathcal{I}_{-n_\ell k})}|\Omega_{n_\ell k,12}| > t) = 0$. Hence we verified Assumption $B^*$ in \cite{andrews2020generic}. Uniform consistency follows from their Corollary 2.1.
	
\end{proof}

\begin{proof}[Proof of Theorem \ref{thm:exact_coverage_foldlvci_pointwise}]
	We break-down the proof into cases.
	
	Case (i): When $V_{\tau}(\gamma) = 0$, then $V_{\tau}(\gamma,\mathcal{I}_{-nk}) = 0$ almost surely. Therefore, $n_k( \widehat{V}_{\tau n k} -V_{\tau}^*(\gamma_n,\mathcal{I}_{-nk})) = n_k( \widehat{V}_{\tau n k}^*-V_{\tau}(\gamma))$. Define $\rho(\gamma,\mathcal{I}_{-nk})$ as in Theorem \ref{thm:exact_coverage_foldlvci}. Then $\underset{n \to \infty}{\lim\inf}\inf_{\mathcal{I}_{-nk} \in  \mathcal{F}(\gamma)} \rho(\gamma,\mathcal{I}_{-nk}) \le \underset{n \to \infty}{\lim\inf} \mathbb{P}_{\gamma}\left( V_{\tau}^*(\gamma,\mathcal{I}_{-nk}) \in \widehat{CI}_{\alpha n k} \right)$. We can prove an analogous upper bound. This implies that we only need to derive coverage bounds under sequences of conditional distributions where $V_{\tau}^*(\gamma_n,\mathcal{I}_{-nk}) = 0$. Exact coverage follows from the proof of the near-homogeneity case in Lemma \ref{lem:coverage_generic}.
	
	Case (ii) When $V_{\tau}(\gamma) > 0$. By Assumption \ref{assump:convergencenuisance}, $\sqrt{n_k}\omega(\gamma)^2 \to 0$ as $n_k \to \infty$. Then by Lemma \ref{lem:convergence_lvci_to_vci}, $\sqrt{n_k}|  V_{\tau}(\gamma)-V_{\tau}^*(\gamma,\mathcal{I}_{-nk}) | = o_p(1)$, and by the continuous mapping theorem, $\sqrt{V_{\tau}^*(\gamma,\mathcal{I}_{-nk})} \to^p \sqrt{V_{\tau}(\gamma)} > 0$. By \eqref{eq:omegank_proof2} in the proof of Lemma \ref{lem:verifyexactcoverage} and for fixed $\gamma_n = \gamma$, $\Omega_{nk} = \Omega + o_p(1)$, where $\Omega$ is a population covariance matrix that does not depend on $\mathcal{I}_{-nk}$. By applying similar limits to the mild heterogeneity case in Lemma \ref{lem:coverage_generic}  we can show that $(\widehat{V}_{\tau nk}-V_{\tau}(\gamma))$ is $\sqrt{n_k}$ asymptotically equivalent to an empirical process  indexed by the oracle $V_{\tau}(\gamma)$. We obtain exact coverage due to Lemma \ref{lem:verifyexactcoverage}.
\end{proof}

\begin{proof}[Proof of Corollary \ref{cor:uniformity low}]
	By the first part of Theorem \ref{thm:exact_coverage_foldlvci},
	$$ \underset{n \to \infty}{\lim\sup}\sup_{\gamma \in \Gamma} \mathbb{P}_{\gamma}\left( V_{\tau}^*(\gamma,\mathcal{I}_{-nk}) \notin \widehat{CI}_{\alpha n k} \right) = 1-  \underset{n \to \infty}{\lim\inf}\inf_{\gamma \in \Gamma}\mathbb{P}_{\gamma}\left( V_{\tau}^*(\gamma,\mathcal{I}_{-nk}) \in \widehat{CI}_{\alpha n k} \right) \le \alpha.$$
	The result follows taking limits on either side of the inequality in \eqref{eq:boundprob_lowerci}.	
\end{proof}

\begin{proof}[Proof of Theorem \ref{thm:uniform_coverage_variational_lvci} ]
	By the definition in \eqref{eq:multisplit_cis},
	$$ V_{\tau}(\gamma) \ge \inf \widehat{CI}_{\alpha n}^{\text{multifold}} \iff \frac{1}{K}\sum_{k=1}^K \mathbbm{1}\left\{ V_{\tau}(\gamma) \ge \inf \widehat{CI}_{\alpha nk} \right\} \ge \frac{1}{2}.$$
	By negating the statement, computing expectations, and applying Markov's inquality,
	\begin{align*} \mathbb{P}_{\gamma}(V_{\tau}(\gamma) < \inf \widehat{CI}_{\alpha n}^{\text{multifold}}) &=
		\mathbb{E}_{\gamma}\left[ \mathbbm{1}\left\{\left(\frac{1}{K}\sum_{k} \mathbbm{1}\{ V_{\tau}(\gamma) < \inf \widehat{CI}_{\frac{\alpha}{2} nk}\} \right)  > 1/2 \right\} \right] \\
		&\le 2 \frac{1}{K}\sum_{k} \mathbb{E}_{\gamma}[\mathbbm{1}\{ V_{\tau}(\gamma) < \inf \widehat{CI}_{\frac{\alpha}{2} nk}\}] \\
		&\le 2 \mathbb{P}_{\gamma}\left(V_{\tau}(\gamma) < \inf \widehat{CI}_{\frac{\alpha}{2} nk}\right) \le 2 \mathbb{P}_{\gamma}\left(V_{\tau}(\gamma,\mathcal{I}_{-nk}^*) \notin \widehat{CI}_{\frac{\alpha}{2} nk} \right)
	\end{align*}
	The last line follows from the fact that the folds are split at random and the inequality in \eqref{eq:boundprob_lowerci} holds almost surely. By Theorem \ref{thm:exact_coverage_foldlvci} the asymptotic size is uniformly bounded by $2(\alpha/2) = \alpha$.	By construction, $
	\mathbb{P}_{\gamma}(V_{\tau}(\gamma) \notin \widehat{CI}_{\alpha n}^{\text{multifold}})  = \mathbb{P}_{\gamma}(V_{\tau}(\gamma) < \inf \widehat{CI}_{\alpha n}^{\text{multifold}}) + \mathbb{P}_{\gamma}(V_{\tau}(\gamma) > \sup \widehat{CI}_{\alpha n}^{\text{multifold}}) $. Applying similar arguments as before, $ \mathbb{P}_{\gamma}(V_{\tau}(\gamma) \notin \inf \widehat{CI}_{\alpha n}^{\text{multifold}}) \le  2 \mathbb{P}_{\gamma}\left(V_{\tau}(\gamma) \notin \widehat{CI}_{\frac{\alpha}{2} nk}\right)$. If Assumptions \ref{assump:convergencenuisance}, and \ref{assump:regularityconditions} also hold, then Theorem \ref{thm:exact_coverage_foldlvci_pointwise} implies that the right-hand side is point-wise bounded by $\alpha$.
	
\end{proof}

\begin{proof}[Proof of Lemma \ref{lem:powerlocal_alternatives}]
	Under the null hypothesis, $G(n_k,0,\widehat{\Omega}_{nk},Z, \zeta) = (e_1'\widehat{\Omega}_{nk}^{1/2}Z)^2/n_k$, where $\widehat{\Omega}_{nk,11}^{1/2}Z$, where $Z \sim \mathcal{N}(0,1)$. The adjusted critical values are $q_{\alpha/2}(n_k,0,\widehat{\Omega}_{nk},\zeta) = 0$ and  $q_{1-\alpha/2}(n_k,0,\widehat{\Omega}_{nk},\zeta) = \widehat{\Omega}_{nk,11} z_{1-\alpha}^2 / n_k$. Then $0 \in \widehat{CI}_{\alpha nk}$ if and only if $ 0 \le \widehat{V}_{\tau n k} - 0 \le \widehat{\Omega}_{nk,11} z_{1-\alpha}^2 / n_k$. Following similar steps to the ``near homogeneity'' regime in Theorem \ref{thm:exact_coverage_foldlvci}, $n_k(\widehat{V}_{\tau n k} -V_{\tau}^*(\gamma_n,\mathcal{I}_{-nk})) = (\Omega_{\infty,11}^{1/2}Z_{nk} + \sqrt{v})^2 - v + o_p(1)$, where $Z_{nk} \sim \mathcal{N}(0,1)$. Consequently, $n_k\widehat{V}_{\tau n k} = (\Omega_{\infty,11}^{1/2}Z_{nk} + \sqrt{v})^2 + o_p(1)$. Then
	\begin{align*} \mathbb{P}_{\gamma_n}(0 \in \widehat{CI}_{\alpha nk}\mid \mathcal{I}_{-nk}) &= \mathbb{P}_{\gamma_n}\left(0 \le n_k\widehat{V}_{\tau nk}  \le  \widehat{\Omega}_{nk,11} z_{1-\alpha}^2  \right) \\
	&= \mathbb{P}_{\gamma_n}\left( 0 \le (\Omega_{\infty,11}^{1/2}Z_{nk}+ \sqrt{v})^2  \le   \Omega_{\infty,11} z_{1-\alpha}^2 \right) + o(1) \\
	&= \mathbb{P}_{\gamma_n}\left( -\frac{\sqrt{v}}{\sqrt{\Omega_{\infty,11}}} \le Z_{nk} +o_p(1)  \le  z_{1-\alpha}-\frac{\sqrt{v}}{\sqrt{\Omega_{\infty,11}}}  \right) + o(1).
	\end{align*}
	and hence $ \lim_{n \to \infty} \mathbb{P}_{\gamma_n}(0 \in \widehat{CI}_{\alpha nk}\mid \mathcal{I}_{-nk}) = \Phi\left(z_{1-\alpha}-\sqrt{v}/\sqrt{\Omega_{\infty,11}}\right) - \Phi\left( -\sqrt{v}/\sqrt{\Omega_{\infty,11}} \right)$. The final result is obtained by $1-\mathbb{P}_{\gamma_n}(0 \in \widehat{CI}_{\alpha nk}\mid \mathcal{I}_{-nk})$.
\end{proof}

\begin{proof}[Proof of Lemma \ref{lem:clustered_se}]
	The first part of the proof is identical to that of Theorem \ref{thm:exact_coverage_foldlvci} in terms of setting up the problem, and defining a sequence the conditional distributions $\{F_{\gamma_n\mid \mathcal{I}_{-nk}} \}_{n=1}^{\infty}$. By equation \eqref{eq:clt_coefficients_cluster} this sequence satisfies Assumption \ref{assump:clt_regcoef} almost surely with an effective sample size $\tilde{n} = n/ r_{n}$ and a particular choice of covariance estimator. To complete the proof, we develop a modified version of Lemma \ref{lem:coverage_generic} to prove uniform coverage under cluster dependence. 
	Consider a sequence of  distributions $\{\gamma_n\}_{n=1}^\infty \in \Gamma^{\infty}$ and  a subsequence $\{n_\ell\}_{\ell=1}^{\infty}$, where $\frac{n_\ell}{r_{n_\ell}} V_{\tau n_\ell}^* \to^p v \in [0,\infty)$, $V_{\tau n_\ell}^* \to^p 0$, and $\Omega_{n_\ell} \to \Omega$ as $n_{\ell} \to \infty$. Applying Lemma \ref{lem:decomposition_remainder_lvci_reg} , and factorizing terms as in Lemma \ref{lem:coverage_generic},	
	\begin{equation} \frac{n_\ell}{r_{n_\ell}}( \widehat{V}_{\tau n_\ell}^*-V_{\tau n_\ell}^*) =  (e_1'\Omega^{1/2}Z_{n_\ell} + \sqrt{v})^2 - v + o_p(1).\label{eq:cluster_limit_nearhomogeneity}
	\end{equation}
	Now we need to show that the quantiles of the empirical process converge to the same limit. The quantity $\frac{n_\ell}{r_{n_\ell}} G(n_\ell ,V_{\tau n_\ell}^*,\widehat{\Omega}_{n_\ell},Z)$ is equal to 
	\begin{equation}
	n_\ell \frac{(e_1'r_{n_\ell}^{-1/2}\widehat{\Omega}_{n_\ell}^{1/2}Z)^2}{n_\ell} + 2 \sqrt{\frac{n_\ell V_{\tau n_\ell}^*}{r_{n_\ell}}}(e_1'r_{n_\ell}^{-1/2}\widehat{\Omega}_{n_\ell}^{1/2}Z)+ \sqrt{\frac{n_\ell V_{\tau n_\ell}^*}{r_{n_\ell}}} \sqrt{V_{\tau n_\ell}^*} (e_2'r_{n_\ell}^{-1/2}\widehat{\Omega}_{n_\ell}^{1/2}Z).
	\label{eq:cluster_limit_nearhomogeneity_empiricalprocess}
	\end{equation}
	In the first term the $n_\ell$ components cancel out and $r_{n_\ell}^{-1}\widehat{\Omega}_{n_\ell} \to^p \Omega$. Similarly, since $\frac{n_\ell}{r_{n_\ell}} V_{\tau n_\ell}^* \to^p v$ the second term of \eqref{eq:cluster_limit_nearhomogeneity_empiricalprocess} converges to $2\sqrt{v}e_1'\Omega^{1/2}Z$, and a suitable factorization with the first and second the expression in \eqref{eq:cluster_limit_nearhomogeneity}. Since $n_\ell/r_{n_\ell} \to \infty$ by assumption, then $V_{\tau n_\ell}^* \to 0$ and the third term of \eqref{eq:cluster_limit_nearhomogeneity_empiricalprocess} is $o_p(1)$. Therefore the estimated quantiles are consistent. Proving consistency of the quantiles for the mild hereterogeneity proceeds analogously. Once we prove that the quantiles are asymptotically correct, the rest of the proof is the same as in Lemma \ref{lem:coverage_generic}.		
\end{proof}

\begin{proof}[Proof of Theorem \ref{lem:degenerate_cate} ]
	The first part of the proof is identical to that of Theorem \ref{thm:exact_coverage_foldlvci} in terms of setting up the problem. In this case the sequence of conditional distributions $\{F_{\gamma_n\mid \mathcal{I}_{-nk}} \}_{n=1}^{\infty}$ only satisfies Assumption \ref{assump:clt_regcoef} for subsequences where $V_{xnk} = 0$. Instead, I will modify the first part of the proof of Lemma \ref{lem:coverage_generic}.(i) for a class of regression-adjusted CIs with possible degeneracy. Consider a sequence of  distributions $\{\gamma_n\}_{n=1}^\infty \in \Gamma^{\infty}$ and let $h_n := (n V_{\tau n}^*,V_{\tau n}^*,vec(\Omega_n),\zeta_n,V_{xn})$ be a sequence of parameters where $V_{xn}$ is the variance of $S_{-k}(X_i)$. Our goal is to show that for $h \in \mathcal{H}$ and all subsequences $\{n_\ell\}_{\ell =1}^{\infty}$ where $h_{n_\ell} \to h \in \mathcal{H}$,
	$$ 	\lim_{n_\ell \to \infty} \mathbb{P}_{\gamma_{n_\ell}}\left( V_{\tau n_\ell}^* \in \widehat{CI}_{\alpha n k}^0 \right)  \ge 1-\alpha. $$
	Partition the subsequences in such a way that $h_{n_\ell}$ either has $V_{xn_\ell} = 0$ or $V_{xn_\ell} > 0$. When $V_{xn_\ell} = 0$, then $\widehat{V}_{\tau n_\ell}^*$ is exactly degenerate and the confidence interval \eqref{eq:degenerate_ci} covers the pseudo-VCATE with probability one. For the sequences where $V_{xn_\ell} > 0$, we can apply the remaining cases from Lemma \ref{lem:coverage_generic} which have coverage $1-\alpha$. This satisfies Assumption B in \cite{andrews2020generic} and we can prove uniform conservative coverage of the pseudo-VCATE applying their Corollary 2.1
	
	We prove point-wise coverage by cases. When $V_{\tau}(\gamma) = 0$, then $V_{\tau}^*(\gamma,\mathcal{I}_{-nk}) = 0$ almost surely. By applying the result above and the near homogeneity case in Theorem \ref{thm:exact_coverage_foldlvci} we find that coverage is at least $(1-\alpha)$. Now consider the case where $V_{\tau}(\gamma)$ is bounded away from zero. By Assumption \ref{assump:convergencenuisance}, the nuisance functions converge to their true value. Then $\widehat{\tau}_{-k}(x)$ converges point-wise to $\tau(x)$, and by the moment bound in \ref{assump:momentbounds_dml}.(v) and the dominated convergence theorem, $\mathbb{V}_{\gamma_n}(\widehat{\tau}_{-k}(X) ) = V_{\tau}(\gamma)+o(1)$, which is bounded away from zero. Then we can apply the mild heterogeneity results in Theorem \ref{thm:exact_coverage_foldlvci} to prove \eqref{eq:degenerate_pointwise}. The proof of \eqref{eq:multifold_zero_ci} is identical to that of Theorem \ref{thm:uniform_coverage_variational_lvci}, substituting the confidence intervals $\widehat{CI}_{\alpha nk}^0$ instead of $\widehat{CI}_{\alpha nk}$. Point-wise coverage of the fold-specific confidence intervals holds by \eqref{eq:degenerate_pointwise}. 	
	
\end{proof}

\begin{proof}[Proof of Corollary \ref{cor:one_sidedtests_degenerate}] The proof of \eqref{eq:onesided_degenerate_singlefold} and  \eqref{eq:onesided_degenerate_multifold}  follows a similar structure to  Corollary \ref{cor:uniformity low} and Theorem \ref{thm:uniform_coverage_variational_lvci}, respectively. In each case the only thing that changes is that we apply the uniformity result for degenerate CIs in \eqref{eq:uniformconservativecoverage_degenerate} (Lemma \ref{lem:degenerate_cate}), rather than the non-degenerate uniformity result in Theorem \ref{thm:exact_coverage_foldlvci}.
	
\end{proof}

\section{Supporting Lemmas and Proofs}

\begin{lem}
	\label{lem:coverage_generic}
	Suppose that $\Gamma$ is a set of distributions constrained in such a way that Assumption \ref{assump:clt_regcoef} holds. Let $V_\tau^*(\gamma) = \beta_2(\gamma)^2V_x(\gamma)$ and $\Omega(\gamma)$ be the pseudo-VCATE and covariance matrix associated with $\gamma$, respectively. If $\widehat{CI}_{\alpha n}$ is a confidence interval obtained by substituting $(\widehat{V}_{\tau n},\widehat{\Omega}_n)$ into \eqref{eq:feasible_confidenceset}, then
	\begin{enumerate}[(i)]
		\setlength{\parskip}{0pt}
		\item $1-\alpha \le \underset{n \to \infty}{\lim\inf}\inf_{\gamma \in \Gamma} \mathbb{P}_{\gamma}\left( V_{\tau}^*(\gamma) \in \widehat{CI}_{\alpha n k} \right)$
		\item If in addition, $\Omega_{12}(\gamma)V_{\tau}(\gamma) \le t$, then  $\underset{n \to \infty}{\lim\sup}\sup_{\gamma \in \Gamma} \mathbb{P}_{\gamma}\left( V_{\tau}^*(\gamma,\mathcal{I}_{-nk}) \in \widehat{CI}_{\alpha n k} \right) \le (1-\alpha) + \tilde{\alpha}(t)$, where $\tilde{\alpha}(t) \ge 0$ and $\lim_{t \to 0}\tilde{\alpha}(t) = 0$.
	\end{enumerate}	
\end{lem}

\begin{proof}	
	Let $\{\gamma_n \in \Gamma: n \in \mathbb{N}\}$ denote a sequence of distributions. Our goal is to verify that the confidence interval satisfies the assumptions of Corollary 2.1(c) in \cite{andrews2020generic}. Define a sequence of parameters, $h_n := (n V_{\tau n}^*,V_{\tau n}^*,vec(\Omega_n),\zeta_n)$. By Assumption \ref{assump:clt_regcoef}, each element is contained in $\mathcal{H}$, a subset of the extended Euclidean space in which $\Omega_n$ is positive-definite with bounded eigenvalues. The quantity $nV_{\tau n}^*$ is positive but unbounded, and can converge to $+\infty$. Assumption $B$ in \cite{andrews2020generic} is stated in terms of subsequences and the first step is to write the problem in this way. We show that for $h \in \mathcal{H}$ and all subsequences $\{n_\ell\}_{\ell =1}^{\infty}$ where $h_{n_\ell} \to h \in \mathcal{H}$,
	$$ 	1-\alpha \le \lim_{n_\ell \to \infty} \mathbb{P}_{\gamma_{n_\ell}}\left( V_{\tau n_\ell}^* \in \widehat{CI}_{\alpha n} \right)  \le  (1-\alpha) + \tilde{\alpha}(t), \qquad  \tilde{\alpha}(t) \ge 0, \alpha \in [0,1]$$	
	We break down the proof by cases. 	
	\textit{(a) Near homogeneity case:} Suppose that  $n_\ell V_{\tau n_\ell}^* \to^p v \in [0,\infty)$, $V_{\tau n_\ell}^* \to^p 0$, $\zeta_{n_\ell} \to \zeta^* \in \{0,1\}$ and $\Omega_{n_\ell} \to \Omega$ as $n \to \infty$, where $\Omega$ is positive-definite. For this case, $\sqrt{n_\ell}V_{\tau n_\ell}^* = o(1)$. By applying Lemma \ref{lem:decomposition_remainder_lvci_reg},
	\begin{align*} n_\ell( \widehat{V}_{\tau n_\ell }^*-V_{\tau n_\ell}^*) &=   
		(e_1'\Omega^{1/2}Z_{n_\ell})^2 + 2 \zeta^*\sqrt{v}e_1'\Omega^{1/2}Z_{n_\ell} + o_p(1) \\
		&= (e_1'\Omega^{1/2}Z_{n_\ell} + \zeta^* \sqrt{v})^2 - v + o_p(1).
	\end{align*}
	Let $Z \sim \mathcal{N}(0,I_{2\times 2})$ (independent of $n_\ell$). Since $\widehat{\Omega}_{n_\ell} = \Omega + o_p(1)$, then the estimated empirical process at $V_{\tau n_\ell}^*$ has the same limiting distribution as the estimator. For a fixed $\zeta \in \{-1,1\}$ (that may differ from $\zeta^*$),
	\begin{align*} n_\ell G(n,V_{\tau n_\ell}^*,\widehat{\Omega}_{n_\ell},Z,\zeta) 
		&= (e_1'\Omega^{1/2}Z + \zeta \sqrt{v})^2 - v + o_p(1).
	\end{align*}
	Define the limiting CDF as $H_{\Omega,v,\zeta}(\tilde{v}) := \mathbb{P}_{\gamma_{n_\ell}}\left( (e_1'\Omega^{1/2} Z + \zeta\sqrt{v})^2 - v \le \tilde{v} \right)$, where $Z \sim \mathcal{N}(0,I_{2\times 2})$. Since $Z$ has mean zero, $H_{\Omega,v,\zeta = 1}(\tilde{v}) = H_{\Omega,v,\zeta = -1}(\tilde{v}) = H_{\Omega,\nu}(\tilde{v})$, which does not depend on $\zeta$. Since $\Omega$ is positive-definite, the function $H_{\Omega,v}(\tilde{v})$ is continuous and strictly increasing. Let $\widetilde{F}_{n_\ell,V_{\tau n_\ell}^*,\widehat{\Omega}_{n \ell},\zeta_{n_\ell}}(\tilde{v}) = F_{n_\ell,V_{\tau n_\ell}^*,\widehat{\Omega}_{n_\ell},\zeta_{n_\ell}}(\tilde{v}/n_\ell) \to H_{\Omega,v}(\tilde{v})$. Since $H_{\Omega,v}$ is continuous, then \cite[Theorem 2.6.1]{lehmann1999elements} implies that this convergence is uniform in $\tilde{v}$, and since the limiting CDF is strictly increasing,
	$$ F_{n_\ell,V_{\tau}^*,\widehat{\Omega}_{n},\zeta}(\widehat{V}_{\tau n_\ell}-V_{\tau}^*) =  \tilde{F}_{n_\ell,V_{\tau n_\ell}^*,\widehat{\Omega}_{n},\zeta}(n_\ell(\widehat{V}_{\tau n_\ell}-V_{\tau}^*)) \to^d U_{n_\ell},$$
	for all $\zeta \in \{-1,1\}$ where $U_{n_\ell}\sim  \sim Uniform[0,1]$. The test statistic converges to a fixed distribution regardless of the choice of $\zeta$. Similarly, $ q_{\alpha/2}(n_\ell,V_{\tau n_\ell}^*,\Omega,\zeta) \to^p \min\{\alpha/2,H_{\Omega,v}(0)\}$. Define a random variable, $\widehat{R}_{n_\ell,\zeta} := F_{n_\ell,V_{\tau}^*,\widehat{\Omega}_{n},\zeta}(\widehat{V}_{\tau n_\ell}-V_{\tau n_\ell}^*) - q_{\alpha/2}(n_\ell,V_{\tau n_\ell}^*,\widehat{\Omega}_{n_\ell},\zeta)  \to^d U_{n_\ell} - \min\{\alpha/2,H_{\Omega,v}(0)\}$. By definition, $ V_{\tau n_\ell}^* \in \widehat{CI}_{\alpha n} \iff \bigcup_{\zeta \in \{-1,1\}}\left\{\widehat{R}_{n_\ell,\zeta} \in [0,1-\alpha] \right\}.$ 
	As $n_{\ell} \to \infty$,
	\begin{align*}
		\begin{split}
			\mathbb{P}_{\gamma_{n_\ell}}\left( V_{\tau n_\ell}^* \in \widehat{CI}_{\alpha n k} \right) &\ge \max_{\zeta \in \{-1,1\}} \mathbb{P}_{\gamma_{n_\ell}}\left( \widehat{R}_{n_\ell,\zeta} \in [0,1-\alpha]  \right) \ge \mathbb{P}_{\gamma_{n_\ell}}\left( \widehat{R}_{n_\ell,\zeta_{n_\ell}} \in [0,1-\alpha]  \right) \\
			&= \mathbb{P}_{\gamma_{n_\ell}}\left(0\le U_{n_\ell} - \min\{\alpha/2,H_{\Omega,v}(0)\} \le 1-\alpha \right)  + o(1)\\
			&= (1-\alpha) + o(1).
		\end{split}
	\end{align*}
	Since the limiting distribution doesn't depend on $\zeta$, we can apply the continuous mapping theorem to show that $ \widehat{R}_{n_\ell}^{\max} := \min_{\zeta \in \{-1,1\}} \widehat{R}_{n_\ell,\zeta}$ and  $\widehat{R}_{n_\ell}^{\min}:=\max_{\zeta \in \{-1,1\}}\widehat{R}_{n_\ell,\zeta} $ both converge to the same limit, $U_{n_\ell} - \min\{\alpha/2,H_{\Omega,v}(0)\}$. As $n_{\ell} \to \infty$,
	\begin{align}
		\begin{split}
			\mathbb{P}_{\gamma_{n_\ell}k}\left( V_{\tau n_\ell}^* \in \widehat{CI}_{\alpha n k} \right) &\le \mathbb{P}_{\gamma_{n_\ell}}\left( 0 \le \widehat{R}_{n_\ell}^{\max},  \widehat{R}_{n_\ell}^{\min} \le 1-\alpha  \right) 
			= (1-\alpha) + o(1).
		\end{split}
		\label{eq:upperbound_probci}
	\end{align}
	For this class of subsequences the confidence interval produces exact coverage.
	
	\textit{Mild heterogeneity case:} Suppose that $n_\ell V_{\tau n_\ell}^* \to \infty$ and $V_{\tau n_\ell}^* \to V_{\tau}^*$, where $V_{\tau} \in [0,\infty)$, $\zeta_{n_\ell} \to \zeta^*$, and $\Omega_{n_\ell} \to \Omega$. Then we can rescale by $\sqrt{n_\ell /V_{\tau n_\ell}^*}$.
	\begin{align*} \sqrt{\frac{n_\ell}{V_{\tau n_\ell}^*}}( \widehat{V}_{\tau n_\ell}^*-V_{\tau n_\ell}^*)  
		&= \frac{(e_1'\Omega_{n_\ell}^{1/2}Z_{n_\ell})^2}{\sqrt{n_\ell V_{\tau n_\ell}^*}} +  2 \zeta_{n_\ell} e_1'(\Omega_{n_\ell}^{1/2}Z_{n_\ell})+ \sqrt{V_{\tau n_\ell}^*}(e_2'\Omega_{n_\ell}^{1/2}Z_{n_\ell}) + o_p(1)\\
		&=  2 \zeta^* (e_1'\Omega^{1/2}Z_{n_\ell})+ \sqrt{V_{\tau}^*}(e_2'\Omega^{1/2}Z_{n_\ell}) + o_p(1). 
	\end{align*}
	The limiting distribution is normal. For convenience, we write this as $ \sqrt{n_\ell/V_{\tau n_\ell}^*}( \widehat{V}_{\tau n_\ell}^*-V_{\tau n_\ell}^*) = \sigma(\zeta^*)\tilde{Z}_{n_\ell} + o_p(1)$, where $\tilde{Z}_{n_\ell} \sim \mathcal{N}(0,1)$ and $\sigma(\zeta)^2 := \Omega_{11} + V_{\tau}^*\Omega_{22} + \zeta \sqrt{V_\tau^*} \Omega_{12}$ for $\zeta \in \{1,-1\}$. Since the norm of $[1,\zeta\sqrt{V_{\tau}^*}]'$ is larger than one, it follows that $\sigma(\zeta)^2 \ge \lambda_{\min}$, where $\lambda_{\min}$ is the smallest eigenvalue of $\Omega$. In this case, the limiting CDF is $H_{\Omega,V _{\tau}^*,\zeta}(\tilde{v}) := \Phi(\tilde{v}/\sigma(\zeta))$ where $\Phi(\cdot)$ is the CDF of a standard normal. Let $z_{\alpha} = \Phi^{-1}(\alpha)$ denote the $\alpha-$quantile, and $\phi(\cdot)$ the marginal of a standard normal.
	\begin{align}
		\label{eq:coverage_mildheterogeneity}
		\begin{split}		
			\mathbb{P}_{\gamma_{n_\ell}}(\widehat{R}_{n_\ell,\zeta} \in [0,1-\alpha]) &= \mathbb{P}\left(-\alpha/2 \le \Phi\left(\frac{\sigma(\zeta^*)}{\sigma(\zeta)}\tilde{Z}_{n_\ell}\right) \le 1-\alpha/2\right) + o(1) \\
			&= \underbrace{\Phi\left( \frac{\sigma(\zeta^*)}{\sigma(\zeta)}z_{1-\alpha/2}\right) - \Phi\left( \frac{\sigma(\zeta^*)}{\sigma(\zeta)}z_{-\alpha/2}\right)}_{\kappa(\sigma(\zeta),\sigma(\zeta^*))} + o(1).
		\end{split}
	\end{align}
	To obtain the lower bound,$
	\mathbb{P}_{\gamma_{n_\ell}}\left( V_{\tau n_\ell}^* \in \widehat{CI}_{\alpha n k} \right) \ge \mathbb{P}_{\gamma_{n_\ell}}\left( \widehat{R}_{n_\ell,\zeta^*} \in [0,1-\alpha]  \right)  = (1-\alpha) + o(1)$. To obtain the upper bound, I write down a Taylor expansion. There is a $\tilde{\sigma}\ge \sqrt{\lambda_{\min}}$ between $\sigma(\zeta^*)$ and $\sigma(\zeta)$ such that
	\begin{align*} \Vert \kappa(\sigma_\zeta,\sigma_{\zeta^*}) - (1-\alpha) \Vert 
		&\le \left\Vert \left[\phi\left(\frac{\tilde{\sigma}z_{1-\alpha/2}}{\sigma(\zeta)}\right)z_{1-\alpha/2} - \phi\left(\frac{\tilde{\sigma}z_{-\alpha/2}}{\sigma(\zeta)}\right)z_{-\alpha/2} \right]\frac{[\sigma(\zeta^*) - \sigma(\zeta)]}{\sigma(\zeta)} \right\Vert \\
		&\le \frac{|z_{-\alpha/2}|}{\sqrt{2\pi \lambda_{min}}}\Vert  \sigma(\zeta^*) -\sigma(\zeta) \Vert .
	\end{align*}
	By another Taylor expansion, $\Vert \sigma(\zeta^*)-\sigma(\zeta)\Vert  \le \frac{1}{2 \sqrt{\lambda_{\min}}}\Vert \sigma(\zeta^*)^2 - \sigma(\zeta)^2\Vert $. If in addition, $\Vert \sqrt{V_{\tau}*}\Omega_{12}\Vert  \le t$, then $\Vert  \sigma(\zeta^*)-\sigma(\zeta)\Vert  \le t$ and we can define a non-negative function $\tilde{\alpha}(t) = t \Vert z_{-\alpha/2}\Vert /(2\lambda_{\min}\sqrt{2\pi})$, which satisfies $\lim_{t \to 0}\tilde{\alpha}(t) = 0$.	Therefore we have bounded the coverage over an exhaustive class of subsequences of distributions. This satisfies Assumption $B$ in \cite{andrews2020generic}. By their Corollary 2.1,
	$$ 1-\alpha \le \underset{n \to \infty}{\lim\inf}\inf_{\gamma \in \Gamma} \mathbb{P}_{\gamma}\left( V_{\tau}^*(\gamma) \in \widehat{CI}_{\alpha n} \right) = \underset{n \to \infty}{\lim\sup}\sup_{\gamma \in \Gamma} \mathbb{P}_{\gamma}\left( V_{\tau }^*(\gamma) \in \widehat{CI}_{\alpha n} \right) \le 1-\alpha + \tilde{\alpha}(t).$$

\end{proof}	

\begin{lem}
	\label{lem:designmatrix_dml}
	 Let $(S_{-k}(X_i),M_{-k}(X_i),W_i)$ and $\lambda(X_i)$ be the set of regressors and weights, respectively, that were defined in \eqref{eq:dml_defn_auxiliaryfunctions}. Define $Q_{ww,nk} := \mathbb{E}[\lambda(X_i)W_iW_i' \mid \mathcal{I}_{-nk}]$ and let $\Pi_{nk}$ be a $4 \times 4$ diagonal matrix with entries $(1,1,1,V_{xnk}^{-1/2})$. If $(X_i,D_i)$ are independent of the data in $\mathcal{I}_{-nk}$, then $\Pi_{nk}'Q_{ww,nk}\Pi_{nk}$ has the form in \eqref{eq:designmatrix_tildew_normalized}.

\end{lem}

\begin{proof}[Proof of Lemma \ref{lem:designmatrix_dml}]
	Let $V_i = (D_i-p(X_i))$. By definition, $\lambda(X_i)W_iW_i'$ is 
	\begin{align}	 
	 &\lambda(X_i)\begin{pmatrix}
		1 & M_{-k}(X_i) & V_i  & V_iS_{-k}(X_i) \\
		M_{-k}(X_i) & M_{-k}(X_i)^2 & M_{-k}(X_i)V_i &  M_{-k}(X_i) V_iS_{-k}(X_i) \\
		V_i & M_{-k}(X_i) V_i &   V_i^2 &  V_i^2S_{-k}(X_i) \\
		V_iS_{-k}(X_i) & M_{-k}(X_i) V_iS_{-k}(X_i) & V_i^2S_{-k}(X_i)  & V_i^2S_{-k}(X_i)^2
	\end{pmatrix}.
	\label{eq:designmatrix_nonexpect_proof}
	\end{align}
	Since $(X_i,D_i)$ are independent of the data in $\mathcal{I}_{-nk}$, then $\mathbb{E}[V_i \mid X_i = x, \mathcal{I}_{-nk}]$ does not depend on $\mathcal{I}_{-nk}$ and is equal to $\mathbb{E}[V_i \mid X_i = x] = \mathbb{E}[D_i \mid X=x] - p(x) = 0$. By a similar reasoning, $\mathbb{E}[V_i^2 \mid X_i = x, \mathcal{I}_{-nk}] = p(X_i)(1-p(X_i)) = \lambda(X_i)^{-1}$. Using both conditional moment results, we can show that $\mathbb{E}[\lambda(X_i)W_iW_i' \mid X_i, \mathcal{I}_{-nk}]$ is equal to
 \begin{align*}
		Q_{ww,nk} = \mathbb{E}\left[\left.  \begin{matrix} 	\lambda(X_i)  & 	\lambda(X_i) M_{-k}(X_i) & 0 & 0 \\
				\lambda(X_i)  M_{-k}(X_i) & 	\lambda(X_i)  M_{-k}(X_i)^2 & 0 & 0 \\
			0 & 0 & 1 & S_{-k}(X_i) \\ 
			0 & 0 & S_{-k}(X_i) & S_{-k}(X_i)^2 \\ \end{matrix} \quad \right| \mathcal{I}_{-nk} \right].
	\end{align*}
	
	We substitute $\mathbb{E}[S_{-k}(X_i) \mid \mathcal{I}_{-nk}] = 0$ and $\mathbb{E}[S_{-k}(X_i)^2 \mid \mathcal{I}_{-nk}] = V_{xnk}$. Finally, $\Pi_{nk}'Q_{ww,nk}\Pi_{nk}$ only normalizes the lower right corner to 1.
	
\end{proof}

\begin{lem}[Convergence of cross-fitted sums]
	\label{lem:convergence_standardizedsums}
	Consider a sequence of distributions $\{\gamma_n\}_{n=1}^\infty$ over a collection of random matrices  $\{Z_{i1},\ldots,Z_{iL}\}_{i \in \mathcal{I}_{nk}}$ where $L$ is a finite constant, $Z_{i\ell} \in \mathbb{R}^{M} \times \mathbb{R}^B$. Define $\widehat{\zeta}_{nk} := \sum_{i \in \mathcal{I}_{-nk}}\sum_{\ell=1}^L Z_{i\ell}$. If for all $\ell \in \{1,\ldots,L\}$, (i) the observations are i.i.d. conditional on $\mathcal{I}_{-nk}$, (ii) $\mathbb{E}_{\gamma_n}[Z_{i\ell} \mid \mathcal{I}_{-nk}] = \uline{0}_{M\times B}$, and (iii) $\mathbb{E}_{\gamma_n}[\Vert Z_{i \ell} \Vert^2]$ has a uniform upper bound and $\gamma_n$, then as $n_k \to \infty$, (a) $n_k^{-1/2}\widehat{\zeta}_{nk} = O_p(1)$, (b) If in addition, $\mathbb{E}_{\gamma_n}[\Vert Z_{i\ell}\Vert^2] = o(1)$, then $n_k^{-1/2}\widehat{\zeta}_{nk} = o_p(1)$, (c) Now suppose (ii) and (iii) do not necessarily hold, but instead $n_k^{r}\mathbb{E}_{\gamma_n}[\Vert Z_{i\ell}\Vert] = o(1)$ for all $\ell$ and some $r > 0$, then $n_k^{r-1}\widehat{\zeta}_{nk} = o_p(1)$.
\end{lem}

\begin{proof}
	Since $M,L,B$ are all finite, it suffices to consider $\widehat{\zeta}_{nk\ell mb} := \sum_{i \in \mathcal{I}_{-nk}} Z_{i\ell mb}$, where $Z_{i\ell mb}$ is the coordinate $(m,b)$ of $Z_{i\ell}$. 	By the law of iterated expectations, $R_{nk\ell mb} := \mathbb{E}_{\gamma_n}[(\widehat{\zeta}_{nk\ell mb}-\mathbb{E}_{\gamma_n}[\widehat{\zeta}_{nk\ell mb} \mid \mathcal{I}_{-nk}])^2]$ is equal to $\mathbb{E}_{\gamma_n}[\mathbb{E}_{\gamma_n}[(\widehat{\zeta}_{nk\ell mb}-\mathbb{E}_{\gamma_n}[\widehat{\zeta}_{nk\ell mb} \mid \mathcal{I}_{-nk}])^2 \mid \mathcal{I}_{-nk}]]$. Substituting the definition of conditional variance, 	
	\begin{equation}
	\label{eq:decomp_crossfitted}
	 R_{nk\ell mb} = \mathbb{E}_{\gamma_n}\left[\mathbb{V}_{\gamma_n}[\widehat{\zeta}_{nk \ell mb} \mid \mathcal{I}_{-nk}]+\mathbb{E}_{\gamma_n}[\widehat{\zeta}_{nk\ell mb} \mid \mathcal{I}_{-nk}]^2\right].
	\end{equation}	
	The first term of \eqref{eq:decomp_crossfitted} is $O(n_k^{-1})$. By Assumption (i), $\widehat{\zeta}_{nk\ell mb}$ is a sum of $n_k$ variables that are i.i.d. conditional on $\mathcal{I}_{-nk}$, and hence $\mathbb{V}_{\gamma_n}[\widehat{\zeta}_{nk\ell mb} \mid \mathcal{I}_{-nk}]  = n_k \mathbb{V}_{\gamma_n}\left( Z_{i\ell mb} \mid \mathcal{I}_{nk}\right) = n_k \mathbb{E}_{\gamma_n}\left[ Z_{i\ell mb}^2 \mid \mathcal{I}_{nk}\right]$. By the law of iterated expectations, $\mathbb{E}_{\gamma_n}\left[\mathbb{V}_{\gamma_n}[\widehat{\zeta}_{nk\ell mb} \mid \mathcal{I}_{-nk}]\right] = n_k\mathbb{E}_{\gamma_n}\left[Z_{i\ell mb}^2\right]$. The second term of \eqref{eq:decomp_crossfitted} is zero by Assumption (ii).
	
	To prove (a), we apply Chebyshev's inequality $\mathbb{P}(n_k^{-1}\widehat{\zeta}_{nk\ell mb}> t) \le \mathbb{E}_{\gamma_n}[Z_{i\ell mb}^2]/t^2$ for some $t>0$. This shows that $n_k^{-1/2}\widehat{\zeta}_{nk\ell mb} = O_p(1)$. To prove part (b), we use the condition that $\mathbb{E}_{\gamma_n}[Z_{i\ell mb}^2] = o(1)$ to show that $n_k^{-1/2}\widehat{\zeta}_{nk\ell mb} = o_p(1)$. To prove (c), we apply the triangle inequality, $\Vert n_k^{r-1}\widehat{\zeta}_{nk\ell mb} \Vert \le n_k^{r-1} \sum_{i \in \mathcal{I}_{-nk}}\Vert Z_{i\ell mb}\Vert $. By Markov's inequality, $\mathbb{P}(\Vert n_k^{r-1}\widehat{\zeta}_{nk\ell mb} \Vert > t) \le n_k^r\mathbb{E}[\Vert Z_{i\ell mb}\Vert ]/t = o(1)$, hence $n_k^{r-1}\widehat{\zeta}_{nk\ell mb} = o_p(1)$.
	
\end{proof}

\begin{lem}[Derivatives of Influence Function]
	\label{lem:derivative_influence_function}
	Let $\varphi_i(\eta) := \varphi(Y_i,D_i,X_i,\eta)$, for $i \in \mathcal{I}_{nk}$, and $U_i = Y_i - \mathbb{E}[Y_i \mid D_i,X_i]$. Suppose that $\{Y_i,D_i,X_i\}_{i \in \mathcal{I}_{nk}}$ is independent of $\{Y_i,D_i,X_i\}_{i \in \mathcal{I}_{-nk}}$ for all $k \in \{1,\ldots,K\}$, and consider a set of $\eta \in \mathcal{T}$ where the propensity score is bounded in $[\delta,1-\delta]$ for $\delta\in (0,1/2]$. Then (a) $\mathbb{E}\left[\frac{\partial \varphi_i(\eta)}{\partial(\eta_1,\eta_2,\eta_4)'}\mid X= x,\mathcal{I}_{-nk}\right] = [0,0,-2(\tau(x)-\tau_{av})]$ almost surely, (b) $\mathbb{E}\left[ \left\Vert \partial \varphi_i(\eta)/\partial \eta_\ell\right\Vert^4 \right]^{1/4} \le (8/\delta)(\mathbb{E}[\Vert \eta(X_i)\Vert^4]^{1/4} + \mathbb{E}\left[ \left\Vert U_i \right\Vert^4 \right]^{1/4})$, and (c) $\left\Vert \frac{\partial^2 \varphi_i(\tilde{\eta})}{\partial(\eta_1,\eta_2,\eta_4)\partial(\eta_1,\eta_2,\eta_4)'} \right\Vert \le 18 \times \sqrt{3}/\delta$ almost surely, for $\tilde{\eta} \in \mathcal{T}$.

\end{lem}

\begin{proof}
	
	Part (a): The jacobian of $\varphi_i(\eta)$, $\frac{\partial \varphi_i(\eta)}{\partial(\eta_1,\eta_2,\eta_4)}$ is
	\begin{align*}
		&\begin{pmatrix} 
			2(\tau(X_i)-\tau_{av})\left[1-\frac{D_i}{p(X_i)} \right]+
			2 \left[ \frac{D_i(Y_i - \mu_0(X_i)-D_i\tau(X_i))}{p(X_i)}-\frac{(1-D_i)(Y_i - \mu_0(X_i))}{1-p(X_i)} \right] \\
			2(\tau(X_i)-\tau_{av})\left[ - \frac{D_i}{p(X_i)} + \frac{1-D_i}{1-p(X_i)} \right] \\
			-2(\tau(X_i)-\tau_{av}) -2\left[ \frac{D_i(Y_i - \mu_0(X_i)-D_i\tau(X_i))}{p(X_i)}-\frac{(1-D_i)(Y_i - \mu_0(X_i))}{1-p(X_i)} \right]
		\end{pmatrix}
	\end{align*}
	Part (a): Substituting $\mathbb{E}[D_i \mid X_i = x,\mathcal{I}_{-nl}] = p(x)$ and $\mathbb{E}[D_iY_i \mid X=x,\mathcal{I}_{-nk}] = p(x)(\mu_0(x)+\tau(x))$, then $\mathbb{E}\left[\partial \varphi_i(\eta)/\partial(\eta_1,\eta_2,\eta_4)'\mid X= x,\mathcal{I}_{-nk}\right] = [0,0,-2(\tau(x)-\tau)]'$. \\

	Part (b): By construction, $U_i = Y_i - \mu_0(X_i)-D_i\tau(X_i)$. Then the jacobian simplifies to $\partial \varphi_i(\eta)/\partial(\eta_1,\eta_2,\eta_4)' = G_{1i} + G_{2i}$, where $G_{1i} := 2(\tau(X_i)-\tau_{av})\times [(1-D_i/p(X_i)),(-D_i/p(X_i)+(1-D)/(1-p(X_i))),-1]$ and $G_{2i} := 2(D_i/p(X_i)-(1-D_i)/(1-p(X_i)))\times [U_i,0,-U_i]$. Since $p(x) \in [\delta,1-\delta]$, $\Vert G_{1i} \Vert \le (4/\delta)\Vert \tau(X_i) - \tau_{av}\Vert \le (8/\delta)\Vert \eta(x) \Vert $ and $\Vert G_{2i} \Vert \le (8/\delta) \Vert U_i \Vert$. Therefore, by the triangle inequality, $\mathbb{E}\left[ \left\Vert \partial \varphi_i(\eta)/\partial \eta_\ell\right\Vert^4 \right]^{1/4} \le (8/\delta)(\mathbb{E}[\Vert \eta(X_i)\Vert^4]^{1/4} + \mathbb{E}\left[ \left\Vert U_i \right\Vert^4\right])$.

	Part (c): The hessian of $\varphi_i(\tilde{\eta})$, which I denote by $H(\tilde{\eta})$, is symmetric and
	\begin{align*}
		&\frac{\partial^2 \varphi_i(\tilde{\eta})}{\partial(\eta_1,\eta_2,\eta_4)\partial(\eta_1,\eta_2,\eta_4)'} =
		\begin{pmatrix} 
			2\left[1-\frac{D_i}{\tilde{p}(X_i)} \right] -2 \frac{D_i}{\tilde{p}(X_i)} & \cdot & \cdot \\
			2\left[ - \frac{D_i}{\tilde{p}(X_i)} + \frac{1-D_i}{1-\tilde{p}(X_i)} \right] & 0 & \cdot  \\
			2\left[ 1 - \frac{D_i}{\tilde{p}(X_i)} \right] & -2 \left[ \frac{D_i}{\tilde{p}(X_i)}-\frac{(1-D_i)}{1-\tilde{p}(X_i)} \right] & -  2
		\end{pmatrix}
	\end{align*}
	Since $\tilde{p}(x) \in [\delta,1-\delta]$, then $ \Vert D_i/\tilde{p}(X_i) \Vert \le 1/\delta$, $\Vert 1-D_i/\tilde{p}(X_i) \Vert \le (1+1/\delta) \le 2/\delta$, and $\Vert D_i/\tilde{p}(X_i) - (1-D_i)/(1-\tilde{p}(X_i)) \Vert \le 2/\delta$. This means that each of the entries  of $H(\tilde{\eta})$ is bounded by $6/\delta$. By Lemma \ref{lem:boundoperatornorm}, $\Vert H(\tilde{\eta}) \Vert \le 3 \times (6/\delta) \times \sqrt{3} = 18 \times \sqrt{3}/\delta$.

\end{proof}

\begin{lem}[Decomposition of Nuisance Functions]
	\label{lem:decomposition_nuisance}
	Define $\widehat{\eta}_{-k,\widehat{\theta}_{nk}}(x)$ as in \eqref{eq:pseudonuisancefunctions_dml}, $\theta_n := (0,1,\tau_{n,av},1)$, $(\widehat{\lambda}_{nk}-\lambda_{nk}):= [(\widehat{\theta}_{nk}-\theta_{n}),\widehat{\tau}_{nk,av}(\widehat{\theta}_{nk}-\theta_{n}),(\widehat{\tau}_{nk,av}-\tau_{n,av})]'$, and let $\{e_j\}_{j=1}^4$ be $4 \times 1$ vectors with 1 in the $j^{th}$ coordinate and zero otherwise. Then there exist $(x,\mathcal{I}_{-nk})-$measurable matrices $\Psi_{-nk}(x)$, $\Delta_{-nk}(x)$, such that
	\begin{equation}
	\label{eq:decomp_estimated_nuisance_lemstatement} \widehat{\eta}_{-k,\widehat{\theta}_{nk}}(x) - \eta(x) = \Psi_{-nk}(x)(\widehat{\lambda}_{nk}-\lambda_{nk}) + \Delta_{-nk}(x),
	\end{equation}
	and for some constant $C< \infty$, (a) $e_3'[\widehat{\eta}_{-k,\widehat{\theta}_{nk}}(x) - \eta(x)] = 0$, (b) $e_3'\Delta_{-nk}(x) = e_4'\Delta_{-nk}(x) = 0$,  (c)  $e_4'\Psi_{-nk} = c$, for $c \in \mathbb{R}^{9}$, (d) $\Vert \Psi_{-nk}(x) \Vert \le C \times \left[1+ 2\Vert \widehat{\eta}_{-k}(x) \Vert\right]$ a.s., (e) $\Vert \Delta_{-nk}(x) \Vert \le C \times \Vert \widehat{\eta}_{-k}(x) - \eta(x) \Vert$ a.s.
	
\end{lem}

\begin{proof}[Proof of Lemma \ref{lem:decomposition_nuisance}]
	Define $\widehat{B}(x) := [e_1 (\widehat{W}(x,1) -\widehat{W}(x,0))' +e_2 \widehat{W}(x,0)']$.
	\begin{align*}
		\widehat{\eta}_{-k,\widehat{\theta}_{nk}}(x)' &=  \widehat{B}(x) + e_3 p_n(x) + e_4 \widehat{\tau}_{nk,av}, \\ 
		\widehat{\eta}_{-k,\theta_{n}}(x) &=  \widehat{B}(x)\theta_{n} + e_3 p_n(x) + e_4 \widehat{\tau}_{nk,av}, \\ 
		\eta(x) &= e_1 \tau_n(x) + e_2 \mu_{0n}(x) + e_3 p_n(x) + e_4 \tau_{n,av}.
	\end{align*}
	The estimation error can be decomposed as
	\begin{equation} \widehat{\eta}_{-k,\widehat{\theta}_{nk}}(x) - \eta(x) = \widehat{B}(x)(\widehat{\theta}_{nk}-\theta_{n}) + e_4[\widehat{\tau}_{nk,av}-\tau_{n,av}] + [\widehat{B}(x)'\theta_{n} - e_1\tau_n(x)+e_2\mu_{0n}(x)].
	\label{eq:decomp_etahat_twodecomps}
	\end{equation}
	By definition $\widehat{B}(x) = \Psi_{1,-nk}(x)+\widehat{\tau}_{nk,av}\Psi_{2,-nk}(x)$, with auxiliary matrices $\Psi_{1,-nk} := e_1[0,0,1,\widehat{\tau}_{-k}(x)]+e_2 [1,M_{-k}(X), -p_n(x), -p_n(x)\widehat{\tau}_{-k}(x)] $ and $\Psi_{2,-nk} := e_1[0,0,0,-1] + e_2 [0,0,0,p_n(x)]$. Substituting the  parameter $\theta_n := (0,1,\tau_{n,av},1)$ and grouping common terms, $\widehat{B}(x)'\theta_{n} -e_1 \tau_n(x) - e_2\mu_{0n}(x)= e_1 [(\widehat{\tau}_{-k}(x)-\tau_n(x)) + (\tau_{n,av}-\widehat{\tau}_{nk,av})]+ e_2[M_{-k}(x) -\mu_{0n}(x) -p_n(x)\widehat{\tau}_{-k}(x)+p_n(x)(\widehat{\tau}_{nk,av}-\tau_{n,av})]$. We can simplify the second term of this expression by substituting $M_{-k}(x) = \widehat{\mu}_{0,-k}(x) + p_n(x)\widehat{\tau}_{-k}(x)$, which produces $e_2[(\widehat{\mu}_{0n}(x) -\mu_{0n}(x)) +p_n(x) (\widehat{\tau}_{nk,av}-\tau_{n,av})]$. Consequently, the second and third terms of \eqref{eq:decomp_etahat_twodecomps} can be written as $\Psi_{3,-nk}(x)(\widehat{\tau}_{nk,av}-\tau_{n,av}) + \Delta_{-nk}(x)$, where $\Psi_{3,-nk}(x) := -e_1 +e_2p_n(x) + e_4$ and $\Delta_{-nk}(x) := e_1(\tau_{-k}(x)-\tau_n(x)) + e_2 (\widehat{\mu}_{-k}(x)-\mu_{0n}(x))$. 
	
	Define $\Psi_{-nk}(x) := [\Psi_{1,-nl}(x),\Psi_{2,-nk}(x),\Psi_{3,-nk}(x)]$ and the parameter error as $(\widehat{\lambda}_{nk}-\lambda_{nk}) := [(\widehat{\theta}_{nk}-\theta_{n})',\widehat{\tau}_{nk,av}(\widehat{\theta}_{nk}-\theta_{n})',(\widehat{\tau}_{nk,av}-\tau_{n,av})]'$. Combining the results,
	$$ \widehat{\eta}_{-k,\widehat{\theta}_{nk}}(x) - \eta(x) = \Psi_{-nk}(x)(\widehat{\lambda}_{nk}-\lambda_{nk}) + \Delta_{-nk}(x). $$
	Measurability with respect to $(x,\mathcal{I}_{-nk})$ can be verified by inspection.  Property (a) follows from the fact that the propensity score is known, and (b) because $\Psi_{-nk}(x)$ and $\Delta_{-nk}(x)$ depend on vectors $e_1,e_2$, which are orthogonal to $e_3,e_4$. Property (c) follows by the fact that $e_4'\Psi_{-nk} = [\uline{0}_{1 \times 8},1]$. To prove part (d), we apply Lemma \ref{lem:boundoperatornorm} to show that $\Psi_{-nk}(x)$ is bounded by $9\sqrt{4}$ times the largest absolute value of the matrix. Since $p_n(x) \le 1$, then the largest value is bounded by $1+\Vert \widehat{\mu}_{0,-k}(x) \Vert + \Vert \widehat{\tau}_{-k}(x) \Vert$ which is less than $1+ 2\Vert \widehat{\eta}_{-k}(x) \Vert$. This means that $\Vert \Psi_{-nk}(x)\Vert \le 12 \times \sqrt{4} \times \left[1+ 2\Vert \widehat{\eta}_{-k}(x) \Vert\right]$. To prove, part (e) we once again apply Lemma \ref{lem:boundoperatornorm}. The quantity $\Delta_{-nk}(x)$ is a $4 \times 1$ vector, whose individual entries are bounded by $\Vert \widehat{\mu}_{0,-k}-\mu_{0n}(x) \Vert + \Vert \widehat{\tau}_{-k}(x)-\tau_n(x) \Vert$, which is weakly less than $2 \Vert \widehat{\eta}_{-k}(x) - \eta(x) \Vert$.
	
\end{proof}

\begin{lem}[Convergence of regression parameters]
	\label{lem:convergence_regression_parameters}
	Consider a sequence of distributions $\{\gamma_n\}_{n=1}^{\infty}$ that satisfy Assumptions \ref{assump:unconfound_sutva},  \ref{assump:momentbounds_dml}, \ref{assump:boundsquantities}, \ref{assump:randomsampling}, \ref{assump:convergencenuisance}, and \ref{assump:regularityconditions}, and that $V_{\tau n} \to V_{\tau} > 0$. Define oracle regressors, $W_i^* :=[1,\mu_{0n}(X_i) + p_n(X_i)\tau_n(X_i),(D_i-p_n(x_i)), (D_i-p_n(X_i))(\tau_n(X_i)-\tau_{av,n})]'$ and $\theta_n := \mathbb{E}_{\gamma_n}[\lambda_n(X_i)W_i^*W_i^{*'}]^{-1}\mathbb{E}_{\gamma_n}[\lambda_n(X_i)W_i^*Y_i]$.  Then (a)  $\widehat{\tau}_{nk,av}-\tau_{n,av} = o_p(n_k^{-1/4})$, (b) $\theta_n = (0,1,\tau_{n,av},1)'$, and (c) $\widehat{\theta}_{nk} - \theta_{n} = o_p(n_k^{-1/4})$.
\end{lem}

\begin{proof}
	Part (a): Decompose, $n_k^{1/4}(\widehat{\tau}_{nk,av} - \tau_{av,n}) =  n_k^{-1/4}\left[ \frac{1}{\sqrt{n_k}}\sum_{i \in \mathcal{I}_{nk}} \tau_n(X_i)-\tau_{n,av}\right] + n_k^{1/4-1}\sum_{i \in \mathcal{I}_{nk}} [\widehat{\tau}_{-k}(X_i) - \tau_{n}(X_i)]$. The first term is a centered random variable that is $o_p(1)$. By Assumption \ref{assump:regularityconditions}.(iv), and the Cauchy-Schwartz inequality, $\mathbb{E}_{\gamma_n}[n_k^{1/4}\Vert \widehat{\tau}_{-k}(X_i) - \tau_n(X_i)] \Vert \le \mathbb{E}_{\gamma_n}[n_k^{1/2}\Vert \widehat{\eta}_{-k}(X_i) -\eta(X_i)\Vert^2]^{1/2} \to o(1)$. By applying Lemma \ref{lem:convergence_standardizedsums}.(c),  $n_k^{1/4}(\widehat{\tau}_{nk,av} - \tau_{av,n}) = o_p(1)$.
	
	Part (b): For given $\{X_i = x,D_i = d\}$, $W_i^{*'}(0,1,\tau_{n_av},1)' = \mu_{d}(x)$. Therefore, by applying Lemma \ref{lem:projection_weak}, $\theta_n = (0,1,\tau_{n,av},1)'$.

	Parts (c):  Define $\widehat{Q}_{ww} :=  \frac{1}{n_k}\sum_{i \in \mathcal{I}_{nk}} \lambda_n(X_i)\widehat{W}_{i}\widehat{W}_i'$, $Q_{ww}:= \mathbb{E}_{\gamma_n}[\lambda_n(X_i) W_i^*W_i^{*'}]$, $M_n(x) = \mu_{0n}(x) + p_n(x)\tau_n(x)$. Following similar derivations to Lemma \ref{lem:designmatrix_dml},
	\begin{equation}
	\label{eq:designmatrix_ww}
	Q_{ww} =  \begin{pmatrix} \mathbb{E}_{\gamma_n}[\lambda(X_i)] & \mathbb{E}_{\gamma_n}[\lambda(X_i)(\mu_{n}(X_i))] & 0 & 0 \\
	\mathbb{E}_{\gamma_n}[\lambda(X_i)(\mu_{n}(X_i))] & \mathbb{E}_{\gamma_n}[\lambda(X_i)(\mu_{n}(X_i))^2] & 0 & 0 \\
	0 & 0 & 1 & 0\\ 
	0 & 0 & 0 & V_{\tau n}\\ \end{pmatrix}.
	\end{equation}
	The upper left block has bounded eigenvalues by Assumption \ref{assump:momentbounds_dml}.(i) and $V_{\tau n}$ is asymptotically bounded. Therefore $Q_{ww}$ is positive definite with bounded eigenvalues. Furthermore, $\widehat{Q}_{ww} - Q_{ww}$ can be decomposed as:	
	\begin{equation}
	\left[ \frac{1}{n_k}\sum_{i \in \mathcal{I}_{nk}}\lambda_n(X_i)W_i^*W_i^{*'} - \mathbb{E}_{\gamma_n}[\lambda_n(X_i)W_i^*W_i^{*'}] \right] + \left[ \frac{1}{n_k}\sum_{i \in \mathcal{I}_{nk}} \lambda_n(X_i)(\widehat{W_i}\widehat{W}_i' -  \widehat{W_i}^*\widehat{W}_i^{*'}) \right].
	\label{eq:decomp_Wtrue_times_Wtrue_proof}
	\end{equation}
	The first term of \eqref{eq:decomp_Wtrue_times_Wtrue_proof} is an average of mean-zero random variables and bounded variance, then it is $O_p(n_k^{-1/2}) = o_p(n_k^{-1/4})$. To prove that it has bounded variance, apply Lemma \ref{lem:boundoperatornorm}, then $\Vert W_i^*\Vert \le \sqrt{4}(1+ \Vert \mu_{0n}(X_i) \Vert + \Vert \tau_n(X_i) \Vert) \le \sqrt{4}(1+ 2\Vert \eta(X_i) \Vert)$. Since $p_n(x) \in [\delta,1-\delta]$ and $\lambda_n(X_i) = [p_n(X_i)(1-p_n(X_i))]^{-1}$, then $\mathbb{E}_{\gamma_n}[\Vert \lambda_n(X_i) W_i^*\Vert^2]^{1/2} \le (1/\delta^2) \mathbb{E}_{\gamma_n}[\Vert W_i^*\Vert^4]^{1/4} \le (\sqrt{4}/\delta^2)(1+2 \mathbb{E}_{\gamma_n}[\Vert \eta(X_i) \Vert^4]^{1/4})$, which is bounded by Assumption \ref{assump:regularityconditions}.(ii).
	
	To bound the second term of \eqref{eq:decomp_Wtrue_times_Wtrue_proof}, we apply the triangle inequality, $\Vert\lambda_n(X_i) (\widehat{W}_i\widehat{W}_i' - W_i^*W_i^{*'})\Vert \le (1 / \delta^2)\widehat{\phi}_i$, where $\widehat{\phi}_i := 2 \Vert W_i^{*'}\Vert \  \Vert \widehat{\zeta}_i \Vert + \Vert \widehat{\zeta}_i \Vert^2$ and $\widehat{\zeta}_i := \widehat{W}_{i} - W_i^*$.  Our goal is to show that $\frac{1}{n_k}\sum_{i \in \mathcal{I}_{-nk}} \widehat{\phi}_i  = o_p(n_k^{-1/4})$. Let $e_\ell$ be a $4 \times 1$ vector with one in the $\ell^{th}$ entry and zero otherwise. We can further decompose  $\widehat{\zeta}_i =\widehat{\Delta}_{nk} \widehat{\zeta}_{-nk,1}(X_i) +  \widehat{\zeta}_{-nk,2}(X_i) $, into components that map into our assumptions: $\widehat{\Delta}_{nk} := (\widehat{\tau}_{nk,av}-\tau_{n,av})$, $ \widehat{\zeta}_{-nk,1}(X_i) := e_4 (D_i-p_n(X_i))$ and $ \widehat{\zeta}_{-nk,2}(X_i) := e_1(M_{-k}(X_i)-M_n(X_i)) + e_2(D_i-p_n(X_i))(\widehat{\tau}_{-k}(X_i)-\tau_n(X_i))$. By construction, $\Vert \zeta_{-nk,1}(X_i) \Vert  \le 1$. Applying the triangle inequality and grouping terms, $\widehat{\phi}_i \le \Vert \widehat{\Delta}_{nk}\Vert ^2 + \Vert \widehat{\Delta}_{nk} \Vert \widehat{\phi}_{i1} + \widehat{\phi}_{i2}$, where $\widehat{\phi}_{i1} := ( 2 \Vert W_i^*\Vert + 2 \Vert \widehat{\zeta}_{-nk,2} \Vert)$, and $\widehat{\phi}_{i2} :=  (2 \Vert W_i^*\Vert \ \Vert \widehat{\zeta}_{-nk,2}(X_i) \Vert + \Vert \widehat{\zeta}_{-nk,2}(X_i) \Vert^2 )$. By Assumption \ref{assump:regularityconditions}.(iv), $ \mathbb{E}_{\gamma_n}[\Vert \widehat{\zeta}_{-nk,2}(X_i) \Vert^2]^{1/2} \le 3 \mathbb{E}_{\gamma_n}[\Vert \widehat{\eta}_{-k}(X_i) - \eta(X_i) \Vert^2]^{1/2} = o(n_k^{-1/4})$. Therefore, by the Cauchy-Schwarz inequality, $\mathbb{E}_{\gamma_n}[\widehat{\phi}_{i\ell}^2]^{1/2} = o(n_k^{-1/4})$ for $\ell \in \{1,2\}$. Then by Lemma \ref{lem:convergence_standardizedsums}.(c), $\left[ \frac{1}{n_k} \sum_{i \in \mathcal{I}_{-nk}} \widehat{\phi}_{i_\ell} \right] = o_p(n_k^{-1/4})$ for $\ell \in \{1,2\}$. Combining terms,
	$$\left\Vert \frac{n_{k}^{1/4}}{n_k}\sum_{i \in \mathcal{I}_{-nk}} \lambda_n(X_i)(\widehat{W_i}\widehat{W}_i' -  \widehat{W_i}^*\widehat{W}_i^{*'}) \right\Vert \le \frac{1}{\delta^2}\sum_{\ell = 0}^2 n_k^{1/4} \Vert \widehat{\Delta}_{nk}\Vert ^{\ell -2}\left[ \frac{1}{n_k} \sum_{i \in \mathcal{I}_{-nk}} \widehat{\phi}_{i_\ell} \right]= o_p(1).$$
	Define $\widehat{Q}_{wy} := \frac{1}{n_k} \sum_{i \in \mathcal{I}_{nk}} \lambda_n(X_i) \widehat{W}_{i}Y_i$ and $Q_{wy}:= \mathbb{E}_{\gamma_n}[\lambda_n(X_i)W_i^*Y_i]$. We can apply similar arguments as above to show that $\widehat{Q}_{wy}- Q_{wy} = o_p(n_k^{-1/4})$.
	
	Substituting the definition, $\theta_{n} := Q_{ww}^{-1}Q_{wy}$ and rearranging terms, $n_k^{1/4}(\widehat{\theta}_{nk} -\theta_{n}) = n_k^{1/4}(\widehat{Q}_{ww}^{-1}\widehat{Q}_{wy}-\theta_{nk}) = \widehat{Q}_{ww}^{-1} n_k^{1/4}(\widehat{Q}_{wy}-Q_{wy})-n_k^{1/4}(Q_{ww}^{-1}Q_{wy}-\widehat{Q}_{ww}^{-1}Q_{wy})$.
	The first term is $o_p(n_k^{-1/4})$. The second term can be rewritten as $n_k^{1/4}(Q_{ww}^{-1}-\widehat{Q}_{ww}^{-1})Q_{wy} = o_p(n_k^{-1/4})$. To prove this, note that $\Vert Q_{ww}^{-1}-\widehat{Q}_{ww}^{-1} \Vert =\Vert Q_{ww}^{-1}(\widehat{Q}_{ww}-Q_{ww})\widehat{Q}_{ww}^{-1}\Vert \le \Vert Q_{ww}^{-1}\Vert \ \Vert \widehat{Q}_{ww}-Q_{ww} \Vert \ \Vert \widehat{Q}_{ww}^{-1}\Vert$.  The right-hand side is $o_p(n_k^{-1/4})$ since $Q_{ww}$ has eigenvalues bounded away from zero and $\widehat{Q}_{ww}$ converges to its true value at rate $n_k^{1/4}$. Combining the results produces $n_k^{1/4}(\widehat{\theta}_{nk} -\theta_{n}) = o_p(n_k^{-1/4})$.
\end{proof}

\begin{lem}[Bound on Operator Norm]
	\label{lem:boundoperatornorm}
	Let $H$ be an $M \times L$ matrix and let $\Vert H \Vert = \sup_{\{z \in \mathbb{R}^L: \Vert z \Vert = 1\}} \Vert Hz \Vert$ be corresponding matrix operator norm. The absolute value of the individual entries of $H$ is bounded  a constant $C$. Then $\Vert H \Vert \le LC\sqrt{M}$. 
\end{lem}

\begin{proof}[Proof of Lemma \ref{lem:boundoperatornorm}]
	Let $H_{m}$ be the $m^{th}$ row and $H_{m\ell}$ be the $(m,\ell)$ entry. Then $\Vert H \Vert =   \sup_{\{z \in \mathbb{R}^L: \Vert z \Vert = 1\}}\sqrt{\sum_{m =1}^M \left[ H_{m}z\right]^2 }= \sup_{\{z \in \mathbb{R}^L: \Vert z \Vert = 1\}} \sqrt{\sum_{m = 1}^M \left[ \sum_{\ell =1}^L H_{m\ell}z_\ell\right]^2  }$.  Since $|z_\ell| \le 1$ and $\Vert H_{m\ell} \Vert \le C$, $\Vert H \Vert = \sup_{\{z \in \mathbb{R}^L: \Vert z \Vert = 1\}}  \sqrt{\sum_{m=1}^M \left[ \sum_{\ell=1}^L\left| H_{m\ell} \right| \ \left| z_\ell \right| \right]^2  } \le CL\sqrt{M}$.
\end{proof}

\end{appendices}

\end{document}